\documentclass[PGL, biber]{nowfnt} 

\catcode`\@=11
\let\cite\citet
\catcode`\@=12

\usepackage[utf8]{inputenc}

\usepackage{etoolbox}
\makeatletter
\patchcmd{\ttlh@hang}{\parindent\z@}{\parindent\z@\leavevmode}{}{}
\patchcmd{\ttlh@hang}{\noindent}{}{}{}
\makeatother

\usepackage{color}
\usepackage{amssymb}
\usepackage{tikz}
\usepackage{mathpartir}
\usepackage{amsmath}
\usepackage[misc]{ifsym}
\usetikzlibrary{shapes,shapes.geometric,arrows,automata,snakes}

\usepackage{listings}
\usepackage{mathpartir}
\usepackage{color,xcolor}
\usepackage{xspace}

\newcommand{\set}[1]{\{#1\}}
\newcommand{\spec}[1]{\set{#1}}

\newcommand{\ie}{\emph{i.e.}\xspace}

\newcommand{\eg}{\emph{e.g.}\xspace}

\newcommand{\wrt}{\emph{wrt.}\xspace}

\newcommand{\poplmark}{\textsc{POPLMark}\xspace}
\newcommand{\code}[1]{\lstinline{#1}}
\newcommand{\join}{\bullet}
\newcommand{\ssreflect}{SSReflect\xspace}

\newcommand{\eqdef}{\triangleq}

\makeatletter
\newcommand{\defeq}{\mathrel{\aban@defeq}}
\newcommand{\aban@defeq}{%
  \vbox{\offinterlineskip\check@mathfonts
    \ialign{\hfil##\hfil\cr
      \fontsize{\ssf@size}{\z@}\normalfont def\cr
      \noalign{\kern1\p@}
      $\m@th=$\cr
      \noalign{\kern-.5\fontdimen22\textfont2}
    }%
  }%
}
\makeatother

\newcommand{\reducedstrut}{\vrule width 0pt height .9\ht\strutbox depth .9\dp\strutbox\relax}
\newcommand{\codediff}[1]{%
  \begingroup
  \setlength{\fboxsep}{0pt}%
  \colorbox{orange!25}{\reducedstrut#1\/}%
  \endgroup
}

\def\checkmark{\tikz\fill[scale=0.4](0,.35) -- (.25,0) -- (1,.7) -- (.25,.15) -- cycle;}
\def\cross{$\mathbin{\tikz [x=1.4ex,y=1.4ex,line width=.2ex] \draw (0,0) -- (1,1) (0,1) -- (1,0);}$}

\definecolor{shadecolor}{gray}{1.00}
\definecolor{darkgray}{gray}{0.30}
\definecolor{violet}{rgb}{0.56, 0.0, 1.0}
\definecolor{forestgreen}{rgb}{0.13, 0.55, 0.13}

\lstdefinelanguage{Coq} {
mathescape=true,						
texcl=false,
morekeywords=[1]{
  Equations,
  Add,
  All,
  Arguments,
  Axiom,
  Bind,
  Canonical,
  Check,
  Close,
  CoFixpoint,
  CoInductive,
  Coercion,
  Contextual,
  Corollary,
  Defined,
  Definition,
  Delimit,
  End,
  Example,
  Export,
  Fact,
  Fixpoint,
  Goal,
  Graph,
  Hint,
  Hypotheses,
  Hypothesis,
  Implicit,
  Implicits,
  Import,
  Inductive,
  Lemma,
  Let,
  Local,
  Locate,
  Ltac,
  Maximal
  Module,
  Morphism,
  Next,
  Notation,
  Obligation,
  Open,
  Parameter,
  Parameters,
  Prenex,
  Print,
  Printing,
  Program,
  Projections,
  Proof,
  Proposition,
  Qed,
  Record,
  Relation,
  Remark,
  Require,
  Reserved,
  Resolve,
  Rewrite,
  Save,
  Scope,
  Search,
  Section,
  Show,
  Strict,
  Structure,
  Tactic,
  Theorem,
  Unset,
  Variable,
  Variables,
  View,
  inside,
  outside
},
morekeywords=[2]{
  as,
  cofix,
  else,
  end,
  exists,
  exists2,
  fix,
  for,
  forall,
  fun,
  if,
  in,
  is,
  let,
  match,
  nosimpl,
  of,
  return,
  struct,
  then,
  vfun,
  with
},
morekeywords=[3]{Type, Prop, Set, True, False},
morekeywords=[4]{
  after,
  apply,
  assert,
  auto,
  bool_congr,
  case,
  change,
  clear,
  compute,
  congr,
  cut,
  cutrewrite,
  destruct,
  elim,
  field,
  fold,
  generalize,
  have,
  heval, 
  hnf,
  induction,
  injection,
  intro,
  intros,
  intuition,
  inversion,
  left,
  loss,
  move,
  nat_congr,
  nat_norm,
  pattern,
  pose,
  refine,
  rename,
  replace,
  revert,
  rewrite,
  right,
  ring,
  set,
  simpl,
  split,
  suff,
  suffices,
  symmetry,
  transitivity,
  trivial,
  unfold,
  unlock,
  using,
  without,
  wlog,
  autorewrite
},        
morekeywords=[5]{
  assumption,
  by,
  contradiction,
  done,
  exact,
  lia,
  gappa,
  omega,
  reflexivity,
  romega,
  solve,
  tauto,
  discriminate,
  unsat
},
morecomment=[s]{(*}{*)},
morekeywords=[6]{do, first, try, idtac, repeat},
showstringspaces=false,
morestring=[b]",
tabsize=3,							
extendedchars=true,  		 		
sensitive=true, 
breaklines=false,
basicstyle=\footnotesize\ttfamily,
captionpos=b,							
columns=[l]fullflexible,
identifierstyle={\color{black}},
keywordstyle=[1]{\color{violet}},
keywordstyle=[2]{\color{forestgreen}},
keywordstyle=[3]{\color{forestgreen}},
keywordstyle=[4]{\color{blue}},
keywordstyle=[5]{\color{red}},
keywordstyle=[6]{\color{violet}},
stringstyle=,
commentstyle=\ttfamily\color{brown},
numberstyle=\tiny,
literate={\\/}{{$\lor~$}}1
         {/\\}{{$\land~$}}1
         {:->}{{$\mapsto~$\!}}1
         {\\->}{{$\mapsto~$\!}}1
         {<--}{{$\asgn~$}}1
         {\\in}{{$\in~$}}1
         {\\notin}{{$\notin~$}}1
         {++}{{$+\!+\!~$}}1
         {->}{{$\to~$}}1         
         {forall}{{$\forall~$}}1
         {not}{{$\lnot~$}}1         
         {exists}{{$\exists~$}}1
         {=>}{{$\Rightarrow~$}}1
         {\\~}{{$\lnot\;$}}1
         {\\+}{{$\!\join\!~$}}1
}

\lstdefinestyle{Coq}{language=Coq}
\title{QED at Large: A Survey of Engineering of Formally Verified Software}

\booltrue{authortwocolumn}
\maintitleauthorlist{
Talia Ringer \\
University of Washington \\
tringer@cs.washington.edu
\and
Karl Palmskog \\
University of Texas at Austin \\
palmskog@acm.org
\and
Ilya Sergey \\
Yale-NUS College\\ and\\ National University of Singapore\\
ilya.sergey@yale-nus.edu.sg
\and
Milos Gligoric \\
University of Texas at Austin \\
gligoric@utexas.edu
\and
Zachary Tatlock \\
University of Washington \\
ztatlock@cs.washington.edu
}

\issuesetup
{%
 copyrightowner={Author1 and Author2},
 volume        = xx,
 issue         = xx,
 pubyear       = 2019,
 isbn          = xxx-x-xxxxx-xxx-x,
 eisbn         = xxx-x-xxxxx-xxx-x,
 doi           = 10.1561/XXXXXXXXX,
 firstpage     = 102, 
 lastpage      = 281
 }

\author[1]{Talia Ringer}
\author[2]{Karl Palmskog}
\author[3]{Ilya Sergey}
\author[4]{Milos Gligoric}
\author[5]{Zachary Tatlock}

\affil[1]{University of Washington; tringer@cs.washington.edu}
\affil[2]{University of Texas at Austin; palmskog@acm.org}
\affil[3]{Yale-NUS College; ilya.sergey@yale-nus.edu.sg}
\affil[4]{University of Texas at Austin; gligoric@utexas.edu}
\affil[5]{University of Washington; ztatlock@cs.washington.edu}

\begin{document}





\makeabstracttitle

\begin{abstract}
Development of formal proofs of correctness of programs can increase actual and perceived reliability and facilitate better understanding of program specifications and their underlying assumptions. Tools supporting such development have been available for over 40 years, but have only recently seen wide practical use. Projects based on construction of machine-checked formal proofs are now reaching an unprecedented scale, comparable to large software projects, which leads to new challenges in proof development and maintenance. Despite its increasing importance, the field of proof engineering is seldom considered in its own right; related theories, techniques, and tools span many fields and venues. This survey of the literature presents a holistic understanding of proof engineering for program correctness, covering impact in practice, foundations, proof automation, proof organization, and practical proof development.
\end{abstract}

\chapter{Introduction}
\label{sec:intro}

A formal proof of program correctness can show that for all possible inputs, the program behaves as expected.
This theoretical guarantee can provide practical benefits. For example,
the formally verified optimizing C compiler CompCert~\citep{Leroy:POPL06}
is empirically more reliable than GCC and LLVM:
the test generation tool Csmith~\citep{Yang2011} found 79 bugs in GCC and 202 bugs in LLVM, but was unable to find any bugs
in the verified parts of CompCert. 

Methodologies for developing proofs of program correctness are as old as the proofs themselves~\citep{turing1949, Floyd1967Flowcharts, Hoare1971}.
These proofs were on paper and of simple programs; tools to support their development followed soon after~\citep{Milner1972b}
and have continued to evolve for over 40 years~\citep{Bjoerner2014}.
Projects based on construction of formal, machine-checked proofs using these tools are now reaching a scale comparable to that of large software engineering projects. 
For example, the initial correctness proofs for an operating system kernel took around 20 person years to develop~\citep{Klein2009}, and as of 2014 consisted of 480,000 lines of specifications and proof scripts~\citep{Klein2014micro}.

This survey covers the timeline and research literature concerning proof development for program verification,
including theories, languages, and tools. It emphasizes challenges and breakthroughs at each stage
in history and highlights challenges that are currently present due to the increasing scale of proof developments. 

\section{Challenges at Scale}

Scaling up leads to new challenges and additional demand for tool support in proof development and maintenance. For example, users may
have to reformulate properties to facilitate library reuse~\citep{Hales2017}, or to encode data structures in specific ways to aid in automation
of proofs about them~\citep{Gonthier2008}. Proof development environments need to allow users to efficiently write, check, and share proofs~\citep{Faithfull2016}; proof libraries need to allow easy search and seamless integration of results into local developments~\citep{Gauthier2015}. Evolving projects face the possibility of previous proofs breaking due to seemingly unrelated changes, justifying design principles~\citep{Woos2016} as well as support for quick error detection~\citep{Celik2017} and repair~\citep{Ringer2018}.

The research community has answered these challenges with theories, techniques, and tools
for proofs of program correctness that scale---all of which fall under the umbrella of \textit{proof engineering},
or software engineering for proofs. 
Many of these techniques draw inspiration from work in software engineering on large-scale development practices and tools~\citep{Klein2014}. 
However, even with close conceptual ties between construction of programs and proofs, research in software engineering requires careful translation to the world of formal proofs. For example, proof engineers can benefit from regression testing techniques by considering lemmas and their proofs in place of tests, as in \textit{regression proving}~\citep{Celik2017}; yet, the standard metric used to prioritize regression tests---statement coverage---has no clear analogue for lemmas with complex conditions and quantification.

This survey serves to gather these theories, techniques, and tools into a single place,
drawing parallels to software engineering, and pointing out challenges that are especially
pronounced in proof development.
It discusses the problems engineers encounter when verifying large systems, 
existing solutions to these problems, and future opportunities for research to address underserved problems.

\section{Scope: Domain and Literature}

We consider proof engineering research in the context of interactive theorem provers (ITPs) or \textit{proof assistants} (used interchangeably with ITPs in this survey) that satisfy the \textit{de Bruijn criterion}~\citep{Barendregt2002,Barendregt2351}, which requires that they produce proof objects that a small proof-checking kernel can verify; the general workflow of such tools is illustrated in Figure~\ref{fig:workflow}. That is, we consider proof assistants such as Coq~\citep{coq}, Isabelle/HOL~\citep{isabelle}, 
HOL Light~\citep{hollight}, and Agda~\citep{agda}; we do not consider program verifiers, theorem provers, and constraint solvers such as Dafny~\citep{Leino2010}, ACL2~\citep{acl2}, and Z3~\citep{z3}
except when contributions carry over. We focus on proof engineering for software verification, 
but consider contributions from mathematics and other domains when relevant.

Sometimes, the key design principles in \emph{engineering} a large program verification effort are not the focus of the most well-known publications on the effort. Instead, 
they can be in less standard references such as workshop papers~\citep{Komendantskaya2012,Blanchette2013,CompanyCoq2016,Mulhern06proofweaving}, invited talks~\citep{WenzelIsabelleFuture}, blog posts~\citep{VerifiedCryptoFirefox}, and online documents~\citep{LeroyDeepSpecSS17,WenzelScalingIsabelle}. One purpose of this survey is to bring such design principles front-and-center.  Naturally, we shall aim to survey the relevant literature with our best effort to provide accurate and thorough citations.  To that end, we will not hesitate to cite both traditional research papers in well-known venues and relevant discussions in less traditional forms, without further distinction among them.

\section{Overview}

\begin{figure}[ht]
  \centering
{\footnotesize
\begin{tikzpicture}[>=stealth',semithick]

\begin{scope}
\node[ellipse, draw=black, inner ysep=8pt, inner xsep=5pt] (user) { user };
\end{scope}

\begin{scope}[xshift=5.5cm]
\node[rectangle, draw=black, minimum width=6.5cm, minimum height=2cm] (coq) { };
\node[rectangle, draw=black, minimum height=1.5cm, fill=lightgray, xshift=-2cm] (engine) { logic engine };
\node[rectangle, draw=black, xshift=1.8cm] (checker) { proof checker };
\node[xshift=1.7cm, yshift=0.7cm] (sys) { proof assistant };
\end{scope}

\begin{scope}[xshift=9.5cm]
\node[yshift=0.6cm] (ack) {\checkmark};
\node[yshift=-0.6cm,xshift=-0.1cm] (nack) {\cross};
\end{scope}

\draw[->, thick] (user.north east) to node[xshift=-0.2cm, yshift=0.2cm] {tactics} ([yshift=0.33cm]engine.west);
\draw[->, thick] ([yshift=0.38cm]engine.south west) to node[xshift=-0.2cm, yshift=-0.2cm] {subgoals} (user.south east);

\draw[->, thick] (engine.east) to node[yshift=0.2cm] {proof} (checker.west);

\draw[->, thick] (checker.east) to (ack.south west);
\draw[->, thick] (checker.east) to (nack.north west);

\end{tikzpicture}
}
\caption{Typical proof assistant workflow, adapted from \cite{pa-history-geuvers-sadhana09}.}
\label{fig:workflow}
\end{figure}
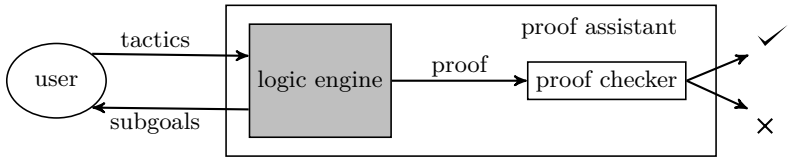

After motivating chapters (Chapters~\ref{ch:ex} and~\ref{cha:appl-mech-proofs}), this survey discusses the history and foundations of 
proof assistants (Chapter~\ref{ch:foundations}). It then surveys proof engineering research under three headings: languages and automation (Chapter~\ref{sec:prooflangs}), proof organization and scalability (Chapter~\ref{ch:organization}), and practical proof development and evolution (Chapter~\ref{ch:development}).
At a glance, Chapter~\ref{sec:prooflangs} concerns proof automation approaches and languages, Chapter~\ref{ch:organization} concerns methods to express and organize programs and proofs, and Chapter~\ref{ch:development} concerns development processes and tools.
Each of these three chapters is divided into sections; each section surveys a more granular area of proof engineering research,
then concludes with a discussion of opportunities for future work within that area when applicable.
The survey concludes (Chapter~\ref{ch:conclusion}) with a discussion of opportunities for future work within proof engineering more broadly.
In the case of factual errors, an errata may be found on \url{https://proofengineering.org}.

%

\section{Reading Guide}

This survey aims to reach a broad audience of researchers, proof engineers, and community members
who are interested in understanding, using, or contributing to proof engineering research.
Readers need not be deterred for lack of background knowledge.
It is not always necessary to understand previous chapters in order to understand later chapters;
readers should feel free to skip sections or chapters, or to consult later chapters or cited resources for more information.
This guide lists topics with which basic familiarity is helpful
in order to get the most out of the referenced chapters (all chapters unless otherwise specified),
along with resources (cited next to \Letter\ icons) for the interested reader:

\begin{itemize}
\item Programming languages, type systems, and metatheory \Letter~\citep{pierce2002types, harper2016practical}, including:
\begin{itemize}
\item ITPs \Letter~\citep{pa-history-geuvers-sadhana09, Harrison2014}, especially:
\begin{itemize}
\item Coq \Letter~\citep{CPDT, Pierce-al:SF, CoqArt} 
\item Isabelle/HOL \Letter~\citep{wenzel2004isabelle, Nipkow2014concrete} 
\end{itemize}
\item Automated reasoning \Letter~\citep{Bradley2007, Kroening2008} (Chapters~\ref{cha:appl-mech-proofs} and~\ref{sec:prooflangs})
\item The Curry-Howard correspondence \Letter~\citep{pfenning2010curry, sorensen2006lectures}
\item Dependent and inductive types \Letter~\citep{CPDT}
\item Equality \Letter~\citep{CPDT, nlab:equality} (Chapters~\ref{cha:appl-mech-proofs} and~\ref{ch:foundations})
\item Compilers \Letter~\citep{cooper2011engineering} (Chapters~\ref{cha:appl-mech-proofs}, \ref{ch:foundations}, and~\ref{ch:organization})
\end{itemize}
\item Software engineering \Letter~\citep{se-canon} 
\item Systems \Letter~\citep{Anderson2014, Cachin2011} (Chapters~\ref{cha:appl-mech-proofs} and~\ref{ch:organization})
\item Formalized mathematics \Letter~\citep{ams-proof, nlab:foundation_of_mathematics} (Chapter~\ref{cha:appl-mech-proofs})
\end{itemize}

\chapter{Proof Engineering by Example}
\label{ch:ex}

To make our notion of engineering of machine-checked proofs of program correctness more concrete, we consider the classic problem of constructing a program that decides whether a string (say, ``$aaab$'') matches a regular expression (say, $a^*b$). This problem has been used to demonstrate the usefulness of rigorous proofs of program properties by induction~\citep{Harper1999,Yi2006}.

\paragraph{Specification}
A necessary step towards a correct program is to obtain a suitable \emph{specification} (Section~\ref{sec:specenc}) of \emph{functional correctness}, in this case, what it means for a string to match a regular expression. We adopt the usual meaning in terms of the theory of regular languages. Consider an enumerable alphabet $\Sigma$ of characters $c$, where it is decidable whether two characters are equal. A string $s$ is an element of the set $\Sigma^*$, consisting of 0 or more characters from the alphabet. We call any such set of strings (subset of $\Sigma^*$) a \emph{regular language}, and refer to the special \emph{empty string} of 0 letters as $\epsilon$. We then define regular expression syntax:
\[
r ::= \;\; \textbf{0} \;\; | \;\; \textbf{1} \;\; | \;\; c \;\; | \;\; r_1 + r_2 \;\; | \;\; r_1 \cdot r_2 \;\; | \;\; r^*
\]
We next define a matching relation between regular expressions and strings informally using inductive rules:

\[
\inferrule{ }{\epsilon \triangleleft \mathbf{1}} \qquad \inferrule{ }{c \triangleleft c} \qquad \inferrule{s \triangleleft r_1}{s \triangleleft r_1 + r_2} \qquad \inferrule{s \triangleleft r_2}{s \triangleleft r_1 + r_2}
\]
\vspace{-3pt}
\[
\inferrule{s \triangleleft r_1 \quad s' \triangleleft r_2}{s\,s' \triangleleft r_1 \cdot r_2} \qquad \inferrule{ }{\epsilon \triangleleft r^*} \qquad \inferrule{s \triangleleft r \quad s' \triangleleft r^*}{s\,s' \triangleleft r^*}
\]

The intended interpretation of $s \triangleleft\, r$, ``$s$ matches $r$,'' is that $s$ belongs to the regular language defined by $r$. One way to convince ourselves that $\triangleleft$ is adequate is by proving that it holds for some cases where we expect $s$ to be in the language of $r$ (e.g., ``$aaab$'' and $a^*b$), and that it does not hold for cases where we do not expect this (e.g., ``$ab$'' and $b^*$). \citet{Yi2006} defines equationally a function which purports to decide whether $s \triangleleft r$. We write $\mathit{r ! s}$ for the application of this function to given $r$ and $s$, omitting cases that include the Kleene star ($r^*$):
\begin{eqnarray*}
 \mathbf{0} ! s \defeq \mathit{false} \qquad \mathbf{1} ! s \defeq s = \epsilon \qquad c ! s \defeq s = c \qquad r_1 + r_2 ! s \defeq r_1 ! s \lor r_2 ! s\\
  \epsilon \cdot r_2 ! c s \defeq r_2 ! c s  \quad c \cdot r_2 ! c s \defeq r_2 ! s \quad (r'_1 \cdot r'_2) \cdot r_2 ! c s \defeq r'_1 \cdot (r'_2 \cdot r_2) ! c s\\
  r_1\cdot r_2 !\epsilon \defeq r_1 ! \epsilon \land r_2 ! \epsilon \qquad (r'_1 + r'_2) \cdot r_2 ! c s \defeq r'_1 \cdot r_2 ! c s \lor r'_2 \cdot r_2 ! c s
\end{eqnarray*}

\vspace{-10pt}
\paragraph{Encoding} Our approach to specifying and verifying this function in an ITP is to make a \emph{deep embedding} (Section~\ref{sec:embedding}) of the language of regular expressions. Specifically, we encode the language as an inductive datatype (Section~\ref{sec:definitional}) in the Coq proof assistant, parameterized by an arbitrary type \lstinline{char} of characters:
\begin{lstlisting}[language=Coq]
Inductive r : Type := r_zero : r | r_unit : r | r_char : char -> r
| r_plus : r -> r -> r | r_times : r -> r -> r | r_star : r -> r.
\end{lstlisting}
We consider strings to be simple lists with elements in \lstinline{char}, i.e., to have type \lstinline{list char}. This allows us to define $s'$ appended to $s$ ($s\,s'$) as regular list concatenation ($s +\!\!\!+\, s'$). The matching relation becomes:
\begin{lstlisting}[language=Coq,keepspaces=true]
Inductive r_match : list char -> r -> Prop :=
| r_match_unit  : r_match [] r_unit
| r_match_char  : forall c, r_match (c :: []) (r_char c)
| r_match_plus' : forall s r1 r2, r_match s r1 -> r_match s (r_plus r1 r2)
| r_match_plus2 : forall s r1 r2, r_match s r2 -> r_match s (r_plus r1 r2)
| r_match_times : forall s s' r1 r2,
    r_match s r1 -> r_match s' r2 -> r_match (s ++ s') (r_times r1 r2)
| r_match_star1 : forall r0, r_match [] (r_star r0)
| r_match_star2 : forall s s' r0,
    r_match s' (r_star r0) -> r_match (s ++  s') (r_star r0).
\end{lstlisting}



We can now define the specification (type) \lstinline{acc_t} of a function \lstinline{acc} that decides the matching relation:
\begin{lstlisting}[language=Coq]
Definition acc_p (rs : r char * list char) := r_match (snd rs) (fst rs).
Definition acc_t (rs : r char * list char) := {acc_p rs}+{not acc_p rs}.
\end{lstlisting}

\vspace{-15pt}
\paragraph{Verified Implementation}
Directly writing a computable function with this specification (with the type \lstinline{acc_t}) requires us to immediately establish termination (Section~\ref{sec:totality-termination}). However, we more conveniently use a \emph{definitional extension} (Section~\ref{sec:definitional}) to encode the function and prove partial correctness and termination separately~\citep{Sozeau2010} using a well-founded relation \lstinline{acc_lt} whose definition we omit:
\begin{lstlisting}[language=Coq]
Equations acc (rs : re char * list char) : acc_t rs by wf rs acc_lt :=
 acc (re_zero, _) := right _;
 acc (re_unit, []) := left _; acc (re_unit, _ :: _) := right _;
 acc (re_char c, [c']) := match char_eq_dec c c' with
  left _ => left _ | right _ => right _ end;
 acc (re_char _, []) := right _; acc (re_char _, _ :: _) := right _;
 acc (re_plus r1 r2, s) :=
  match acc (r1, s) with left _=> left _ | right _ =>
   match acc (r2, s) with left _ => left _ | right _ => right _ end end;
 (* ... re_times/re_star cases and proof scripts omitted ... *)
\end{lstlisting}

While our certified matching function follows the equational definition provided by Yi quite closely, we may have made some mistakes. One way to ensure the adequacy of our proof (Section~\ref{sec:trust-proofs}) is to connect it to a Coq theory of regular languages~\citep{Doczkal2018}. We use this theory's notion of language given by a regular expression (\lstinline{re_lang}) to build a matcher that uses our original \lstinline{acc} function, by converting from the theory's notion of a regular expression:
\begin{lstlisting}[language=Coq]
Program Definition acc' (r : regexp char) (w : list char) :
 {w \in re_lang r}+{w \notin re_lang r} :=
  match acc (@eq_comparable char) (regexp2re r, w) with
  left _ => left _ | right _ => right _ end.
\end{lstlisting}

\vspace{-15pt}
\paragraph{Executable Code}
Finally, we can extract (Section~\ref{sec:trust-programs}) the Coq code for the function \lstinline{acc} to a practical programming language such as OCaml and run it on example regexp-string pairs where we have an intuition whether they should match or not. We could also refine (Section~\ref{sec:refinement}) the code to more performant Coq code, or to an imperative language (Section~\ref{sec:reas-about-imper}) embedded in Coq, making it possible to use 
certified compilation (Section~\ref{sec:verified-compilers}) to assembly language.

%




\chapter{Why Proof Engineering Matters}
\label{cha:appl-mech-proofs}

Formal verification of a program can improve actual and perceived reliability.
It can help the programmer think about the desired and actual behavior of the program,
perhaps finding and fixing bugs in the process~\citep{murraybp}.
It can make explicit which parts of the system are trusted, and further decrease the burden
of trust as more of the system is verified.

One noteworthy program verification success story is the CompCert~\citep{Leroy:POPL06, Leroy2009} verified optimizing C compiler.
Both the back-end and front-end compilation passes
of CompCert have been verified, ensuring the correctness of their composition~\citep{Kaestner2017}.
CompCert has stood up to the trials of human trust: it has been used, for example, to compile code for safety-critical flight control software~\citep{Frana2011}.
It has also stood up to rigorous testing: while the test generation tool Csmith~\citep{Yang2011} found 
79 bugs in GCC and 202 bugs in LLVM, it was unable to find any bugs in the verified parts of CompCert.

CompCert, however, was not a simple endeavor: The original development comprised of approximately 35,000 lines of Coq code;
functionality accounted for only 13\% of this, while specifications and proofs accounted for the other 87\%.
This is not unusual for large proof developments. The initial correctness proofs for an operating system (OS) kernel in Isabelle/HOL,
for example, consisted of 480,000 lines of specifications and proofs~\citep{Klein2014micro}.

Much like software engineering theories, techniques, and tools help software engineers deal with the challenges
of programming for scale, so proof engineering helps proof engineers.
To highlight the importance of proof engineering, we start by surveying a sample of domains in which 
proof engineering for program verification has been particularly influential (Section~\ref{sec:app-prog}).
We then discuss proof engineering for other domains (Section~\ref{sec:other-domains}),
as the lessons learned from those domains transfer to program verification as well.
We conclude by describing some examples of practical impact from program verification (Section~\ref{sec:empirical}).



\section{Proof Engineering for Program Verification}
\label{sec:app-prog}

We discuss a sample of the domains in which proof engineering has had a large impact:
certified compilers (Section~\ref{sec:verified-compilers}),
low-level systems software (Section~\ref{sec:systems-software}),
concurrent and distributed systems (Section~\ref{sec:distributed-systems}),
and ceritified solvers and checkers (Section~\ref{sec:solvers}).

\subsection{Certified Compilers}
\label{sec:verified-compilers}

Development of certified compilers is a classic application of proof assistants (Section~\ref{sec:early-history}). 
In spite of their long history, however, certified compilers for practical and widely used programming languages have only started to appear in the last
decade. This is mainly due to the sheer size and complexity of the semantics of these languages---challenges that have necessitated 
developments in proof engineering.

The inertia for those developments came in 2006 with the help of CompCert.
This project was a pivotal moment in the history of program verification.
\cite{Leroy:POPL06} received the POPL test of time award in 2016~\citep{popl-time}, with the release noting
its pivotal role:

\begin{quote}
The paper was (and still is) groundbreaking in that it demonstrates the feasibility of using an interactive theorem prover---specifically, Coq---to both program and formally verify a realistic compiler ... [it] made a convincing case that theorem-proving technology is mature enough to be applied to the full functional verification of realistic systems, and in so doing heralded a new age of ``big verification.''
\end{quote} 

Many projects built on CompCert by, for example, using its C semantics, or by targeting C or any of its intermediate representations.
\cite{Appel2011} developed a program logic based on CompCert's C semantics, allowing Coq users to prove properties of deeply embedded C programs in Coq. These properties then hold for machine code generated by CompCert, via CompCert's main correctness result~\citep{Appel:BOOK14}. \cite{Kaestner2017} described many extensions and enhancements to the basic compilation toolchain of CompCert, e.g., translation validation of the process of \emph{linking} machine code to produce object files and executable files. \cite{Cao2018} presented a CompCert-based C program verification environment called VST-Floyd in Coq based on separation logic, simplifying the process of specifying and verifying properties that hold down to machine code.

Certified compilers now span broad applications:
For example, the \textsc{Cogent} language for system programming and verification is accompanied a certifying compiler~\citep{oconnor2016b} 
in Isabelle/HOL that produces a proof that the generated C code is correct. The compiler uses a refinement framework (Section~\ref{sec:refinement}) to
automate relating the \textsc{Cogent} semantics to the generated code~\citep{Rizkallah16}.
The Standard ML language variant CakeML~\citep{Kumar2014} has a verified compiler in HOL4
with a certified machine-code implementation produced by bootstrapping (applying the compiler to itself).
Using the Vellvm~\citep{zhao2012} framework in Coq, proof engineers can reason about transformations on the LLVM intermediate language representation.

Certified compilers have also covered new ground with respect to modularity and compositionality.
The IMM~\citep{Podkopaev2019} memory model modularizes certified compilation from high-level concurrent programming languages to different hardware models.
The Bedrock~\citep{Chlipala2013} intermediate language and verification environment in Coq for low-level programming
contains a notion of macros that can be reasoned about modularly.
The Pilsner~\citep{Neis2015} certified compiler in Coq from a higher-order ML-like language to machine code
supports programs that can, in contrast to CompCert, be compositionally verified.
Compositional CompCert~\citep{Stewart2015} is a variant of CompCert with a correctness theorem that can be applied compositionally. 

Proof engineering for certified compilers has had applications directly to the languages these compilers are verified in.
For example, there are certified compilers both from HOL4~\citep{Myreen2012} and from Isabelle/HOL~\citep{Hupel2018}
to CakeML; it is possible to produce machine code by composing these compilers with the certified CakeML compiler.
Both CertiCoq~\citep{Anand2017} and \OE uf~\citep{Mullen2018} describe certified compilers for Coq's specification
language Gallina. CertiCoq is an ongoing project to build a certified compiler from Gallina to machine code, using a 
hierarchy of custom intermediate languages. \OE uf presents a certified compiler from a subset of Gallina to assembly code.
Both compilers target CompCert intermediate languages: CertiCoq targets Clight, while \OE uf targets Cminor.
Each of these compilers provide an alternative to untrusted program extraction (Section~\ref{sec:trust-programs}).

\subsection{Low-Level Systems Software}
\label{sec:systems-software}

In addition to compilers, low-level software such as operating systems, file systems, and network stacks are building blocks of large software systems. In turn, such software relies on interfaces to hardware, and on hardware behavior. Through considerable effort, proof engineers have specified and verified important pieces of systems software and their hardware bases.

In pioneering work, \cite{Klein2009,Klein2014micro} developed a small general-purpose OS kernel, called seL4, in the C programming language, with correctness proofs in Isabelle/HOL. The proven properties include correctness of interprocess communication, access-control enforcement, and information-flow noninterference. Many verified extensions to seL4 have been proposed since its inception, e.g., to ensure hard deadlines are met for system calls, which is required for applications in real-time systems~\citep{Sewell2017}.

Verification of OS kernels has brought new developments in proof engineering.
For example, as part of the CertiKOS project, \cite{Gu-al:POPL15} presented a framework in Coq for specifying and verifying abstraction layers, which they used to develop several certified OS kernels. In doing so, they introduced the idea of a deep specification (Section~\ref{sec:des-scale}). 
\cite{Gu-al:OSDI16} designed and verified a concurrent kernel for the x86 architecture with fine-grained locking.

An Instruction Set Architecture (ISA) provides an important interface to hardware for, .e.g., compilers.
\cite{Fox2010} developed formal specifications (semantics) of the ISA for ARMv7 in HOL4. \cite{Armstrong2019} later used a domain-specific language to provide ISA specifications for Isabelle/HOL, HOL4, and Coq for the ARMv8, RISC-V, and CHERI-MIPS architectures.
\cite{Morrisett2012} modeled a subset of the x86 ISA in Coq, and used it to build a verified checker for a sandbox policy.

Security policies of systems software are apt targets for verification due to the importance of them being correct.
Along these lines, \cite{Dam2013} proposed a separation kernel (hypervisor) based on the ARMv7 processor architecture and its formalization by Fox and Myreen, and proved, in HOL4, an information flow security property that ensures OS instances can only communicate via explicit channels. \cite{Guanciale2016} proved memory virtualization security in HOL4 for ARMv7.

OS kernel subsystems and formats are additional formalization targets: \cite{Bishop2006,Bishop2018} defined and validated executable specifications in HOL4 for the TCP/IP stack, and \cite{Kell2016} formalized the ELF binary format in HOL4 for executables used in Unix-like operating systems such as Linux.

Verification has also reached file systems, such as FSCQ~\citep{Chen2015}, a file system with guarantees about crash safety that have been verified in Coq,
and DFSCQ~\citep{Chajed2017}, an efficient crash-safe file system with several verified optimizations.
Using the \textsc{Cogent} language and its certifying compiler (Section~\ref{sec:verified-compilers}),
\cite{Amani16} developed a file system called BilbyFS in Isabelle/HOL with executable code in C; they also implemented and verified the legacy Linux file system ext2. \cite{Ridge2015} developed a specification of POSIX file systems in HOL4, which they tested against real-world file system behavior.

\subsection{Concurrent and Distributed Systems}
\label{sec:distributed-systems}

Concurrent and distributed systems can be difficult to develop, understand, and debug. Researchers have developed many different theories and frameworks for reasoning about such systems in proof assistants.
For example, FCSL~\citep{Sergey-al:PLDI15} is a framework in Coq for reasoning about fine-grained concurrency, based on a shallow embedding of concurrent programs.
The \textsc{Disel}~\citep{Sergey2017} framework for distributed separation logic in Coq builds on this,
using the shallow embedding approach from FCSL. 
The formalization is extracted to executable OCaml code and run on real hardware.
Two different frameworks~\citep{zeller2014, gomes2017verifying} in Isabelle/HOL exist for verifying two different models
of CRDTs, replicated datatypes
which provide strong eventual consistency guarantees.

Progress in proof engineering for verification of concurrent and distributed systems has had implications for the formalization of practical 
programming languages and protocols. This is illustrated by the \cite{Jung2017} formalization of the imperative and threaded Rust programming language
using the Iris~\citep{Jung2018} framework for concurrent separation logic in Coq,
and by the \cite{Woos2016} implementation and verification of the key correctness property of the Raft consensus algorithm 
using the Verdi~\citep{Wilcox2015} framework for verification of asynchronous message-passing distributed systems.

\subsection{Certified Solvers and Checkers}
\label{sec:solvers}

Proof engineers have formalized and proven correct automated solvers for first-order logic and other more restricted logics.
\cite{blanchette2018verified} verified a SAT solver in Isabelle/HOL with conventional features such as clause learning. \cite{schlichtkrull2019verified} verified a purely functional first-order superposition-based solver in Isabelle/HOL and obtained executable code in Standard ML.

Another line of work in Isabelle/HOL formalizes various model checkers, for example
for Linear Temporal Logic~\citep{esparza2013fully} and for timed automata~\citep{wimmer2018verified}.


\section{Proof Engineering for Other Domains}
\label{sec:other-domains}

While this survey focuses on proof engineering for program verification,
domains outside of program verification encounter similar proof engineering challenges.
The solutions that these communities develop have implications for proof engineering
for program verification. We discuss these impliciations for two domains:
mathematics (Section~\ref{sec:dsmathematics})
and programming languages metatheory (Section~\ref{sec:dsmetatheory}).

\subsection{Mathematics}
\label{sec:dsmathematics}


Mathematics is a natural application domain for proof assistants.
\textit{Formalized mathematics} is the attempt to formalize mathematical 
theories in part or in whole using proof assistants, so that proofs can be 
mechanically checked. 
Formalized mathematics was one of the first major application domains for proof assistants;
several early ITPs and their predecessors
were designed with this domain in mind (Section~\ref{sec:early-history}).

Early formal developments in mathematics arose in the 
1990s~\citep{zucker1994formalization, van1994checking, bancerek1990fundamental, qed-manifesto-boyer-cade94}.
Since then, there have been many notable developments for formalized mathematics,
including the Four Color Theorem~\citep{Gonthier2008}, the Kepler
conjecture~\citep{Hales2011, Hales2017}, the fundamental
theorem of algebra~\citep{Geuvers2000}, G{\"o}del-Rosser incompleteness~\citep{OConnor2005}, and the Jorden curve theorem~\citep{Hales2007}. 
Other interesting mathematical proofs can be found in a comparative overview of different proof assistants for mathematics~\citep{Wiedijk2006}.
The QED manfesto~\citep{qed-manifesto-boyer-cade94} called for a complete database of formalized mathematics.
The UniMath~\citep{UniMath} library is an ongoing attempt to formalize foundations of mathematics~\citep{Voevodsky2015}
in Coq, using homotopy type theory (Section~\ref{sec:equality}).

This section samples tools and design principles for formalized mathematics,
and discusses how they are relevant to program verification.

\paragraph{Design Principles} Many proof developments in mathematics are mature and
involve a large community of contributors. Several of these developments have 
style guides for contributors.
For example, UniMath library has a style guide
that serves to make proofs rigorous, easy to
port to other proof assistants, and less fragile, and to standardize
and improve appearance and readability. Among other things, the style
guide prohibits the addition of new axioms, and encourages the use of
tactics whose semantics are well-defined.  Similarly, the HoTT
library~\citep{Bauer2017} for homotopy type theory contains a style
guide which, among other things, encourages uniform naming principles, outlines methods for
defining equivalences, describes how to use axioms uniformly,
and encourages the use of tactics that have well-defined relationships
with the terms they produce. 

In addition to style guides, proof engineers have developed design principles to handle certain kinds of problems
common in mathematics. \cite{Gonthier2013}, for example, outlines
a number of techniques used in the proof of the Odd Order Theorem.

\cite{Wiedijk2006} compares the styles
of seventeen different proof assistants for mathematics. The book is a collection of
proofs of the irrationality of $\sqrt{2}$ from users of each of the proof assistants,
and a discussion of each of the proof assistants and the proofs in those proof assistants.
This comparison can be useful for understanding the tradeoffs of and design considerations in each of the proof assistants.

\paragraph{Tooling} Formalized mathematics has also seen the development of tooling
to support entire classes of proofs. Notable examples include autarktic
computations for algebraic reasoning~\citep{Barendregt2002}, special support
for equational reasoning in proof checkers~\citep{Barthe1996}, 
decision procedures for fragments of arithemic (Section~\ref{sec:autotactics}),
techniques for reasoning modulo associativity and commutativity~\citep{Braibant2011},
transport methods (Sections~\ref{sec:languagereuse} and~\ref{sec:toolingreuse}),
and theory exploration (Section~\ref{sec:autotactics}).

\paragraph{Beyond Mathematics} In formalized
mathematics, user communities of certain frameworks or
libraries adhere to style guides and design principles. These style
guides and design principles have different emphases. 
Proof engineers in communities outside of mathematics may similarly benefit from standardizing
some elements of style depending on the desired outcome.  On a
project-by-project basis, this may make collaboration between proof
engineers easier, and limit the accidental introduction of untrusted
code.

Style guides and design principles also have implications for proof understanding. 
Mathematicians are the original proof checkers, so it's perhaps unsurprising that many
mathematics communities emphasize proof understanding by humans. 
Beyond mathematics, human understanding of
proofs communicates information to the reader beyond the theorem
statement itself. Furthermore, just like in software engineering,
effective collaboration between proof engineers hinges on mutual
understanding of the underlying code.

Standardizing style may also make automation easier. For example, by limiting
the set of tactics used within a community, the produced proofs are
more clearly defined, which may make higher-level automation such as
refactoring and repair tools (Section~\ref{sec:evolve}) less challenging.

Proof engineers in communities outside of mathematics may also benefit
from comparitive studies across different proof assistants, similar to
\cite{Wiedijk2006} but for other domains.

Finally, much of the tooling developed for mathematics addresses problems
that occur in proof developments outside of mathematics. 
For example, dealing with equivalences and isomorphisms is a problem that
is not exclusive to the domain of mathematics.
Many of the techniques and tooling from mathematics that solve
this problem may be useful for proof engineers who encounter this 
same problem in other domains.

\subsection{Programming Language Metatheory}
\label{sec:dsmetatheory}


One large domain of focus is \textit{mechanized metatheory}:
proofs about programming languages. The desire to mechanize theory led to the introduction of the 
Edinburgh Logical Framework (LF) by \cite{Harper1987} (and later in more detail, \cite{Harper1993}),
building on ideas from Automath.
LF defined a methodology for encoding and reasoning about a simpler programming language from within the higher-order dependently typed
lambda calculus. Mechanizations of metatheory followed shortly after, both in Nuprl (for example, \cite{howe1988computational})
and in LF (see \cite{Harper2007} for an overview of mechanized metatheory in LF, and \cite{harrison-reflection} for
an early history of mechanized metatheory more broadly).

Since then, the domain has grown to reach practical languages:
The mechanization~ \citep{mechanized-sml} of Standard ML in Twelf formalized
the metatheory of a practical language in its entirety.
WebAssembly has had a formal semantics from the very beginning, which has been mechanized~\citep{watt2018mechanising}
in Isabelle/HOL.
Simplified languages representing the core underlying theories of Scala~\citep{Rompf2016, Amin2017} and of
OCaml~\citep{owens2008} have been verified. 
Results from mechanized metatheory have also influenced verification of real compilers,
like CakeML (Section~\ref{sec:verified-compilers}).

The success of metatheory has brought with it benchmark suites and design principles.
In addition, mechanized metatheory has influenced new additions to the core languages of ITPs.
This section describes a small sample of those benchmark suites, design principles, and language features,
and discusses how the lessons learned from mechanized metatheory generalize beyond
this domain.

\paragraph{Benchmark Suites} Some of the success of mechanized metatheory is attributable
to benchmark suites that have clearly established the importance of the domain
and set out to define how to measure progress within it. The
\poplmark challenge~\citep{Aydemir2005} has been particularly influential.

The benchmarks in the \poplmark challenge are proofs of properties of
the language System F-Sub~\citep{Cardelli1994}, which has parametric polymorphism and
subtyping. \poplmark highlights specific problems in
proof engineering for metatheory, 
and outlines criteria for evaluating the success of technology that
addresses these problems.



15 solutions to the \poplmark challenge remain accessible online~\citep{poplmark-website}. 
Of these solutions, 8 are in Coq, 2 are in Isabelle/HOL,
and the remaining 5 are spread across 5 other ITPs. The solutions cover 8 different ways to represent variable \textit{binders},
one of the problems that \poplmark highlights.

Personal communications~\citep{harpersonal, piercenal2}
suggest that the solution using Twelf~\citep{Pfenning1999} (an LF implementation) was the first solution to solve
all of the difficult parts of the challenge.
The website notes that this solution demonstrated the benefits
of using that framework, including the style of binders it supports, while also sparking an interesting discussion
on different ways of specifying the problem across different frameworks.
Only Arthur Chargu{\'e}raud attempted the same solution in the same proof assistant with different styles of binders;
while his solutions were inconclusive, they inspired later work on making binders easier to represent~\citep{Aydemir2008, Chargueraud2011}.

The \poplmark solutions may be thought of as a springboard for later work. They provided information about what
the state-of-the-art was at the time, which enabled later researchers to measure progress.
Over 300 papers have cited \poplmark since its introduction in 2005.

Still, there is some dissatisfaction with the outcomes of \poplmark.
For example, mst papers that cite \poplmark focus on the difficulty of dealing with
binders, which is just one challenge that proof engineers face when
mechanizing metatheory~\citep{piercenal}. The List-machine Benchmark~\citep{appel2012list},
developed in parallel with \poplmark, deemphasizes
binders and instead emphasizes connections between proofs and real
implementations of compilers. 
\poplmark Reloaded~\citep{abel2017poplmark} emphasizes
logical relations proofs. The ORBI~\citep{felty2018benchmarks, felty2015benchmarks} benchmarks
focus on the tradeoffs of design decisions of different systems for mechanizing metatheory,
rather than on different approaches using a given system.

\paragraph{Design Principles} 
Many papers that cite \poplmark focus on difficulties of dealing with binders. Paper proofs
typically use a \textit{named} representation, where variables are
represented by names. These are easy for humans to reason about, but
make it difficult for tools to reason about alpha-equivalence.
Nominal logics~\citep{Aydemir2006, Urban2008} encode names such that
equivalent terms are alpha-equivalent. These representations have
the benefits of named representations, but without the cost of
difficulty reasoning about alpha-equivalence. While the nominal
approach is common in Isabelle~\citep{Urban2008}, only preliminary
approaches exist in Coq~\citep{Aydemir2006}. 
Developing such a tool in Coq may be difficult due
to the differences in logics between Coq and Isabelle/HOL, as Nominal Isabelle
makes use~\citep{urban2011} of quotient types~\citep{Homeier05} in HOL. 

In the absence of a nominal tool for Coq, proof engineers may
explore tradeoffs between \textit{de Bruijn indexes}
(introduced by \cite{debruijn1972} and explored by \cite{owens2008} and
\cite{schafer2015}, among others) and a \textit{locally nameless}
representation (introduced by \cite{Gordon1993b} and explored by \cite{Leroy2007}, \cite{Aydemir2008},
and \cite{Chargueraud2011}, among others). De Bruijn indexes make reasoning
about alpha-equivalence simple because alpha-equivalence is
definitional equality (Section~\ref{sec:equality}), but they require shifting operations in the code,
which complicates human understanding; 
\cite{Berghofer2007} detail the tradeoffs between de Bruijn indices and names.
Locally nameless attempts to capture the best of both worlds:
it uses names for free variables and de Bruijn indexes
for bound variables. Consequentially, locally nameless requires
reasoning about indexes and shifting operations only for locally closed terms.

LF-based systems like Twelf~\citep{Pfenning1999}, Delphin~\citep{poswolsky2009system}, Beluga~\citep{pientka2008programming} use \textit{higher-order abstract syntax} (HOAS)~\citep{Pfenning1988} to simplify reasoning
about binders. HOAS gives an encoding of binders in the object language (the language that is reasoned about) as binders in the meta-language
(the language of reasoning---in the case of LF, the higher-order dependently typed lambda calculus). 
Beluga uses ideas from contextual modal type theory~\citep{nanevski2008contextual} to further simplify reasoning about HOAS encodings.
The Hybrid~\citep{ambler2002} tool makes it possible to use HOAS within Isabelle/HOL;
\cite{capretta2007hybrid} describe a version of Hybrid for Coq, and 
\cite{felty2012hybrid} describe how to use Hybrid to reason about an object language in a manner
similar to Twelf. \cite{felty2010hoas} compare Twelf, Beluga, and Hybrid on case studies
of metatheory using HOAS, and presents
a set of challenge problems which highlights the differences between these systems.
Variants of HOAS such as weak HOAS~\citep{Ciaffaglione2012}
and \textit{parametric higher-order abstract syntax} (PHOAS)~\citep{Chlipala2008}
make HOAS more tractable in general-purpose ITPs like Coq.

There is a small amount of work on design principles that address the concerns of \poplmark beyond dealing with binders.
For example, Engineering Formal Metatheory~\citep{Aydemir2008} 
identifies specific lemmas that are useful and discusses the organization of theorems, proofs, and automation.
It also introduces \textit{cofinite quantification} of free variables in inductive
relations---defining relations that hold on all but finitely many variables, 
rather than for some fresh variable. This strengthens the premise of the relations,
which in turn strengthens inductive hypotheses for proofs. 

Design principles for mechanized metatheory often go hand-in-hand with high-level frameworks such as 3MT (Section~\ref{sec:metaframe}),
or with domain-specific languages such as Ott and Lem (Section~\ref{sec:specenc}).
Other work in design principles for mechanized metatheory includes an overview
of different ways of formalizing language semantics in an ITP
for the same language~\citep{Bertot2009}, and the use of the coinductive partiality monad~\citep{Capretta2005}
in Agda to define denotational semantics~\citep{Danielsson2012}.

\paragraph{Beyond Metatheory} Few domains have seen as much movement
in the development of design principles for proof engineering as
mechanized metatheory. 
Opinions on the role that \poplmark played in this are mixed~\citep{piercenal, personappel}. There is little
disagreement that \poplmark was timely: Proof assistants were becoming
more usable, and the ongoing development of CompCert (Section~\ref{sec:verified-compilers})
inspired confidence in their usefulness. At the same time, the properties that researchers wanted to
prove about their languages were becoming larger and more complex. It
was becoming difficult to know that these properties were actually correct,
and to maintain confidence in correctness in the face of changes. 

\poplmark gave a common platform
for experimentation and offered a concrete criteria for success in a
timely domain.  The benchmarks were difficult enough to stress
technology, but simple enough that they were easy to understand and
that experts could prove them in a few weeks.
The work in metatheory that followed \poplmark demonstrated
that general-purpose proof assistants really were usable to prove these properties
that researchers cared about.

Work to this
day continues to use \poplmark as an evaluation metric. 
\cite{Amin2017}, for example, introduce a
proof technique using definitional interpreters that addresses the
open challenge of scaling type soundness proofs to realistic
languages, and evaluate the success of
this technique using the F-Sub language from the \poplmark
benchmarks. This highlights that sometimes, proving
a slightly different property and then showing how that relates to the
original property can be much simpler than proving the original
property directly.

\poplmark suggests that timely benchmark suites are instrumental in
bringing the challenges of design for proof engineering to the
attention of the research community; in doing so, however, they can
narrow the focus to one particularly difficult problem, sometimes to
the exclusion of the bigger picture. Domains outside of metatheory
can take this into consideration.

In addition, the success of experts using LF and Twelf on the \poplmark benchmarks and in mechanizing a practical programming language
suggests that it is worth weighing carefully the tradeoffs of using different ITPs, including ITPs with
special support for a given domain. Along those lines, \cite{Miller2018} argues
that handling of variable bindings should be built into ITPs. 
It is also worth considering the barriers to adoption by non-experts of tools with which experts
have demonstrated success within a domain, and how to overcome those barriers. 
SASyLF~\citep{Aldrich2008}, for example, is one attempt to make LF-based ITPs
more accessible to students.

\section{Practical Impact}
\label{sec:empirical}

Proof engineering has already had a large impact on program verification in many domains,
including those from Section~\ref{sec:app-prog}.
Proof engineers have in recent years verified
operating system~\citep{Klein2009} and web browser~\citep{Jang2012} kernels,
machine learning systems~\citep{DBLP:journals/corr/SelsamLD17},
distributed systems~\citep{Woos2016},
quantum circuits~\citep{rand2017},
constraint solvers~\citep{blanchette2018verified, schlichtkrull2019verified},
compilers~\citep{Leroy:POPL06, Kumar2014},
and file systems~\citep{Chen2015, Amani16}.

So far, proof assistants have had the strongest practical impact in systems software.
The CompCert verified compiler, sold as a commercial product, is finding applications in embedded systems, such as those used in aviation~\citep{CompCert-ERTS-2018}. The BoringSSL library, used in the popular Google Chrome Web browser, recently started to include high-performance cryptographic code in C verified in Coq~\citep{Erbsen2019}. The seL4 verified operating system kernel is used in SCADA systems, and aviation and automotive systems~\citep{Klein2018}.

\chapter{Foundations and Trusted Bases}
\label{ch:foundations}

We give a short overview of the pre-history (Section~\ref{sec:pre}) and early history (Section~\ref{sec:early-history}) of proof assistants, and then 
discuss their foundations (Section~\ref{sec:foundations}) and trusted bases (Section~\ref{sec:trustedbases}). 

\section{Proof Assistant Pre-History}
\label{sec:pre}

Specification and verification of software systems can be viewed as reducing human, informal notions and reasoning to systematic application of logical principles and axioms. From this perspective, Aristotle's systematization of the principles of correct reasoning~\citep{AristotlePriorAnalytics,sep-aristotle-logic} is arguably the oldest precursor. The proposal of \cite{Leibniz1685} to reduce human reasoning to mathematical calculation 
is a second important step. Leibniz also laid the foundations for symbolic propositional logic, although this was also done independently by, e.g., Boole.

A later important development was the introduction of predicate logic (or \emph{predicate calculus}) by Frege in the late 19th century~\citep{Frege1893,sep-frege}. Two key innovations in Frege's logical system were (1) the introduction of \emph{quantifiers} of expressions in propositions, and (2) a notion of proof (sequences of valid inferences) for propositions with quantifiers. Frege was able to capture and prove concepts from number theory in his system from first principles. However, his system included an axiom later shown by Russell to make the system inconsistent~\citep{sep-frege-theorem}. Nevertheless, the logics of ITPs based on higher-order logic are reminiscent of Frege's logic (which included second-order quantification), and his notion of proof is similar to the modern conception.

In the early 20th century, Russell and Whitehead continued Frege's work of putting mathematics on a firm logical basis. Crucially, this included developing methods for avoiding inconsistencies, e.g., due to unrestricted formation of sets of entities~\citep{Russell1918}. In the end, they proved many significant theorems of arithmetic and set theory in their logical system by rigorous inference~\citep{Whitehead1997}, but relied on axioms that were considered questionable at the time~\citep{sep-principia-mathematica}; this is echoed in more recent concerns for philosophical justification of the basis for the logical system underpinning a proof assistant~\citep{Barras2010}.
\citet{Goedel1930} established the connection between truth and provability for \emph{first-order} predicate logic, showing that proofs of true propositions can always be constructed, in principle (systems of inference rules can be made \emph{complete}). However, he then subsequently established that even modest extensions of expressibility in first-order logic lead to \emph{incompleteness}~\citep{Goedel1931}: there can be no system that allows constructing proofs for all true propositions. Together with other negative results, e.g., by \cite{Tarski1936}, this ended the search for a single universal logical system as a foundation for mathematics and all mathematical endeavors.

At roughly the same time, a theoretical basis for computation and computer programs was given by \cite{Church1936} in the $\lambda$-calculus, and computers were developed more practically by von Neumann and others in the 1940s~\citep{vonNeumann1993}. As pointed out by \cite{Backus1978}, the $\lambda$-calculus and computers as described by von Neumann gave rise to two distinct program styles: the functional style is characterized by computational steps as reductions of \emph{expressions} and an absence of state, and the imperative style is characterized by computation as transitions between complex states and \emph{statements} that effect such transitions.

Also around that time, \cite{Curry1934} observed a connection between axioms and type systems.
This and later observations culminated in 1969 (published in \cite{Howard1980}) with the principle of \emph{formulae-as-types},
also known as \emph{propositions as types}, the Curry-Howard correspondence, or the Curry-Howard isomorphism.
This principle established the connection between programs and proofs, which provided groundwork for the
later development of ITPs.

\cite{turing1949} first considered the problem of correctness for an imperative program that computes the factorial of its input by repeated additions. He described how full correctness could be decomposed into verifying assertions associated with certain points in the code (today called \emph{invariants}), and how to ensure program termination by finding a consistently decreasing quantity (today called a \emph{variant} or \emph{ranking function}). However, this work remained obscure, and more systematic approaches for reasoning about imperative programming languages were presented only late in the 1960s~\citep{Floyd1967Flowcharts,Hoare:CACM69}.

\cite{McCarthy1960} proposed a practical realization of the functional style of programming in the form of the Lisp language. \cite{McCarthy1963} also highlighted the problem of putting computing and programs on a formal foundation. To this end, he proposed several formalisms for capturing different classes of functions, and showed how to reason about the equivalence of such functions. He also described how datatypes could be constructed recursively and be subject to inductive reasoning. \cite{Burstall1969} showed how to reason about more practical programs in the functional style using the principle of structural induction.

Research on logical reasoning using computers initially took two main forms~\citep{Warden2009Book}: (1)~fully automated proofs of propositions in simple proof systems such as Robinson's resolution system~\citep{Robinson1965}, and (2)~computer checking of the validity of single steps in human-constructed mathematical proofs, as in the Automath system by de~Bruijn~\citep{DeBruijn1970,DeBruijn1994}; Automath is notable for representing both propositions and proofs in the same formal system (a variant of the $\lambda$-calculus). The former approach is limited by the difficulty (and resulting long machine time) of finding proofs algorithmically and its bounds on expressiveness of propositions, while the latter is limited by the ingenuity (and labor supply) of the humans that construct the proofs that the system checks.
The legacy of Automath includes the \emph{de~Bruijn principle}~\citep{Barendregt2002}, which states that proof-checking programs should be as small and simple as possible to facilitate high assurance and trustworthiness.


\section{Proof Assistant Early History}
\label{sec:early-history}

In the early 1970s, Milner proposed an approach to computer proofs in between full automation and basic inference checking. One of his insights was that fine-grained automation can be directed by human ingenuity through so-called \emph{proof tactics} (Section~\ref{sec:tactics}), alleviating the burden on users in Automath-style systems. He also chose an underlying formal system (Scott's logic of computable functions~\citep{Scott1993}) that could represent concepts familiar to computer scientists and programmers, such as integers, lists, and computer programs themselves~\citep{Gordon2000}. The first implementation of his approach, called Stanford LCF~\citep{Milner1972b}, provided a workflow still used in several modern proof assistants, where the user inputs a command (e.g., a single tactic to apply to attempt to reach the current proof goal) and the system executes the command, resulting in a complete proof or in a number of subgoals. Although tactics could be complex, the system guaranteed that a proof reported as finished could be exported and verified independently by an Automath-style checker~\citep{Warden2009Book}.

Limitations on the flexibility and scalability of Stanford LCF prompted Milner to develop ML (Meta Language), a programming language for use in a new version of LCF. ML was a typed language, and Milner defined a theorem in LCF as an abstract data type whose predefined values were instances of axioms and whose operations were inference rules. This technique, which persists in some proof assistants today, is usually referred to as the ``LCF approach,'' and it in effect reduces the soundness of inferences in an embedded logical system to the soundness of the type system (and type checking mechanism) of the host language. Adventurous and flexible tactics could be implemented in ML and applied without concern for affecting soundness, although, e.g., termination was not guaranteed.

The resulting implementation of LCF in ML was called Edinburgh LCF~\citep{Milner1979}, and was further developed mainly by Paulson and Huet, who enhanced its reasoning capabilities and wrote a compiler for ML to avoid the overhead of interpretation~\citep{Gordon2000}. Paulson then went on to develop the Isabelle proof assistant framework~\citep{Paulson1994,Paulson2000}, and Huet to develop, with Coquand, the first version of the Coq proof assistant~\citep{Coquand1985}. 
Gordon used the last version of the LCF system, called Cambridge LCF, as a basis for the HOL proof assistant~\citep{Gordon1993}. 
The Nuprl proof assistant~\citep{Constable1986} also followed in the LCF tradition.
Together, these proof assistants comprise the \emph{LCF family}, and their recent incarnations
are now widely used in the research community. 
Recently developed proof assistants such as RedPRL~\citep{redprl} and Lean~\citep{deMoura2015} have also joined the LCF family.
ML was standardized as Standard ML~\citep{Milner1997}, and it and its dialects are widely used as implementation languages for proof assistants.

While Automath targetted mathematics, the initial applications of LCF-style systems for verification was in the area of programming languages and compilers. Stanford LCF had case studies for verified compilation of an imperative language to a stack-based language~\citep{Milner1972}, and equational theories on integers and lists~\citep{Newey1973}. Edinburgh LCF had case studies for verified programming language implementations~\citep{Cohn1983}, and an important use case of HOL was hardware verification~\citep{Boulton1992}.

\cite{pa-history-geuvers-sadhana09} and \cite{Harrison2014} provide a more comprehensive description of the history of ITPs.

\section{Proof Assistant Foundations}
\label{sec:foundations}

The foundational theories of many ITPs are based on some variation of the theory of \emph{types}, which goes back to Russell and his attempt in the early 1900s to avoid inconsistency in formal systems by forbidding pathological cases such as the set of all sets with some arbitrary property~\citep{sep-russell-paradox}. Specifically, \cite{Church1940} introduced the \emph{simply typed} $\lambda$-calculus to avoid inconsistencies in the original $\lambda$-calculus. This typed calculus, also referred to as Higher-Order Logic (HOL), is the basis of the proof assistants HOL4 and HOL Light.

While $\lambda$-calculus gives an account of computation, the principle of propositions-as-types was a later development, related to the conception of \emph{intuitionistic} mathematics by Brouwer, Heyting, Kolmogorov, and others~\citep{Heyting1956}. A basic tenet of intuitionism (or \textit{constructivism}) is to only admit mathematical objects that can be mentally construed from basic principles; postulation of existence or axiomatization is not enough. At the level of logical reasoning, this leads to intuitionists rejecting certain proofs established by an appeal to the law of excluded middle (LEM)---that all propositions are either true or false. Moreover, intuitionists interpret functions as effective methods of computation rather than, say, relations that satisfy some set of equations. Consequently, many classical mathematical theorems do not hold as typically formulated with such a restricted logic. However, similar theorems turn out to be possible to prove in many cases, as shown, e.g., by \cite{Bishop1985}. Widely used proof assistants based on intuitionistic type theories, following the tradition of \cite{MartinLof1982a, MartinLof1982b}, include Coq, Agda, and Lean.

\textit{Logical frameworks}~\citep{Harper1993} support reasoning about many different logics from within a single system. 
Automath is a logical framework, as is LF (Section~\ref{sec:dsmetatheory}).
The popular general-purpose proof assistant Isabelle similarly supports many logics, as long as they can be made to conform to underlying framework for simply-typed higher-order natural deduction. While HOL is the most commonly used logic for Isabelle, other bundled logics include first-order logic with Zermelo-Fraenkel set theory, and constructive type theory~\citep{Paulson2000}.





\subsection{Proof Objects}
\label{sec:proof-objects}

\citet{Barendregt2013} characterizes proof assistants according to how they deal with \emph{proof objects}, i.e., the certificates that some property is true according to the underlying logic. In proof assistants closely related to LCF such as Isabelle, HOL4, and HOL Light, proof objects are normally not represented in full, but constructed and checked piece by piece, i.e., they are \emph{ephemeral}. In contrast, Coq and Agda produce complete proof objects, although such objects are usually not kept in memory once constructed, but are stored on disk.

The time to construct and validate proof objects (piecemeal) in Isabelle, HOL4, and HOL Light is directly proportional to the object's size. However, checking proofs in Coq, Agda, and other proof assistants that support \emph{reflection}~\citep{Boutin1997}, i.e., computational steps in proofs, may not be proportional to proof object size. In effect, one computional step can take as long as all conventional steps combined.


\citet{Barendregt2007} uses a formula of the following kind to illustrate the usefulness of reflection and its relation to proof object sizes:
$$
A \defeq p \leftrightarrow (p \leftrightarrow (p \leftrightarrow (p \leftrightarrow (p \leftrightarrow (p \leftrightarrow (p \leftrightarrow (p \leftrightarrow (p \leftrightarrow p))))))))
$$
Proving $A$ directly requires repeated and tedious use of basic derivation rules. In addition, even if rule application is automated, the proof object will be large. Instead, it is possible to perform the proof indirectly by using computation. To this end, we define:
$$
B(1) \defeq p \qquad B(n + 1) \defeq p \leftrightarrow B(n)
$$
We then prove by induction on $n$ that whenever $n \geq 1$, we have $B(2 \times n)$. We conclude the proof of $A$ by rewriting using two equalities, $A = B(10)$ and $10 = 2 \times 5$, and apply the fact about $B$ that we just proved. In proof assistants that support reflection, the final proof object contains no trace of the proofs of the two equalities, since they are established using reductions in the logic engine. In proof assistants without reflection, full proofs must be provided for the equalities before they can be used for rewriting, resulting in large proof objects. For example, proof objects in Isabelle/HOL are typically large, but this does not necessarily mean that proof checking is slower overall, since they are ephemeral~\citep{Wenzel2015}. In effect, large proof objects in Isabelle can be viewed as a consequence of deliberate design decisions, e.g., concerning how to perform rewriting, computation, and proof checking.





\subsection{Equality}
\label{sec:equality}

In logical systems and type theories, there is a conceptual difference between \emph{definitional equality}, used for proof checking (type checking), and \emph{propositional equality}, used in expressing statements to prove. 

In \emph{intensional} type theories, such as the early intuitionistic type theory by \cite{MartinLof1982a} and the Calculus of Constructions~\citep{Coquand1988} (CoC), these concepts are completely distinct, which can limit what can conveniently be proven to be equal. For example, in Coq, $0+n=n$ follows by definitional equality, while $n+0=n$ requires inductive reasoning. 
In contrast, in \emph{extensional} type theories, such as that implemented in Nuprl~\citep{Constable1986}, definitional and propositional equality coincide. However, this means that proof checking is inherently undecidable, since propositional equality can be used to specify undecidable problems. Intuitively, intensionally equal entities are such that they are ``constructed in the same way,'' while extensionally equal entities ``behave in the same way.'' Proofs in extensional systems can sometimes be translated into intensional theories after adding a few axioms~\citep{Oury2005}.

Even within intensional type theories, not all notions of propositional equality are created equal.
Homotopy type theory~\citep{univalent2013homotopy} (HoTT), for example, is an intensional type theory~\citep{nlab:intensional_type_theory} in which 
the notion of propositional equality corresponds to type equivalence.
A \textit{type equivalence} between types \lstinline{A} and \lstinline{B} is a function:
\begin{lstlisting}
f : A -> B.
\end{lstlisting}
for which there exists some function:
\begin{lstlisting}
g: B -> A.
\end{lstlisting}
that is a mutual inverse:
\begin{lstlisting}
section$\phantom{ion}$ : $\forall$ (a : A), g (f a) = a.
retraction : $\forall$ (b : B), f (g b) = b.
\end{lstlisting}
\textit{Univalence} in HoTT states that propositional
equality between types is equivalent to type equivalence between those types.
Consequentially, in HoTT, it is possible to treat equivalent types as being the same.

Both CoC and HoTT are intensional. In both of these type theories,
propositional equality corresponds to inhabitance of the identity type. However, in HoTT,
univalence provides a means of constructing a term of the identity type~\citep{escardo2018self}
that is not present in CoC. This has implications for other properties of these intensional
type theories. For example, in HoTT, as a consequence of univalence, \textit{functional extensionality} holds: functions
can be proven equal merely from the fact that they always return the same values for the same arguments.
This is not true in CoC, though it may be consistently assumed as an axiom.

There are many other weaker equalities 
than propositional equality that can be useful for reasoning about programs and systems. \cite{McBride2002} proposes a heterogenous equality relation for type theories where terms can be considered equal despite having different types. As \cite{CPDT} remarks, researchers are continually discovering new ways for entities such as functions and data to be equal.

\subsection{Predicativity}

The term \emph{predicative} was first used by \cite{Russell1906} to describe so-called propositional functions $\phi(x)$ that define a class, i.e., for which the class $\{x : \phi(x)\}$ actually exists. He distinguished such functions from \emph{impredicative} functions for which no such class exists~\citep{Feferman2005}. For example, the propositional function specifying that a class has itself as member does not define a class, and is thus impredicative. In modern logical systems, enforcing predicativity means that when objects are defined using quantifiers, no such quantifier may be instantiated with the object itself (see Chapter 12 of \cite{CPDT}).

Isabelle's meta-logic stays within predicative simple type theory~\citep{Paulson2000}.
In Coq, the \lstinline{Type} universe (including \lstinline{Set}) is predicative, but \lstinline{Prop} is impredicative.
Including both of these provides a balance to users of consistency with common axioms and expressivity:
Impredicative \lstinline{Type} with \textit{large elimination} (pattern matching that returns terms of type \lstinline{Type})
would not be consistent with LEM, which can be added as an axiom in Coq~\citep{predicavity1}.
On the other hand, there is an informal
consensus that the impredicativity of \lstinline{Prop} adds expressivity which is useful for expressing
most mathematical proofs~\citep{predicavity1}.
Otherwise, in predicative logic, some proofs are more complex~\citep{predicativity2},
though it is not known to what extent this has practical implications on what it is possible to express in each logic.
Thus, Coq includes an impredicative \lstinline{Prop} universe in which large elimination is disabled.

\subsection{Definitional Mechanisms}
\label{sec:definitional}
Programs of interest to computer scientists and engineers often involve classic data structures such as lists, trees, and natural numbers. These data structures can be described, e.g., by initial algebras in category theory or in fixpoint theory~\citep{Scott1970}. Most proof assistants provide mechanisms for defining such data structures; these mechanisms take one of three forms~\citep{Berghofer1999}: (1)~axiomatic, (2)~inherent, and (3)~definitional.

In the first approach, taken by early users of the LCF system, datatype constructors are defined by introducing new axioms, from which induction principles are proved~\citep{Paulson1984}. In the second approach, the underlying logic is extended to support custom datatypes, which requires metatheoretic investigation, e.g., as carried out by \cite{Coquand1990} for inductive types in CoC, and then implemented in Coq by \cite{PaulinMohring1993}. In the third approach, datatype support is added on top of already existing mechanisms; this is done by \cite{Pfenning1990} for the CoC and by \cite{Berghofer1999} for HOL. Church's classic encoding of numbers as functions repeatedly applying an argument function in the $\lambda$-calculus may be considered an example of the definitional approach.

Initial support for datatypes in proof assistants only included inductive datatypes, i.e., the minimal solutions to fixpoint equations. Inductive datatypes are arguably the most important, since they facilitate proofs by the fundamental technique of structural induction~\citep{Burstall1969,harper2016practical}. However, some applications require coinductive datatypes (maximal solutions), which can be accounted for in most type theories~\citep{Coquand1994}. \citet{Gimenez1995} initially implemented support for coinductive and corecursive functions in Coq, while \citet{Paulson1997} did the same for Isabelle/HOL.

A long-standing issue in proof assistants is developing mechanisms for \emph{quotient types}, which are defined by dividing members of an existing type into equivalence classes. Quotients are widely used in mathematical reasoning, in particular in algebra. An initial approach to quotients in Isabelle/HOL was proposed by \cite{Slotosch1997}, with later alternatives by \cite{Paulson2006} and \cite{Huffman2013}. \cite{Cohen2013b} proposed an approach to quotient types in Coq which elides the conventional approach using \emph{setoids}~\citep{GEUVERS2002271} that significantly restricts the scope of rewriting tactics.

The dependently-typed language Cedille~\citep{stump2017calculus} makes it possible to define induction principles
in a language based on Church-encodings~\citep{Church1941}, and encodes datatypes in terms of induction principles. 
This allows for, among other things, zero-cost reuse of functions and proofs across certain datatypes (Section~\ref{sec:proof-reuse}).

Research on definitional mechanisms is still an active topic.
\citet{Sozeau2010} designed and implemented a Coq extension for defining functions equationally which compiles definitions down to eliminators for inductive types; this extensions was used for the function \lstinline{acc} in Chapter~\ref{ch:ex}. \citet{Biendarra2017} presented a redesigned Isabelle/HOL library, following the definitional approach, for writing and reasoning about inductive and coinductive datatypes. While Coq and Agda inherently allow \emph{nonuniform} datatypes, i.e., recursive types whose arguments vary recursively, HOL systems did not support them until the advent of this library, which reduces such definitions to uniform counterparts.

\subsection{Totality of Functions and Termination}
\label{sec:totality-termination}
The logic of Church's simply-typed $\lambda$-calculus, HOL, is a logic of total functions. This means that partial functions cannot be directly described in proof assistants based on HOL, such as Isabelle/HOL. Similarly, the Calculus of Inductive Constructions (CIC), which Coq is based on, supports only total functions. Partial functions can still be indirectly encoded in CIC and HOL, for example by (a) returning values in a monad~\citep{McBride2015}, such as the coinductive delay monad described by \cite{Capretta2005}, (b) requiring proofs of argument value subset membership as function arguments~\citep{CoqArt}, (c) letting functions return values in the option type, or (d) capturing functions as inductive relations between input and output.

Functions in proof assistants based on intuitionistic type theories like CIC need to be terminating for the sake of consistency. When a function is defined in Coq, for example, termination is automatically proven for cases where functions recurse on a subterm of the input and in other simple cases; in more advanced cases, users must manually prove termination or rely on approaches that use, e.g., \emph{sizes} of argument terms~\citep{Abel2017}. In contrast, functions in HOL are not required to be computable at the outset, and thus do not need to be accompanied by termination witnesses. On the other hand, the uses of functions with unproven termination are somewhat limited.

Requirements for totality and termination are two hurdles that new users of ITPs face. They are constraints even to users familiar with functional programming languages, where no such requirements are typically imposed. For certain functions, arguing and formally proving termination may not even be a key concern. 
In that spirit, Zombie~\citep{Casinghino2014} separates out a logical, terminating fragment from a programmatic,
possibly non-terminating fragment, that way the programmer can move freely between those fragments.

In other languages, a common technique to encode such functions in a total setting is to define a ``fuel'' argument, such as a natural number. Either the fuel argument is empty (0) and the function terminates, or there is enough fuel to continue to, e.g., perform recursive calls. This allows 
for proving termination by a simple structural argument on the fuel type. When calling the function in some other context, passing ``infinite fuel'' may be possible, which implicitly trusts that the function always terminates. \cite{Jourdan2012} provide a detailed description of the fuel technique in the context of a verified parsing function on a potentially infinite stream of tokens in Coq.




 
\section{Trusted Computing Bases of Proofs and Programs}
\label{sec:trustedbases}

The concept of a Trusted Computing Base (TCB) was introduced by \cite{Rushby1981} in the context of security of computer systems. The basic idea is that the security of a system may be reduced to the security of a proper subset of all system components. If these components behave as expected, the system as a whole is secure. For verified software, security is replaced with correctness, e.g., functional correctness. \cite{Kumar2015} divides the TCB into the following categories:
\begin{itemize}
\item formal models of system components (e.g., model of a processor);
\item system components for which there are no explicit formal models (e.g., linker or operating system);
\item tools used to check proofs about the system (e.g., proof assistant).
\end{itemize}

\subsection{TCB of Proofs}
\label{sec:trust-proofs}

\cite{Pollack1998} considers the question of how to trust specific machine-checked proofs, and by extension, programs that such proofs pertain to. He divides the question into a purely formal part---whether a provided proof is derivable in a given formal system---and an informal part that asks whether the proof has a purported meaning as expressed outside any formal system. 

As to the formal part, trusting the proof can be reduced, by computer, to trusting the (implementation of) the proof checker of the formal system; the source code for such proof checkers can be compact and readable. However, Pollack argues that a complicated semantics of the checker's implementation language can still provide serious obstacle to trust, and proposes that the language itself should be a logical framework designed to represent formal systems, such as LF~\citep{Harper1993} or Isabelle~\citep{Paulson1994}. As to the informal part, Pollack points to that understanding specific pieces of mathematics relies on acceptance of previous mathematics, whose trust may be partly due to its wide acceptance. 

Based on Pollack's investigation, \cite{Wiedijk2012} defined the notion of \emph{Pollack-inconsistency}, which is expressed in terms of the mechanisms a proof checker uses to print and parse its formulas (which are what the user ultimately must intepret informally). In particular, Wiedijk argues that a system should always be able to parse formulas it outputs. He then demonstrates that current proof assistants are Pollack-inconsistent to some extent, but outlines how this can be addressed by modifying the implementations of printing and parsing.

With the goal of determining how small a trusted proof checker can be for a practical application,
\cite{Appel2003} attempted to minimize the size of a proof checker for proof-carrying machine code.
The result was less than 2700 lines of code.
The Lean theorem prover attempted to minimize the size of the proof-checking kernel from the start~\citep{deMoura2015}.

\cite{coqincoq} encoded a limited version of the formal system underlying Coq in Coq itself, and proved strong normalization of its type system. \cite{Barras2010} addressed the problem of providing set-theoretical models of CoC, the logic that underpins Coq, with the ultimate goal of ensuring that Coq's theory is consistent with the theory implicitly or explicitly assumed by most mathematicians. \cite{Anand2014} encoded and verified the foundations of the Nuprl proof assistant in Coq. \cite{Davis2015} certified the Milawa theorem prover. \cite{Kuncar2018} proved the relative consistency of extensions made to the foundations of Isabelle/HOL. 
\cite{Anand2018} encoded Coq's internal data structures in Coq itself and gave a semantics for type checking, leading up to the MetaCoq project~\citep{Sozeau2019} for building verified checking and extraction for Coq.

\subsection{TCB of Programs}
\label{sec:trust-programs}

Coq and other similar proof assistants contain logic engines that can execute functions that have been verified. However, execution inside such a logic engine is generally slow compared to execution of native functional programs~\citep{Leroy2015}, and does not directly support handling of input and output. Instead, to obtain practical verified programs, proof assistant users rely on mechanisms such as \emph{program extraction} to produce programs that can be integrated into larger systems or executed in conventional runtime environments. However, these mechanisms may increase the trusted base of verified programs.

\paragraph{Program Extraction}

\citet{Paulin1989,Paulin-Mohring1989} proposed realizability for CoC to $F_\omega$. To obtain practical executable programs from Coq functions, \cite{Paulin1993} extended the realizability from $F_\omega$ to ML. \cite{Letouzey2003,Letouzey2004} later introduced a new extraction mechanism for Coq which removed several restrictions. This introduced an intermediate language called MiniML, which can be translated to OCaml, Haskell, and Scheme. The new mechanism, argued correct by a conventional proof, was evaluated on several large projects~\citep{Berger2005,CruzFilipe2006,Letouzey2008}.

\cite{Berghofer2002} identified a subset of HOL that can be translated to practical functional languages and implemented code generation for Isabelle/HOL. \cite{Haftmann2010} proposed a redesign of the code generation mechanism in Isabelle/HOL; their approach is based on translating HOL to an intermediate language called Mini-Haskell, and then further to Standard ML, Haskell, and OCaml. The correctness argument is reminiscent of that for Coq's extraction mechanism. \cite{Haftmann2013} proposed a data refinement (Section~\ref{sec:refinement}) framework which replaces abstract datatypes with concrete ones, which widens the scope of code generation.



\paragraph{Beyond Extraction and Code Generation}

Practical functional programming languages such as OCaml and Haskell lack a fully formal (and machine-checked) semantics. Proof assistant users who want to avoid trusting extraction may use \emph{deep embeddings} (Section~\ref{sec:embedding}) of target practical programming languages along with language semantics. These embeddings and semantics can then be used in certified compilers (Section~\ref{sec:verified-compilers}), which may include formal models of system components such as processors, to produce verified machine code. However, these approaches may have more restrictions and inconveniences than extraction.
Removing or circumventing these restrictions may be fruitful. Continuing to develop and improve certified compilers for Coq like \OE uf and 
CertiCoq, for example, may help proof engineers circumvent extraction to OCaml and Haskell altogether.
Instead, proof engineers may be able to directly compile certified programs to machine or assembly code and run those programs directly.

%




\chapter[Between the Engineer and the Kernel: Languages and\\ Automation]{Between the Engineer and the Kernel: Languages and Automation}
\label{sec:prooflangs}

At the core of every ITP that meets the de Bruijn criterion
is a proof object that a small kernel can check. Directly constructing this proof object
may be too low-level to provide for a positive user experience, and in some cases the proof object may not
be exposed to the end-user at all. Typically, several layers of languages and automation
act as an interface between the proof engineer and the kernel.


Automation is all about finding a \textit{proof}---what the proof object (Section~\ref{sec:proof-objects}) is can differ by ITP.
Because of this, and because of the different philosophies and needs of different communities,
ITPs have different approaches to automation.
We consider two commonly used ITPs as examples: Coq and Isabelle/HOL.

In Coq, a proof is a term in the language Gallina, which the kernel type-checks.
The Coq proof engineer can write proofs in Gallina directly, but it is common
to write proofs using high-level \textit{tactics}, or proof search procedures. 
In Coq, these tactics search for and ultimately produce a Gallina term, which the kernel then type-checks.
Users can combine existing tactics, or write their own, either in the general-purpose programming language
OCaml or in the tactic language Ltac. 

In contrast, in Isabelle/HOL, a proof is a term in the language ML with
a specific type, combined with the guarantee that this term must have been
produced using the axioms of the logic.\footnote{It is possible to explicitly produce proof terms that can be checked
with a small kernel using Isabelle/HOL-Proofs, but it is not common to explicitly do so.}
It is not possible to directly write a term of this type, since correctness hinges on the guarantees
about its construction. Instead, users commonly write proofs in Isabelle/Isar, which is a high-level \textit{proof language}:
a language for structuring and composing propositions, facts, and proof goals. 

Ltac and Isabelle/Isar are examples of two languages which embody different styles of automation.
This chapter discusses these and other styles of automation (Section~\ref{sec:languagesforauto}),
then concludes with a discussion of automation in practice (Section~\ref{sec:specializedauto}) built in these different styles.






\section{Styles of Automation}
\label{sec:languagesforauto}

Proof engineers can construct proofs of theorems in a wide variety of ways.
There are three common styles of proof automation:
writing sequences of proof search \textit{tactics} (which can be defined
either using a \textit{metalanguage} like Standard ML or a specialized
\textit{tactic language} like Ltac), writing high-level programs in a structured \textit{proof language},
and using \textit{reflection} to write proof-checking procedures within the host language itself.
When executed, all expressions or commands reduce to primitive
inference rule applications in the proof-checking kernel.

Ltac and Isabelle/Isar are examples of a tactic language and a proof language, respectively.
Some proof assistants (for example, both Coq and Isabelle) have
support for all three styles of automation, either natively or through
extensions. However, not all do; Agda, for example,
takes a minimilistic approach,
supporting only reflection (it is possible to imitate tactics
using reflection). Table~\ref{tab:styles}
references examples of supported styles of automation for a sample of major proof assistants.

\begin{table}
  \begin{footnotesize}
\centering
\begin{tabular}{ |l|l|l|l|l| }
\hline
   & \multicolumn{2}{|c|}{\textbf{Tactics}} & &  \\ \hline
   & \textbf{Metalanguage} & \textbf{Tactic Lang.} & \textbf{Proof Lang.} & \textbf{Reflection} \\ \hline
  \textbf{Agda} & & & & Possible \\\hline
  \textbf{Coq} & OCaml & Ltac & \ssreflect & Possible \\\hline
  \textbf{Isabelle} & ML & Eisbach & Isar & Possible \\\hline
  \textbf{Nuprl} & ML & & & Possible \\
\hline
\end{tabular}
\caption{Examples of styles of automation a sample of major proof assistants support, including some external developments.}
\label{tab:styles}
\end{footnotesize}
\end{table} 

These styles often merge,
and the lines between them can be blurry. It is possible, for example,
to write proofs in the high-level proof language SSReflect in Coq, or to
combine this language with Ltac tactics. It is possible to write ML tactics in Isabelle/HOL, and to write tactic-style proofs in
Isabelle/HOL that look similar to Ltac proofs in Coq.

This section explores the design space and uses in common proof assistants of languages for different styles of automation:
Tactics and tactic languages (Section~\ref{sec:tactics}), proof languages (Section~\ref{sec:structuredprooflangs}), and
reflection (Section~\ref{sec:refl}).
It then concludes with a discussion of future styles of automation (Section~\ref{sec:futurestyle}).

\subsection{Tactics \& Tactic Languages}
\label{sec:tactics}

LCF introduced the language ML (\textit{metalanguage}) to let users
write high-level proof automation~\citep{Gordon1978}.
In LCF, theorems are represented using the abstract type \lstinline{thm} in ML;
the only way to construct an inhabitant of \lstinline{thm} is using the axioms and inference rules of the logic.
A proof in LCF has the following type in ML:
\begin{lstlisting}
type proof = thm list $\rightarrow$ thm
\end{lstlisting}
In other words, a proof is a function that takes a list of hypotheses and, from them, proves the conclusion.

A basic unit in LCF proof automation is the \lstinline{tactic}:
\begin{lstlisting}
type tactic = goal $\rightarrow$ (goal list $\times$ proof)
\end{lstlisting}
The \lstinline{goal} type (not shown) represents proof goals.
Thus, a tactic is a function that takes a proof goal and then produces a list
of new goals which the goal reduces to; such goals are conventionally called \emph{subgoals}.
When no more goals remain, the tactic produces a value of type \lstinline{proof}.

Based on this tactic definition, it is possible to define higher-order functions that take tactics as arguments
and return new tactics. Milner called such functions \emph{tacticals}. For example, a collection
of tactic combinators may include the tactical \lstinline{repeat} with type \lstinline{tactic $\rightarrow$ tactic},
which repeatedly applies its argument tactic to the proof goal. In Coq, the composition tactical \lstinline{t1; t2}
runs the first tactic \lstinline{t1}, then runs the second tactic on all goals
produced by the first tactic \lstinline{t2}.


However, the LCF representation of tactics means that, by definition, the outputted goals are mutually independent---a proof for one
is unrelated to proofs for others. In practice, constraints during proof
search can apply across subgoals~\citep{SpiwackTactic}. For this reason, Coq previously used, up to at least version 8.3 in 2010~\citep{Spiwack2010},
a tactic type definition along the following lines: 
\begin{lstlisting}
type proof = thm list $\rightarrow$ thm
type tactic = goal $\times$ state $\rightarrow$ (goal list $\times$ state $\times$ proof)
\end{lstlisting}
Here, the state returned from a tactic call can be used to figure out dependencies between subgoals, such as shared variables.

In the early days of proof assistants, users combined custom tactics written in an ML dialect with built-in tactics and tactical combinators to write custom 
automation~\citep{Constable1986, cornes1995coq, paulson1988preliminary, paulson1983tactics}. This tradition for programming proof automation is the default workflow in the HOL4 and HOL Light proof assistants, and remains a possibility in Coq and Isabelle. However, Coq and Isabelle also support writing tactics in tactic languages rather than in the tool's implementation language.

\subsubsection{Tactic-Based Proofs}

The original proof development workflow in LCF was to write sequences of tactic calls until no proof goals were left. This style is still prevalent in modern proof assistants.
Consider an inductive proof of the theorem \lstinline{app_nil_r} in Coq,
which states that appending the empty list to any list produces the original
list. We can write this using tactics and tacticals:
\begin{lstlisting}
Theorem app_nil_r : forall (A : Type) (l : list A), l ++ [] = l.
Proof.
  intros. induction l; auto. simpl. rewrite IHl. auto.
Qed.
\end{lstlisting}
Executing these tactics produces a Gallina proof term:
\begin{lstlisting}
(fun (A : Type) (l : list A) =>                         (* hypotheses *)
  list_ind                                              (* induction principle for lists *)
    (fun (l$_0$ : list A) => l$_0$ ++ [] = l$_0$)                (* motive to prove *)
    eq_refl                                             (* base case (by reflexivity) *)
    (fun (a : A) (l$_0$ : list A) (IHl : l$_0$ ++ [] = l$_0$) => 
      (* inductive case (by rewriting) *)
      eq_ind_r (fun (l$_1$ : list A) => a :: l$_1$ = a :: l$_0$) eq_refl IHl)
    l)                                                  (* argument to induction *)
: forall (A : Type) (l : list A), l ++ [] = l.          (* theorem type *)
\end{lstlisting}

There is an analogous in the Isabelle/HOL standard library:
\begin{lstlisting}
lemma append_Nil2 : "append xs [] = xs"
by (induct xs) auto
\end{lstlisting}
While it is not typical to do so, using Isabelle/HOL-Proofs, it is possible to reconstruct and inspect a proof object
for this proof.

\subsubsection{Tactic Languages}
\label{sec:tacticlangs}


Tactic languages allow proof engineers to write custom tactics alongside specifications and proofs, rather than in an implementation language such as Standard ML.
We describe two tactic languages---Ltac for Coq and Eisbach for Isabelle---in detail,
then conclude with a brief discussion of other tactic languages. 

\paragraph{Ltac}
Nearly 20 years ago, Coq introduced the Ltac tactic language~\citep{Delahaye2000},
which has since become the standard for tactic development in Coq.
Ltac is an untyped domain-specific language with support for pattern
matching on terms and goals, as well as writing custom tactics and tacticals.
The Ltac manual can be found in the Coq documentation~\citep{ltac}.

To understand Ltac, consider the \lstinline{break_match} tactic from the StructTact~\citep{structtact} library:
\begin{lstlisting}
Ltac break_match := break_match_goal || break_match_hyp.
\end{lstlisting}
This tactic breaks down match statements, both in goals:
\begin{lstlisting}
Ltac break_match_goal := match goal with
  | [ |- context [ match ?X with _ => _ end ] ] =>
    match type of X with
      | sumbool _ _ => destruct X
      | _ => destruct X eqn:?
    end
end.
\end{lstlisting}
and in hypotheses:
\begin{lstlisting}
Ltac break_match_hyp := match goal with
  | [ H : context [ match ?X with _ => _ end ] |- _] =>
    match type of X with
      | sumbool _ _ => destruct X
      | _ => destruct X eqn:?
    end
end.
\end{lstlisting}
Both tactics perform syntactic pattern matching
over the goal, using the syntax \lstinline{name : cpattern |- cpattern}, where the 
left \lstinline{cpattern} represents hypotheses and the right \lstinline{cpattern} represents the conclusion.
In \lstinline{break_match_goal}, pattern matching looks only in the conclusion;
in \lstinline{break_match_hyp}, pattern matching looks only in the hypotheses.
Both tactics use the \lstinline{context} syntax to find all subterms of the term that are \lstinline{match} statements,
then pattern match on the type of the result, using the \lstinline{destruct} tactic to break down those \lstinline{match} statements.

The effect of \lstinline{break_match} is to simplify tedious but conceptually simple proofs by case analysis,
and to do so without relying on the names of hypotheses, which can make proofs likely to break as specifications change~\citep{Woos2016}. 
Consider, for example, an interpreter correctness proof:
\begin{lstlisting}
Lemma interp_eval : forall op v v', interp op v = Some v' -> eval op v v'.
Proof.
  unfold interp. intros. destruct op; destruct v; 
  try discriminate; inversion H; constructor.
Qed.
\end{lstlisting}
Here, the \lstinline{intros} tactic introduces hypotheses named \lstinline{op}, \lstinline{v}, \lstinline{v'}, and \lstinline{H}.
The \lstinline{destruct op; destruct v} sequence of tactics then does case analysis on \lstinline{op} and \lstinline{v}.
We can use \lstinline{repeat break_match} instead of \lstinline{destruct op; destruct v} to simplify this proof:
\begin{lstlisting}
Proof.
  unfold interp. intros. repeat break_match;
  try discriminate; inversion H; constructor.
Qed.
\end{lstlisting}
The resulting proof is more concise. It is also less likely to break as specifications change,
since it does not depend on the automatically generated names \lstinline{op} and \lstinline{v}, which may later change
(Section~\ref{sec:des-scale}).

Ltac was designed to
achieve a balance between the flexibility of writing tactics in a powerful language like OCaml,
and the ease of flexibility from using built-in combinators to write custom tactics
directly inside of the Coq proof assistant.
Like ML and OCaml, it is Turing-complete (see Chapter 16 of \cite{CPDT}). However, it gives proof engineers
limited access to underlying features like environment management. This means that
proof engineers do not have to deal with low-level issues in OCaml such as managing de Bruijn indexes;
proof engineers who want that level of control may write plugins in OCaml.

Ltac2~\citep{ltac2}, the next generation of Ltac, is in development.
It comes full circle, returning to the ML family of languages. 

\paragraph{Eisbach}

The tactic language Eisbach~\citep{Matichuk2015EisbachAP} for Isabelle
was inspired by Ltac. Eisbach is tactic language that is integrated
into the proof language Isabelle/Isar. Using Eisbach, proof engineers 
write tactics (called \textit{proof methods}) directly in Isar syntax.
The Eisbach manual can be found in \cite{matichuk2015eisbach}.

To understand Eisbach, consider an example proof method from the Eisbach manual for solving existentials:
\begin{lstlisting}
method solve_ex =
  (match conclusion in $\exists$ x. Q x for Q $\Rightarrow$
    <match premises in U : Q y for y $\Rightarrow$
      <rule exI [where P = Q and x = y, OF U]>>)
\end{lstlisting}
This matches the current conclusion with the \lstinline{Q} in $\exists$ \lstinline{x. Q x},
then looks in the hypotheses for a term that matches \lstinline{Q y} for the matched \lstinline{Q} and some \lstinline{y},
and then calls the introduction rule for existentials with that hypothesis. 
In other words,
if some hypothesis is \lstinline{Q y} and the goal is $\exists$ \lstinline{x. Q x},
then \lstinline{solve_ex} will automatically prove the goal from the hypothesis.
The manual shows an example of calling this proof method to solve a proof:
\begin{lstlisting}
lemma halts p $\Rightarrow$ $\exists$ x . halts x
  by solve_ex
\end{lstlisting}

Like Ltac, Eisbach provides limited access to low-level details.
Isabelle proof engineers who want access to these details may write
proof methods directly in Isabelle/ML. 

%

\paragraph{Other Tactic Languages}


The untyped nature of languages like Ltac may make it difficult to
debug custom automation and to provide strong guarantees on custom
tactics. Typed tactic languages such as those of the
Delphin~\citep{poswolsky2009system} and Beluga~\citep{pientka2008programming} logical frameworks, 
the tactic language for VeriML~\citep{Stampoulis2010}, and
Mtac~\citep{Ziliani2015} and Mtac2~\citep{Kaiser2018} in Coq
address these problems.

PSGraph~\citep{lin2016understanding} is a graphical tactic language which aims to make
debugging and refactoring of tactics easier. In PSGraph, tactics are flow graphs for proof subgoals,
and executing tactics amounts to following the flow graph directly. 

The Matita~\citep{asperti2007user} proof assistant introduces
a language of \textit{tinycals} to address some challenges that tacticals
pose for user interaction. In particular, proof assistants traditionally execute tacticals
atomically. This makes it difficult to communicate how the tactical is executed to the user,
as well as to debug tacticals and provide useful error messaging when they fail.
Tinycals, in contrast, act like traditional tacticals, except that they allow
for more fine-grained execution.

Some tactic languages merge the tactic approach with metaprogramming constructs,
so that users can write tactics in the proof language itself.
In Idris~\citep{brady2013idris}, it is possible to implement tactics using the elaboration monad, which exposes
elaboration to users and enables metaprogramming within Idris~\citep{Christiansen2016}.
Agda recently replaced its reflection mechanism
with one based on Idris' elaborator reflection~\citep{agda-elab}.

Similarly, Lean exposes several metaprogramming constructs (including a tactic monad),
which enable users to write tactics in Lean itself, to access proof state, and to access
internal methods in the underlying C++ codebase~\citep{Ebner2017};
this enables proof engineers to write powerful procedures.

Some tools merge the approach of proof by reflection
with a tactic language; we discuss these in Section~\ref{sec:refl}.

\subsection{Proof Languages}
\label{sec:structuredprooflangs}

A proof language is a mechanism for structuring and composing
propositions, facts, and proof goals. Proof languages generally allow both
\emph{backward} reasoning, going from the proof goal to a new set of goals,
and \emph{forward} reasoning, where new facts are added to the proof context
but the goal remains the same. When formulating reasoning steps in a proof language,
procedures in lower-level languages, such as tactics, can be invoked explicitly
or implicitly. In turn, these procedures can invoke specialized external proof search
programs. In contrast to plain unstructured sequences of commands (``tactic soups''), proofs written in
proof languages are usually meant to convey key proof ideas, i.e., to be understandable by humans.
To ensure readability, proof languages take inspiration from
traditional mathematical vernacular. We consider two proof languages---Coq/\ssreflect and Isabelle/Isar---in detail, and then briefly discuss other proof languages.

\paragraph{Coq/\ssreflect}

Coq/\ssreflect is a proof language for Coq that emphasizes reasoning by rewriting using equalities and proofs by computation via small-scale reflection~\citep{Gonthier2010}. It was originally developed by Gonthier in the context of his proof of the four-color theorem~\citep{Gonthier2008}. Idiomatic proofs in Coq/\ssreflect make use of ``bullets'' (\lstinline{-}, \lstinline{*}, or \lstinline{+}) to structure proofs similarly to how Isabelle/Isar uses indentation of blocks to indicate structure.

The basis of Coq/\ssreflect is that a step in a proof is one of the following:
\begin{itemize}
\item a \emph{deduction step} that directly constructs parts of a proof, either by backwards or forwards reasoning;
\item a \emph{bookkeeping step} that performs a management operation on the proof context, e.g., introducing or renaming assumptions;
\item a \emph{rewriting step} that changes parts of the proof goal or some assumption, either by way of some equality lemma or by computation.
\end{itemize}
In idiomatic proofs, these kinds of steps are interleaved and tend to be used in equal proportions.

The \emph{reflect} part of \ssreflect refers to the convention of performing deduction steps that translate between symbolic representations, such as expressions involving boolean functions, and logical representations, such as inductive predicates. For example, when a proof goal can be solved by reasoning in propositional logic, we can convert the proof goal to boolean form and perform computation in Coq's logic engine instead of applying multiple propositional derivations manually. In contrast to general proofs by reflection (Section~\ref{sec:refl}), which focus on computational efficiency and managing the sizes of proof objects, \ssreflect leverages reflection for convenience and user productivity. For example, a conjunct can be reflected to a boolean value that directly computes to \lstinline{true}, saving the manual effort of applying tactics.

The proof language contains special syntax for translating between representations, which is called application of \emph{view lemmas}. Moreover, Coq users traditionally use different commands for rewriting, definition expansion, and partial evaluation. In Coq/\ssreflect, all of these tasks are performed via parameters to the \lstinline{rewrite} tactic. Rewriting operations can pinpoint specific subterms in the current proof goal through the use of \emph{pattern expressions}~\citep{Gonthier2012} that may mention some constant names for disambiguation but leave others implicit.

Consider the following lemma from the Mathematical Components project whose proof is written in idiomatic \ssreflect:
\begin{lstlisting}[language=Coq]
Lemma edivnP : forall m d, edivn_spec m d (edivn m d).
Proof.
rewrite /edivn => m [|d] //=; rewrite -{1}[m]/(0 * d.+1 + m).
elim: m {-2}m 0 (leqnn m) => [|n IHn] [|m] q //=; rewrite ltnS => le_mn.
rewrite subn_if_gt; case: (ltnP m d) => [// | le_dm].
rewrite -{1}(subnK le_dm) -addSn addnA -mulSnr; apply: IHn.
apply: leq_trans le_mn; exact: leq_subr.
Qed.
\end{lstlisting}
Here, the first rewrite unfolds the \lstinline{edivn} definition, while the last rewrite performs chained rewriting using facts arithmetic, such as that addition is associative (\lstinline{addnA}). Explicit names for quantified variables are given after the operator \lstinline{=>}, which reduces the chance of brittleness due to reliance of machine-generated variable names. The tactic \lstinline{elim} performs induction on the natural number \lstinline{m}.

\paragraph{Isabelle/Isar}

Isabelle/Isar is a proof language for Isabelle that aims for human readability while retaining some symbolism of formal deduction systems~\citep{Wenzel2007isar}. It is built on the Isabelle/Pure logic, which is an intuitionistic fragment of HOL. Isabelle/Isar can be understood as an interpreter for block-structured syntax capturing the flow of facts and proof goals~\citep{Wenzel2006}.

As an example, the Isabelle/Isar manual~\citep{wenzel2004isabelle} contains a definition of a group
that assumes only a left identity element, along with the following proof
that for any group, the identity element of the group is a right identity (we expand the 
\lstinline{term} $\ldots$ notation from the version in the reference manual for clarity):
\begin{lstlisting}
theorem right_unit : x $\circ$ 1 = x
proof -
  have 1 = x$^{-1}$ $\circ$ x by (rule left_inv [symmetric])
  also have x $\circ$ (x$^{-1}$ $\circ$ x) = (x $\circ$ x$^{-1}$) $\circ$ x by (rule assoc [symmetric])
  also have x $\circ$ x$^{-1}$ = 1 by (rule right_inv)
  also have 1 $\circ$ x = x by (rule left_unit)
  finally show x $\circ$ 1 = x.
qed
\end{lstlisting}
Translated directly into English, we can think of this as the following proof (with implicit symmetry of equality):
\begin{theorem}[Right unit]
$x \circ 1 = x$
\end{theorem}
\begin{proof}
By left inverse, $1 = x^{-1} \circ x$. By associativity, $x \circ (x^{-1} \circ x) = (x \circ x^{-1}) \circ x$.
By right inverse, $x \circ x^{-1} = 1$. By left unit, $1 \circ x = x$. Then by the above, $x \circ 1 = x$.
\end{proof}
Isabelle/Isar completes and checks this proof much like a human reader would, by making all of the appropriate substitutions.
We could render an alternate English proof with all of these substitutions explicit, rather than leaving them to the reader:
\begin{proof}
We can write $x \circ 1 = x$ as $x \circ (x^{-1} \circ x) = x$ by left inverse,
which is $(x \circ x^{-1}) \circ x = x$ by associativity,
which is $1 \circ x = x$ by right inverse,
which holds by left unit.
Thus, $x \circ 1 = x$.
\end{proof}
At this point, we are much closer to proof by a sequence of commands. The key difference for readability is that the intermediate goals are explicit
in the proof language, so to reconstruct an English proof, the reader does not need to step through tactics one-by-one and track the transformation of the goal.



\paragraph{Other Proof Languages}

Mathematically-inclined proof languages such as Isar for Isabelle and Czar for Coq~\citep{Corbineau2008} were influenced by the language of the Mizar proof system~\citep{trybulec1985computer}, which tilts more towards natural language than logical symbolism for representing deduction steps. Many proof assistants following the LCF tradition now have Mizar modes; for example, Mizar modes have been implemented in HOL~\citep{harrison-mizar} and HOL light~\citep{wiedijk2001mizar}, as well as in Coq~\citep{giero2003mmode}.

Building a proof system specifically for human-readability means that there is less of a barrier for humans and computers to check the same proofs.
Following in this spirit, the Formalized Mathematics journal consists entirely of mathematical properties and proofs in 
Mizar that are automatically translated into English and generated as PDFs, such as the properties of sets~\citep{darmochwal1990finite},
naturals~\citep{bancerek1990fundamental}, and reals~\citep{JFR1411}.

While most proof languages were designed with mathematics in mind, their use has not been confined to mathematics.
For example, using
Isar 
is recommended style for submission to the Isabelle Archive of Formal Proofs~\citep{isabelleafp},
which consists of more computer science than mathematics formalizations~\citep{Blanchette2015}.

The language PSL~\citep{Nagashima2017} for Isabelle/HOL allows expressing high-level \emph{proof strategies}. PSL generates efficient Isar proof scripts
from user-written strategies.


\subsection{Proofs by Reflection}
\label{sec:refl}

Writing proofs by reflection, that is, calling certified procedures within the host language itself~\citep{Allen1990}, can be viewed as an alternative to writing proofs in a proof language or using tactics. 

Chapter 15 of CPDT~\citep{CPDT} illustrates this style of proof in Coq and demonstrates its benefits on a proof that a natural number is even;
we present a slightly modified version of that example that is self-contained.
Given some inductive predicate for evenness:
\begin{lstlisting}
Inductive isEven : nat -> Prop :=
| Even_O : isEven O
| Even_SS : forall n, isEven n -> isEven (S (S n)).
\end{lstlisting}
we construct a verified function to check evenness:
\begin{lstlisting}
Fixpoint check_even (n : nat) : option (isEven n) := match n with
| 0 => Some Even_O
| 1 => None
| S (S n') => 
    match check_even n' with
    | Some p => Some (Even_SS n' p)
    | _ => None
    end
end.
\end{lstlisting}
For a given \lstinline{n}, \lstinline{check_even} returns an optional a proof of \lstinline{isEven n} (\lstinline{None} when \lstinline{n}
is not even). As CPDT notes, this type signature guarantees that it only returns a proof when \lstinline{n}
actually is even. 

Our goal is to write a tactic that uses \lstinline{check_even} to prove evenness.
To write this tactic, we need to extract the proof from the \lstinline{option} type above 
when possible. We define a dependently-typed function \lstinline{optionOut}
that does this:
\begin{lstlisting}
Definition optionOutType (P : Prop) (o : option P) := 
  match o with
  | Some _ => P
  | _ => True
  end.

Definition optionOut (P : Prop) (o : option P) : optionOutType P o := 
  match o with
  | Some pf => pf
  | _ => I
  end.
\end{lstlisting}
We then write a tactic that extracts the proof that \lstinline{check_even} returns:
\begin{lstlisting}
Ltac prove_even_reflective := 
  match goal with
  | [ |- isEven ?N] => exact (optionOut (isEven N) (check_even N))
  end.
\end{lstlisting}
With this, the following proof goes through:
\begin{lstlisting}
Theorem even_256 : isEven 256.
Proof.
  prove_even_reflective.
Qed.
\end{lstlisting}
and similarly for any other even number; the tactic fails (as expected) for odd numbers.
As CPDT notes, the size of the resulting proof term is manageable even for large numbers.
This is a particular advantage of this style of proof.

The concept of computational reflection predates ITPs;
an early history of reflection can be found in \cite{Demers95reflectionin},
and an early history of its use in theorem proving can be found in \cite{harrison-reflection}.
Accordingly, it is one of the oldest styles of proof automation.
Its use in modern proof assistants with support for higher-order logics can be traced back to the 1990s, 
starting with a proof of the existence of this class of proofs in Nuprl~\citep{Allen1990},
and following soon after in other ITPs such as LEGO~\citep{Pollack1994} and Coq~\citep{Boutin1997}; 
in Coq, this approach predates Ltac~\citep{Delahaye2000}. 

Idris recently replaced its specialized tactic language with a mechanism for reflection called \textit{elaborator reflection}~\citep{Christiansen2016}.
This mechanism exposes Idris' elaborator directly to the programmer, 
which allows for powerful proof automation. For example, \cite{Christiansen2016} demonstrates how to use
elaborator reflection to write a \lstinline{mush} tactic, which can be used to dispatch many goals in the Idris standard library:
\begin{lstlisting}
mush : Elab ()
mush =
  do attack
    x <- gensym "x"
    intro x
    try intros
    induction (Var x) `andThen` auto
    solve
\end{lstlisting} 


Proof by reflection is also the dominant style of proof automation in Agda, which does not support tactics.
\cite{van2012engineering} demonstrate the \lstinline{isEven} example from earlier using Agda's old mechanism for reflection.
This mechanism was replaced in 2016 with a reflection mechanism based on Idris' elaborator reflection.
Other proof assistants that support proofs by reflection include HOL4~\citep{fallenstein2015proof}, 
Isabelle/HOL~\citep{chaieb2008proof}, and Milawa~\citep{Davis2015}. 

Nowadays, there is a trend of integrating the approach of proof by reflection with tactic languages.
For example, Cybele~\citep{claret2013lightweight} is a plugin for writing reflective tactics in Coq, with support for effects and non-termination. 
Rtac~\citep{Malecha2016} is a reflective tactic language for Coq, which contains specialized automation to make it simpler to write soundness proofs of 
decision procedures when writing reflective tactics.
These mixed approaches enable proof engineers to take advantage of the benefits of both approaches
to more easily build efficient automation. 

\subsection{Future Styles of Automation} 
\label{sec:futurestyle}

One drawback of using tactics is that they can sometimes impede proof understanding. 
In the future, we expect more tools
for proof understanding (in addition to existing structured proof languages). For example, a tool could use tactics to find proofs,
then simplify the result, or otherwise output a format that is easier to understand.

Along those lines, debugging tactics and tacticals can be difficult,
since the execution of tactics and tacticals often is not conducive to fine-grained debugging,
and since fully informative debugging of tactics sometimes requires interfacing 
with multiple languages (such as Ltac and OCaml).
Future tactic languages should better support debugging. 
The continued development of alternative tactic execution models
as well as typed and graphical tactic languages may help with both of these problems.

Another opportunity for improvement with existing automation is 
improved performance of tactics and tactic languages.
We expect more exploration of improving tactic performance, both by
writing tactics differently and by improving the performance of the underlying
engine.

The continued development of tools that integrate several styles of automation may help
proof engineers better take advantage of the benefits of each of these approaches.


\section{Automation in Practice}
\label{sec:specializedauto}

Both specialized and general-purpose automation help move the burden of proof away from the proof engineer and toward the tooling with which
the proof engineer interacts. This section briefly discusses a non-exhaustive sample of automation procedures (Section~\ref{sec:autotactics}).
It then concludes with a discussion of the future of automation (Section~\ref{sec:autofuture}).

\subsection{Automation Procedures}
\label{sec:autotactics}

Automation can be built using any of the various styles of automation (Section~\ref{sec:languagesforauto});
since these styles of automation overlap, we consider automation by what it achieves,
rather than by the style of automation that it utilizes.

\paragraph{Domain-Specific Automation} Domain-specific automation automates
proofs within particular domains. For example, the
\lstinline{omega}~\citep{coq-omega} tactic in Coq implements a decision procedure for quantifier-free Presburger 
arithmetic based on the Omega Test~\citep{Pugh1991}, an integer programming algorithm.
It can automatically prove mathematical statements that can be difficult for Coq users to
prove by hand.

Some domains include verifying programs within specific languages~\citep{cao2015practical,Ricketts2014},
writing mathematical proofs~\citep{nipkow1990, slind1994ac, Braibant2011, narboux2004, gregoire2005, agdasemiring}, 
deciding regular expressions~\citep{braibant2010}, and reasoning about embedded logics
such as separation logic~\citep{appel2006tactics,mccreight2009practical, Krebbers-al:POPL17}.

\paragraph{General-Purpose Automation} 
General-purpose automation is machinery that is useful across many domains.
For example, the \lstinline{break_match} tactic that we used as an example for the Ltac tactic language (Section~\ref{sec:tacticlangs})
contains useful machinery to make proofs by case analysis simpler and more robust.

Many proof assistants ship with useful general-purpose automation;
third-party tools may build on these.
For example, many proof assistants come with automation for inversion and induction~\citep{nipkow1989term, mcbride1996inverting, cornes1995automating}. 
Isabelle/Isar has special support for performing induction proofs~\citep{Wenzel2006}; besides specifying the variable to perform induction on, and the induction principle (rule) to use, a user can indicate that certain variables are to be \emph{arbitrary}, i.e., that they are not bound in the resulting assumptions and proof goals.
For example, the following clause opens a proof by strong induction for natural numbers on the expression $x - y$, where $x$ and some other variable $z$ from the previous context are arbitrary:
\begin{lstlisting}
proof (induct "x - y" arbitrary: z x rule:less_induct)
\end{lstlisting}
In Coq, a similar proof requires building a custom induction principle as a separate lemma.

Hint databases~\citep{coq-commands} in Coq store theorems 
that its other tactics~\citep{coq-tactics} such as \lstinline{auto} and \lstinline{rewrite} can use as hints.
For example, the tactic \lstinline{auto with arith} tells \lstinline{auto} to use the arithmetic theorems
defined in the \lstinline{arith} database when it tries to solve the goal. Hints can help make proofs not only
simpler, but more robust (see Chapter 3.8 of CPDT), though they may negatively impact proof search performance
or even cause it not to terminate (see Chapter 13 of CPDT). 

Some third-party libraries such the StructTact library~\citep{structtact}
and the code distributed with CPDT and FRAP~\citep{FRAPBook} ship a variety of general-purpose automation
that builds on the standard library packaged in one place.
Automation from these libraries ranges from machinery to better handle induction such as \lstinline{prep_induction}~\citep{structtact}
and \lstinline{induct}~\citep{FRAPBook} to powerful tactics like \lstinline{crush}~\citep{CPDT}, which can dispatch many proof obligations
automatically. The \lstinline{agda-prelude}~\citep{agda-prelude} library provides efficient alternatives 
to automation in the Agda standard library.

\textit{Theory exploration}---the automatic discovery and sometimes proof of theorems for a given theory---is a form of 
general-purpose automation that first arose in the context of automated theorem proving for 
mathematics~\citep{buchberger2000theory}. This style of automation aims to mimic the way that
mathematicians explore theories when writing proofs by hand.
While theory exploration tooling began with the development of specialized tooling,
specialized tools can be used in tandem with an ITP;
\cite{dramnesc2015theory}, for example, uses the theory exploration tool \textit{Theorema}~\citep{buchberger2006theorema}
in combination with Coq to explore the theory of binary trees.
The tool Hipster~\citep{JohanssonRSC14, DBLP:journals/eceasst/ValbuenaJ15, Johansson2017} for Isabelle/HOL integrates
theory exploration directly with an ITP.

Other examples of useful general-purpose automation include simple general-purpose proof automation~\citep{coq-tactics, auto2, Lindblad2004},
rewriting~\citep{coq-tactics, nipkow1989term}, 
and solving logical fragments~\citep{paulson1999generic, lescuyer2009improving, hurd2003first, kumar1991integrating, busch1994first, dahn1997integration, hurd1999integrating}, and techniques for reasoning about executable specficiations~\citep{barthe2002efficient},
as well as an implementation of a generalization of congruence closure to dependent type theory~\citep{DBLP:journals/corr/SelsamM17}.
In addition, Chapter~\ref{ch:organization} describes general-purpose automation and tooling for proof reuse (Section~\ref{sec:toolingreuse}),
as well as general-purpose automation built on type classes and canonical structures (Section~\ref{sec:desabs}).


\paragraph{Hammers}
Hammers are systems for general reasoning over large libraries of formal proofs~\citep{Blanchette2016}.
Like the verification language F*~\citep{Swamy2016} or the congruence closure algorithm~\citep{DBLP:journals/corr/SelsamM17} in Lean,
hammers leverage automated theorem provers (ATPs) from within an ITP. Hammers leverage ATPs while preserving the small trusted bases of ITPs.
They are able to learn from previous proof efforts.
In proof assistants, a hammer is exposed as a collection of tactics that in effect comprise a brute-force method for discharging a proof goal.

A hammer for a proof assistant typically has three components:
\begin{enumerate}
\item a \emph{premise selector} that selects facts (axioms) to be used by ATPs from the large library available to the proof assistant;
\item a \emph{translator} that converts the selected facts and proof goal to the restricted logics of the ATPs;
\item a \emph{proof reconstructor} that builds proofs accepted by the proof assistants from the evidence provided by the ATPs.
\end{enumerate}

The premise selector arguably has the most challenging task, since the database of facts can be large, and it is difficult to determine whether a fact is relevant to the given proof goal. The standard approach is to leverage machine learning techniques, such as naive Bayes and $k$-nearest neighbors~\citep{Blanchette2016,Czajka2018}.

The translator must take into account the particular foundations and features of the proof assistant, such as polymorphic or dependent types, and provide faithful representation in target logics, which may have no types or only monomorphic types.

The proof reconstructor, like the translator, is highly specific to the proof assistant. Reconstructors can use many different approaches of varying robustness, such as ATP proof replay, reflection, or proof assistant source generation, augmented by various heuristics. As a result, reconstruction may sometimes fail.

Implementations of hammers include Sledgehammer for Isabelle/HOL~\citep{Blanchette2013}, HOL(y)Hammer for HOL Light and HOL4~\citep{Kaliszyk2014}, and CoqHammer for Coq~\citep{Czajka2018,coqhammer}. Invocations of hammers typically spawn many parallel instances of different ATPs, such as Z3, Vampire, the E theorem prover, and CVC4. Hammer services to proof assistants can also be provided remotely, overcoming local limitations on processing power and memory~\citep{Kaliszyk2015}.

As an example of applying a hammer, consider a Coq lemma about lists from the StructTact library:
\begin{lstlisting}
Lemma app_cons_singleton_inv : forall A xs (y : A) zs w,
 xs ++ y :: zs = [w] -> xs = [] /\ y = w /\ zs = [].
\end{lstlisting}
Invoking the CoqHammer \lstinline{hammer} tactic finds a proof via Z3:
\begin{lstlisting}[language=sh]
Extracting features...
Running provers (using 8 threads)...
Z3 (nbayes-32) succeeded
- dependencies: List.app_eq_unit
\end{lstlisting}
The output also gives the following tactic call to replace the \lstinline{hammer} invocation, yielding a proof of the lemma without ATPs: 
\begin{lstlisting}
Proof. Reconstr.rcrush List.app_eq_unit Reconstr.Empty. Qed.
\end{lstlisting}

While CoqHammer generates sequences of calls to custom tactics for reconstruction, Isabelle's Sledgehammer typically results in calls to the built-in superposition prover \lstinline{metis}~\citep{Blanchette2011}, which works similarly to ATPs such as Vampire.

The usual way to evaluate the effectiveness of a hammer for a particular proof assistant is to apply the hammer on a standard library by replacing proof scripts with invocations of hammer tactics. For example, CoqHammer was able to reprove 44.5\% of all results in the Coq standard library~\citep{Czajka2018}, which is in line with success rates for HOL Light benchmarks (40\%). Success rates for benchmarks in Isabelle/HOL can be as high as 70\%, for databases with upwards of 100,000 facts~\citep{Blanchette2016}. However, these rates do not reflect practical application of hammers in evolving projects, where proof goals may be reformulated based on manual exploration using certain proof strategies. 

In contrast to property-based testing~\citep{Paraskevopoulou2015} and counterexample generators~\citep{Blanchette2010}, hammers do not give feedback when ATPs are unable to discharge a proof goal. Consequently, applying hammers does not necessarily lead to progress. On the other hand, hammers do not require decidable properties, generation of datatype instances, or domain knowledge.

Augmentation of hammer components to increase effectiveness and success rates is an active research topic~\citep{Blanchette2016b,Wang2017,Peng2017}.

\subsection{Future of Automation in Practice}
\label{sec:autofuture}

Hammers have been around in Isabelle for a long time, but until recently,
it was not known if a hammer could be implemented for a dependent type theory.
We expect more development to follow in the lines of CoqHammer.

The existence of third-party libraries for general-purpose proof automation in many ways mirrors
the rise of third-party libraries for other programming languages which supplement
the standard library; we expect more of these libraries to come into existence, and we expect
existing libraries to grow in popularity.
We also expect domain-specific tactics for common domains to continue to
develop, and to cover new domains as they arise.

\chapter{Proof Organization and Scalability}
\label{ch:organization}

As verification projects have grown in size, proof engineers have increasingly
stressed strategies and tools for organizing large proof developments and for dealing
with the challenges of scale.
This section describes some of these strategies and tools:
constructs for property specification and encodings (Section~\ref{sec:specenc}), proof design principles (Section~\ref{sec:proofdes}),
high-level frameworks (Section~\ref{sec:frameworks}), and constructs and tools for proof reuse (Section~\ref{sec:proof-reuse}).



\section{Property Specification and Encodings}
\label{sec:specenc}

Proof engineers leverage many constructs and notations to express programs and their specifications. For example, Coq offers a single basic language called Gallina for both logical formulas and (computable) functions, while Isabelle offers both an object language (e.g., HOL) and a metalogic with different operators and quantifiers. The specification languages can be extended inside proof assistants by using \emph{notations}, which provides new syntax for existing concepts; this is crucial for emulating mathematical vernacular, which can aid understanding of formal definitions.

On top of basic specification constructs, sophisticated properties can be expressed using inductive predicates via familiar definitional
mechanisms (Section~\ref{sec:definitional}). These predicates can be interpreted as higher-order Prolog programs~\citep{CoqArt}. For generality and reuse, collections of such specifications can be abstracted over using mechanisms such as parametric polymorphism~\citep{Strachey2000}, modules, and type classes (Section~\ref{sec:proofdes}).

\subsection{Domain-Specific Specification Languages}
\label{sec:dsspec}

\cite{Sewell2010} presented a domain-specific language called Ott for expressing inductive definitions and inductive properties over such definitions, suited in particular for formalizing programming language semantics. Ott files can be exported to Isabelle/HOL, Coq, and HOL4, using specific annotations for each proof assistant. For example, the regular expression datatype from Chapter~\ref{ch:ex} can be expressed in Ott as
\begin{lstlisting}
regexp :: regexp_ ::= {{ com regexp }} {{ coq-universe Type }}
 | 0 :: :: zero | 1 :: :: unit | c :: :: char
 | r + r' :: :: plus | r r' :: :: times | r * :: :: star
\end{lstlisting}
while the last matching rule for the Kleene star becomes
\begin{lstlisting}
s in L ( r )  s' in L ( r * )
----------------------------- :: star_2
s s' in L ( r * )
\end{lstlisting}
Note that the extra spacing is necessary for Ott's parser to properly disambiguate the syntax.

The more general proof assistant-agnostic specification language Lem~\citep{Mulligan2014} also includes definition of recursive functions and other programming language constructs,
as well as a standard library useful for semantic definitions. Ott files can be exported to Lem format and thus incorporated into larger definitions.

%

\subsection{Refinement of Programs, Data, and Proofs}
\label{sec:refinement}

Stepwise \textit{program refinement} is the construction of a program by a sequence of refinement steps,
where each refinement step breaks the original problem into a subproblem~\citep{Wirth1971};
these steps can be verified in an ITP.
Each refinement can be a refinement of a program without changing the datatypes, or a refinement of
the datatypes themselves (\textit{data refinement}~\citep{de1998data}).
Via the principle of propositions-as-types, similar approaches can be used to develop proofs by stepwise \textit{proof refinement} of an existing
specification.

\paragraph{Program Refinement}
A proof of refinement formally relates an abstract program to a concrete, refined version of that program.
It establishes that all of the behaviors of the concrete program are contained in the set of behaviors of the abstract program~\citep{de1998data}.
This relation can also be stated and proven in terms of the program specifications (as in the \textit{refinement calculus}~\citep{back1988calculus})
or in terms of a simulation relation (Section~\ref{sec:techniques}).
\cite{FRAPBook}	contains an overview of using program refinement to derive verified correct programs from their specificationss.

\cite{back1991} formalized the refinement calculus in HOL.
\cite{vonwright1994} presented a tool for verified program refinement using the refinement calculus in HOL.
Since then, there have been a number of refinement tools in Isabelle/HOL with support for logic~\citep{hemer2001}, object-oriented~\citep{Liu2011},
functional~\citep{lammich2013refinement}, and imperative~\citep{lammich2015refinement} programs.
\cite{Cohen2013} developed a framework for Coq called CoqEAL which automates key steps of data refinement.
\cite{Delaware2015} presented Fiat, a refinement framework for deductive synthesis of abstract data types in Coq.

Proof engineers use proofs of program refinement to break down large proof developments or to compose modular proof developments,
for example for the verification of storage systems~\citep{Chajed2019}, compilers~\citep{Leroy2009, Rizkallah16, Kumar2014}, and
OS kernels~\citep{Klein2014micro, Gu-al:POPL15, Gu-al:OSDI16}.
Refinement proofs can also help make proof developments robust to changes (Section~\ref{sec:des-scale}).

\paragraph{Proof Refinement}
Reasoning backwards in proof assistants, from goals to premises, can be viewed as a form of proof refinement~\citep{Bates1979, krafft1981},
where the proof is the refinement of the specification. The idea of proof refinement is to refine the goal to proofs of subgoals, then refine
those subgoals further.
\cite{Bates1979}, for example, describes the rule for refining a conjunction \lstinline{A }$\land$ \lstinline{ B} given hypotheses \lstinline{S}:
\begin{lstlisting}
  S pr A $\land$ B by
    S pr A
    S pr B
\end{lstlisting}
In other words, \lstinline{A }$\land$ \lstinline{ B} follows from \lstinline{S} if each of \lstinline{A} and \lstinline{B} follow from \lstinline{S}.
Each of \lstinline{A} and \lstinline{B} follow from \lstinline{S} if they can be refined using other rules.

Refinement logics such as Nuprl~\citep{Constable1986} and RedPRL~\citep{angiuli2018} as well as other proof assistants following
in the LCF tradition (Section~\ref{sec:early-history}) encourage this style of reasoning.
\cite{SterlingH17} contains an overview of proof refinement.

%
%

\section{Proof Design Principles}
\label{sec:proofdes}


Good design principles can make proofs easier to develop and maintain.
These design principles mirror software engineering design principles
in many ways, but also address challenges unique to proof
engineering.

Consider an example in Coq from \cite{Woos2016},
which demonstrates a design principle that addresses challenges unique to proof engineering.
In this example, we have a proof \lstinline{eg_proof} of a theorem \lstinline{eg},
which shows that if two functions map equal inputs to equal outputs,
then any proposition that holds on all outputs of \lstinline{g}
must also hold on all outputs of \lstinline{f}:
\begin{lstlisting}
Definition eg : Prop :=
  forall (A B : Type) (f g : A -> B) (P : B -> Prop),
    (forall x, P (g x)) ->
    (forall x, f x = g x) ->
    (forall x, P (f x)).

Lemma eg_proof : eg.
Proof.
  unfold eg. intros. rewrite H0. auto.
Qed.
\end{lstlisting}
Suppose we later change \lstinline{eg} (using \codediff{orange} to show changes):
\begin{lstlisting}
Definition eg : Prop :=
  forall (A B : Type) (f g : A -> B) (P $\codediff{Q}$ : B -> Prop),
    (forall x, P (g x)) ->
    $\codediff{(forall x, P (g x) -> Q (g x))}$ ->
    (forall x, f x = g x) ->
    (forall x, P (f x) /\ $\codediff{Q (f x)}$).
\end{lstlisting} 
As the authors note, our proof \lstinline{eg_proof} no longer holds,
since the automatically generated hypothesis name \lstinline{H0} is now called \lstinline{H1}.
One way to address this is to change the hypothesis name in the proof as well:
\begin{lstlisting}[language=Coq]
Proof.
  unfold eg. intros. $\codediff{rewrite H1.}$ auto.
Qed.
\end{lstlisting}
But if we continue changing \lstinline{eg}, then we will need to keep making these kinds of changes.
Instead, the authors advocate for using the tactic \lstinline{find_rewrite},
since it does not depend on hypothesis names:
\begin{lstlisting}[language=Coq]
Proof.
  unfold eg. intros. $\codediff{find\_rewrite.}$ auto.
Qed.
\end{lstlisting}
This proof goes through for both definitions of \lstinline{eg}.

The design principle from this example addresses a challenge unique to proof engineering, since it deals with the consequences of proof automation.
Other proof engineering design principles mirror software engineering design principles.
This section provides an overview of design principles for proof engineering, drawing
parallels to software engineering when appropriate.
It focuses on general-purpose design principles, and discusses domain-specific design principles (beyond those from Chapter~\ref{cha:appl-mech-proofs})
only when relevant more broadly.


\subsection{Design Principles for Abstraction} 
\label{sec:desabs}

As in software engineering, design principles for proof engineers prevent changes in implementation
from breaking dependencies that ought to rely only on specifications.
For example, much like a software engineer may write an interface for a collection of functions so that he can
switch out implementation details such as the underlying data structure 
without breaking functionality that
depends on those functions, so a proof engineer may write an interface for a collection of lemmas
so that changes to the proofs of those lemmas do not break other lemmas and theorems that depend on those lemmas~\citep{Woos2016}.

There are many ways to achieve this sort of abstraction in ITPs,
some of which have different impliciations for proof automation. For example,
in Coq, it is possible to write interfaces using modules, type classes, or canonical structures;
Coq has special support for proof search for type classes~\citep{Sozeau2008} and canonical structures~\citep{Saibi:PhD}. 
This section describes some means of abstraction in an ITP.


\paragraph{Modules}
Modules, as manifested in languages such as Standard ML~\citep{MacQueen1986} and proof assistants such as Coq~\citep{Chrzaszcz2003}, are collections of named components which may be types, values, or nested modules. A central property is the separation of module interfaces (signatures or module types) and module implementations (structures). The interface-implementation relation is many-to-many; one signature can be implemented by several structures, and one structure can implement several signatures. A structure can choose to hide all information not specified in the signatures it implements.


Parametric modules, called functors, take structures that implement certain signatures as arguments. In proof assistants, functors can provide abstraction and reuse of both functions and proofs. This approach is taken to implement finite sets and maps in Coq using AVL trees~\citep{Filliatre2004}, and later to implement balanced binary search trees
using red-black trees~\citep{Appel2011b}.

\paragraph{Type Classes}
Type classes were first implemented in the Haskell programming language~\citep{Wadler1989}.
In a proof assistant context, type classes have notably been implemented for Coq~\citep{Sozeau2008} and Isabelle/HOL~\citep{Haftmann2006};
instance arguments~\citep{Devriese2011} are a similar feature in Agda.
Type classes can be viewed as a particular use of a module system as in Standard ML, and type classes can coexist with such a module system~\citep{Dreyer2007}.

A type class can be viewed as an abstract data type that defines a collection of functions by their parameter types, while not fixing function implementations. The abstract data type can then be implemented in different ways for different parameter types. For example, Volume 4 of Software Foundations~\citep{Pierce-al:SF} describes an
equality type class with a single function \texttt{eqb} which, when provided two arguments of the same type, returns a boolean:
\begin{lstlisting}
Class Eq A :=
{
  eqb: A -> A -> bool;
}.
\end{lstlisting}
Different implementations (\emph{instances}) of this class can then be provided for different types;
Software Foundations describes one for booleans:
\begin{lstlisting}
Instance eqBool : Eq bool :=
{
  eqb := fun (b c : bool) =>
     match b, c with
       | true, true => true
       | true, false => false
       | false, true => false
       | false, false => true
     end
}.
\end{lstlisting}
and one for natural numbers:
\begin{lstlisting}
Instance eqNat : Eq nat :=
{
  eqb := Nat.eqb
}.
\end{lstlisting}
A compiler can translate programs that use functions defined for type classes to programs that do not by looking up and applying the appropriate function instances, using information about function invocation types.

A key use of type classes in Haskell programs is as a way to structure programs by abstracting certain code over appropriate type classes, and concretizing the abstracted code with appropriate type instances elsewhere, avoiding duplication and facilitating reuse;
proof engineers can use type classes similary to achieve both code and proof reuse.
However, type classes in ITPs also provide additional benefits beyond those that type classes in other languages such as Haskell provide.
For example, one drawback of using only type classes for structuring programs in Haskell is that a type can implement a type class in exactly one way~\citep{HarperModules2011};
it may be useful to define different type class instances for sorting of integers depending on the size of the input data, or use different orders on integers for sorting.
Unlike in Haskell, type classes in Coq can support multiple instances.

The Coq implementation of type classes is \emph{first-class}, meaning that it is a thin layer on top of existing functionality (specifically, implicit arguments and dependent records). In addition, \cite{Sozeau2008} added specific support for type class resolution into Coq's proof search mechanism. In contrast to Coq's type classes, the type classes of Isabelle/HOL are restricted to one type variable, and are not first-class. One particular advantage of type classes in proof assistants (as opposed to type classes in Haskell) is that propositions can be type class members, e.g., a type class for a monad can require witnesses (proofs) for the monad laws along with monad operations. By extension, this means that proofs in one type class instance can be derived partly from proofs in other type class instances (e.g., of some more general class). In contrast with Haskell,
the type class instance resolution system in Coq can always be elided by manually passing implicit type class instances.

Type classes have been used for abstraction and reuse in many proof developments.
For example, \cite{Spitters2011} used type classes to represent a standard algebraic hierarchy in Coq, along with parts of category theory.
\cite{Woos2016} used type classes to organize the correctness proof of the Raft consensus protocol in Coq, and for abstracting the Raft protocol implementation for replication over arbitrary state machines.

Type classes are closely tied to other language features for abstraction.
General parametrization of theories in Isabelle can be achieved via \emph{locales}~\citep{Kammueller1999,Ballarin2006}, which is the mechanism used to provide type class support~\citep{Haftmann2009}; a locale can be viewed as a persistent proof context that includes arbitrary variables and assumptions, and which can be instantiated in other proofs.



\paragraph{Canonical Structures}
In Coq, canonical structures~\citep{Mahboubi-Tassi:ITP13, Saibi:PhD} provide an alternative to type classes.
Canonical structures are a mechanism to provide theory-specific dictionaries to datatypes, allowing for more
flexible resolution strategy than more the more widely used type classes.

To show their typical use, consider partial commutative monoids
(PCMs); an algebraic structure which recurs in our current ongoing
work on the verification of stateful and concurrent
programs~\citep{Nanevski-al:ESOP14}. We implement PCMs using two of the
Coq's native constructs: \emph{dependent records} and \emph{canonical
  structures}.  We follow the established \ssreflect design pattern of
defining algebraic data structures by means of \emph{mix-in}
composition~\citep{Garillot:PhD}, whereby different dependent records
formalize different algebraic properties, which can be combined using
\emph{packed classes} mechanism. The latter also defines the field
resolution strategy~\citep{Garillot-al:TPHOL09} in a case of
overlapping names.
%
%
%
For instance, in Coq the \emph{mix-in} defining PCMs is represented by the
following dependent record:
\begin{lstlisting}
Record mixin_of (T : Type) := Mixin {
  valid : T -> bool;
  join : T -> T -> T;
  unit : T;
  _ : commutative join;
  _ : associative join;
  _ : left_id unit join;
  _ : forall x y, valid (join x y) -> valid x;
  _ : valid unit }.
\end{lstlisting}
The type \code{T} is the \emph{carrier type} of the structure. The
field \code{valid} selects a subset of \code{T}, standing for the
``defined'' elements. The invalid (or ``undefined'') elements help
model partiality: a partial function over \code{T} will return some
invalid element on an input on which it is mathematically undefined.
\code{join} is the binary operation of the PCM, and \code{unit}
is the unit element. The remaining five unnamed fields enumerate the
axioms that have to be satisfied by each PCM instance.

Next, the mix-in ``interface'' is packaged with a carrier type, into a
dependent record type, which represents PCMs. We also introduce a
coercion from the package to the underlying carrier type, so that the
two can be conflated. This coercion essentially accounts for the
delegation hierarchy from object-oriented languages.
\begin{lstlisting}
Structure pcm : Type := Pack {type : Type; _ : mixin_of type}.
Coercion type : pcm >-> Sortclass.
\end{lstlisting}

Next, we explain the mechanism of packaging all necessary definitions
along with lemmas about data structures (such as \emph{join}'s
commutativity and associativity in the case of PCMs) into a single
module that should be imported by the clients of the algebraic
structure.  For example, we introduce appropriate notation for the
join operation, and specifically name and prove the lemmas that
correspond to the PCM properties that we left unnamed in the mixin.
\begin{lstlisting}
Notation x \+ y := (join x y).
Lemma joinC (U : pcm) (x y : U) : x \+ y = y \+ x.
Lemma joinA (U : pcm) (x y z : U) :  x \+ y \+ z = x \+ (y \+ z).
\end{lstlisting}
The lemmas such as \code{joinC} and \code{joinA} are proved by
destructing the package \code{U}, but notice how the coercion allows
conflating \code{U} with its carrier type. Also notice how the
notation \code{\\+} allows the PCM \code{U} to be ommitted from
the equations themselves, as the typechecker can infer it from the
context.

Algebraic structures can inherit the properties of other, more basic
structures. Thus, we also require an analogue of object-oriented
\emph{inheritance}. We illustrate how this can be done in Coq, by
defining an interface for a \emph{cancellative} PCM, which inherits
from an ordinary PCM. The cancellative PCM is defined as the following
mix-in record:
\begin{lstlisting}
Record mixin_of (U : pcm) := Mixin {
  _ : forall a b c: U, valid (a \+ b) -> a \+ b = a \+ c -> b = c
}.
\end{lstlisting}
Notice that the dependent record \code{mixin_of} in this case is
parametrized via the carrier PCM \code{U}, which is used as a target
for a coercion whenever an instance of a plain PCM or a carrier type
\code{U} is required, since coercions a transitive.

Let us now \emph{instantiate} the definition of abstract structure
with concrete datatypes. It turns out that it is insufficient to
merely prove that a datatype satisfies the PCM axioms. To work
comfortably with an algebraic structure in practice, one has to
explicitly ``register'' the structure with the type inference engine.

We first show what goes wrong if one doesn't perform the
``registration.''  For instance, assume we first define an instance of
a PCM for \code{nat} with addition, by proving that \code{+} with
\code{0} satisfies the PCM axioms.  Then the following lemma which
uses the generic notation \code{\\+} for the PCM operation, is
considered ill-formed by Coq. The reason is that Coq cannot figure
that there is a PCM associated with \code{nat}, and that the generic
notation \code{\\+} should be resolved with addition. Indeed, we could
have defined the PCM for \code{nat} via multiplication $(\times)$ with
$1$, in which case \code{\\+} should be resolved by~$\times$.
\begin{lstlisting}
Lemma add_perm (a b c : nat) : a \+ (b \+ c) = c \+ (b \+ a).
\end{lstlisting}
In the above case, once a structure is registered as the default PCM
for \code{nat}, the \code{add_perm} lemma can be proved by selective
rewriting using the standard PCM properties.

Some notable uses of canonical structures include telescopes~\citep{Garillot-al:TPHOL09} and
higher-order tactics for separation logic~\citep{Gonthier2011}.

\subsection{Design Principles for Programming with Dependent Types}

In order to use dependent types to their full extent, proof engineers have developed many
paradigms to deal with the challenges they present.
For example, using dependent types, we can define heterogenous lists and a selection function
over hetereogenous lists. To write the selection function, however, we must first define
the \textit{type} that it has, which depends on the case. \cite{CPDT} accomplishes this
using a membership predicate:
\begin{lstlisting}
Section hlist.
  Variable A : Type.
  Variable B : A -> Type.

  Inductive hlist {A : Type} {B : A -> Type} : list A -> Type :=
  | HNil : hlist nil
  | HCons : forall (x : A) (ls : list A), B x -> hlist ls -> hlist (x :: ls).

  Variable elm : A.

  Inductive member : list A -> Type :=
  | HFirst : forall ls, member (elm :: ls)
  | HNext : forall x ls, member ls -> member (x :: ls).

  Fixpoint hget ls (mls : hlist ls) : member ls -> B elm :=
    (* ... *)
End hlist.
\end{lstlisting}
In other words, the type of \lstinline{hget} states that whenever \lstinline{elm} is a member of some list \lstinline{ls},
then we can select some element of type \lstinline{B elm} from any \lstinline{mls : hlist ls}.
This is one example of a common style of writing functions and proofs using dependent types,
wherein the proof engineer first defines the type of the function or proof inductively,
and then defines the function or proof that has that type.
It is a powerful style that makes it possible to define very expressive types.

\cite{CPDT} provides a comprehensive overview
of dependently-typed programming in Coq with many more examples.
\cite{Tanter2015} outlines design principles for gradual verification in Coq,
which may help reduce the burden of verification with dependent types and increase adoption. 

\subsection{Design Principles for Scale}
\label{sec:des-scale}

The scale of programs verified in ITPs has increased over the years.
In recent years, proof engineers have begun to look at how to address the challenges that come with this increase in scale.
Proof engineering in the large~\citep{Kaivola2003}, for example, describes a methodology
for verifying large-scale ciruits; it is among the earliest work
noting that proof design for large verification projects is important.
This section describes design principles dealing with the challenges of scale such as
robustness in the face of changes, compositionality of components, and efficiency of code and proofs.
In addition, Section~\ref{sec:proof-reuse} describes design principles for proof reuse.

\paragraph{Design Principles for Robustness}
A major source of inefficiency in verification is \textit{proof
  brittleness}: Even a minor change to a single theorem or definition can break many
dependent proofs. This makes proofs difficult to
maintain~\citep{Woos2016, Aydemir2008, plse-coevolve-djg-fose14, Delaware2013ICFP}.
Design principles help make proofs robust in the face of changes.
This is one approach to proof evolution (Section~\ref{sec:evolve}).

One approach to building robust proofs is to make use of proof automation (Chapter~\ref{sec:prooflangs}) to dispatch similar goals.
Proof engineers who use the default tactics included in many proof assistants
already take advantage of this, since the same tactic can discharge different goals.
For example, a proof engineer who uses \lstinline{omega} in Coq or \lstinline{presburger} in
Isabelle/HOL need not change the proof script as the goal changes, so long as the goal stays within
the same fragment of arithmetic solved by those tactics.

The degree to which proof engineers rely on automation varies by style.
CPDT~\citep{CPDT}, for example, advocates for the heavy use of program-specific automation,
noting that this makes proofs more robust; the tagless interpreter proofs from Chapter 8.3 of CPDT contain
an example of automation of this kind.
This style of development localizes the burden of change to the automation itself
as opposed to the many proofs that use the automation.

While automation can make proofs more robust, it can also be brittle in itself.
For example, some Coq tactics automatically generate hypothesis names; small changes in specifications can
cause proofs that rely on those names to break. One approach to this problem is to always explicitly specify hypothesis names,
so that Coq never generates hypothesis names automatically; the IDE Company-Coq~\citep{CompanyCoq2016} provides some built-in
support for this approach.
Planning for Change~\citep{Woos2016} notes that, while this approach helps, it is still necessary to update
those explicit names as specifications change. Instead, the authors advocate for the use of \textit{structural tactics}, or tactics that do not
depend on hypothesis names and hypothesis ordering; many tactics of this style can be found in the associated
StructTact~\citep{structtact} library.

Planning for Change addresses design for robustness not only at the level of automation,
but also at the level of specifications and proof objects. It presents a methodology
for writing robust proofs independently of any domain or framework. This methodology
is informed by a large proof engineering effort verifying the Raft consensus protocol.
It is a set of five recommendations. Some of these recommendations draw on software engineering design principles.
For example, the authors recommend using information hiding techniques similar to those used in software engineering to hide definitions.
That way, the burden of change is localized to interface changes, and changes in only implementation do not cause
breaking changes in dependencies. Other recommendations tackle challenges that are unique to proof engineering. For example, the authors
advocate for the use of custom induction principles to capture common patterns in inductive proofs.

Refinement (Section~\ref{sec:refinement}) can help make proofs robust to changes.
For example, the proofs of the seL4 microkernel in Isabelle/HOL have evolved alongside the implementation for
over eight years~\citep{Klein2014micro}.
The proof development makes use of two layers of specifications: an \textit{abstract specification} which describes only behavior of
the system, and an \textit{executable specification} which includes implementation details. These two layers are connected by
a refinement proof.
Using this approach, the authors found that both making low-level changes and adding new simple features were not very costly,
though more complex changes that interacted with other parts of the code significantly were still costly.

\paragraph{Design Principles for Compositionality}
CompCert (Section~\ref{sec:verified-compilers}) employs a compositional design for describing the different intermediate languages
and how they interact with each other. Affinity lemmas from Planning for Change
also capture this concept.
The CertiKOS project introduces the idea of a \textit{deep specification}~\citep{Gu-al:POPL15} that makes
compositional verification more tractable. DeepSpec~\citep{DeepSpec},
an ongoing project, is addressing this problem more generally.
Section~\ref{sec:frameworks} describes frameworks for compositional verification.

\paragraph{Design Principles for Efficiency}
Some proof assistants like Coq work by extraction (Section~\ref{sec:trust-programs}) from the core language into an executable language.
The resulting extracted code can be slow, which can be a barrier for verifying a realistic system.
\citet{CruzFilipe2003, CruzFilipe2006} describe proof design principles for optimizing the efficiency of extracted
code.


\subsection{Style Guides and Proof Techniques}
\label{sec:techniques}

Style guides and proof techniques help
guide proof engineers in dealing with common patterns to address common challenges.

\paragraph{Style Guides}
Section~\ref{sec:dsmathematics} described style guides in mathematics.
A few general-purpose style guides exist.
Gerwin's style guide~\citep{isabellestyle} for Isabelle, for example,
is a set of guidelines that are used within Isabelle itself and in several
large developments. The Isabelle Archive of Formal Proofs requires that submitted proofs
follow some of these guidelines, and recommends others~\citep{isabelleafp}.
The CoqStyle~\citep{coqstyle} style guide is a set of guidelines for Coq
in the main Coq repository which is used within the standard library. 

\paragraph{Proof Techniques}
Proof techniques are techniques that handle common classes of proofs, or that make it easier
to write proofs in a particular style.
For example, one widely-used proof technique is \textit{simulation}~\citep{Lynch1994};
an overview of this technique can be found in \cite{FRAPBook}.
This technique helps proof engineers prove that systems
preserve liveness and safety properties.
Refinement (Section~\ref{sec:refinement}) reduces to simulation~\citep{Klein2014micro}.

One application of simulation is to show compiler correctness.
CompCert~\citep{Leroy2009}, for example, uses this technique to show that the
program transformations that the compiler makes are semantics-preserving.
In the case of compiler correctness as in CompCert, both directions of simulation (\textit{forward simulation} and \textit{backward simulation})
start with a source program $s$ and a target program $t$ that are related along some relation $r$.
%
The forward simulation (Figure~\ref{fig:sim}, left) states that if $s$ steps to $s'$, then $t$ can step to some $t'$, where $s'$ and $t'$ are related by $r$.
Similarly, the backward simulation (Figure~\ref{fig:sim}, right) states that if $t$ steps to $t'$, then $s$ can step some $s'$, where $t'$ and $s'$ are related by $r$.
Intuitively,
  a forward simulation shows that
  ``anything the source program could do, the target program could do too,'' and
  a backward simulation shows that
  ``anything the target program could do, the source program could do too.''
Together,
  a forward and a backward simulation establish \textit{indistinguishability},
  any entity restricted to only observe ``visible'' program transitions
  (e.g., input and output) will never be able to determine if they
  are interacting with the source or target program~\citep{Sangiorgi2011}.
If the source language is deterministic, then the forward simulation follows from the backward simulation;
if the target language is deterministic, then the backward simulation follows from the forward simulation~\citep{Leroy2009}.
CompCert takes advantage of this to show backward simulation from only forward simulation and determinism of the target language.

\begin{figure}
    \begin{minipage}{0.48\textwidth}
    \centering
     \begin{tikzpicture}[->,>=stealth',shorten >=1pt,auto,node distance=2.8cm,
                    semithick]
  \tikzstyle{every state}=[fill=none,draw=none,text=black]

  \node[state]         (A)              {$s$};
  \node[state]         (B) [right of=A] {$t$};
  \node[state]         (D) [below of=A] {$s'$};
  \node[state]         (C) [below of=B] {$t'$};

  \path (A) edge[-]             node {r} (B)
        (B) edge[dotted]      node {} (C)
        (C) edge[dotted, -]      node {r} (D)
        (A) edge              node {} (D);
\end{tikzpicture}
    \end{minipage}
    \hfill
    \begin{minipage}{0.48\textwidth}
    \centering
    \begin{tikzpicture}[->,>=stealth',shorten >=1pt,auto,node distance=2.8cm,
                    semithick]
  \tikzstyle{every state}=[fill=none,draw=none,text=black]

  \node[state]         (A)              {$s$};
  \node[state]         (B) [right of=A] {$t$};
  \node[state]         (D) [below of=A] {$s'$};
  \node[state]         (C) [below of=B] {$t'$};

  \path (A) edge[-]   node {r} (B)
        (B) edge              node {} (C)
        (C) edge[dotted, -]              node {r} (D)
        (A) edge[dotted]              node {} (D);
\end{tikzpicture}
    \end{minipage}
    \caption{Forward (left) and backward (right) simulation for compiler correctness.
    Premises are show as solid lines and goals are shown as dashed lines.}
    \label{fig:sim}
\end{figure}
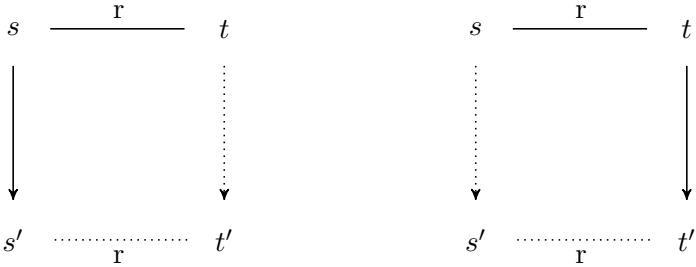

Section~\ref{sec:languagesforauto} discusses some techniques for interacting with automation.
Proof techniques can also help proof engineers reason within certain domains. \cite{bahr2015}, for example,
describes a technique for deriving correct compilers from specifications in Coq.
Section~\ref{sec:reas-about-imper} describes techniques for reasoning about imperative programs.


\subsection{Design Principles for Reasoning about Imperative Programs}
\label{sec:reas-about-imper}

When desigining a verification framework for imperative programs based on a dependently-typed
proof assistant (\eg, Coq), the most common approach is to implement a
version of a Floyd-Hoare style program
logic~\citep{Floyd1967Flowcharts,Hoare:CACM69} in it. When doing so,
the framework designer is faced with the following choices:

\begin{itemize}
\item How to embed, into a proof assistant, the language with the
  features, which the \emph{host} language does not have (\eg, mutable
  state and concurrency)?

\item How to encode verification conditions for imperative programs
  specified in Floyd-Hoare style, and implement the corresponding
  reasoning principles in a proof assistant?
\end{itemize}

\noindent
Below, we elaborate on these two design choices and provide a survey
of the most prominent approaches implementing them, both for
sequential and concurrent reasoning about imperative programs.

\subsubsection{On Shallow and Deep Embedding}
\label{sec:embedding}
An important design decision to take when designing a framework for
verification of effectful (i.e., heap-manipulating or concurrent)
programs on top of a general-purpose proof assistant is its use of
shallow or deep embedding the language to be verified.

\textit{Shallow embedding} is an approach of implementing programming
languages, characterized by representation of the language of interest
(usually called a \emph{domain-specific language} or DSL) as a subset
of another general-purpose host language, so the programs in the
former one are simply the programs in the latter one. The idea of
shallow embedding originates at early '60s with the beginning of era
of the Lisp programming language~\citep{Graham:BOOK}, which, thanks to
its macro-expansion system, serves as a powerful platform to implement
DSLs by means of shallow embedding (such DSLs are sometimes called
internal or embedded).

Shallow embedding in the world of practical programming is advocated
for a high speed of language prototyping and the ability to reuse
most of the host language infrastructure.  An alternative approach of
implementing and encoding programming languages 
is called \textit{deep embedding}, and amounts to the
implementation of a DSL from scratch, essentially,
writing its parser, interpreter and type-checker in a general-purpose
language. Deep embedding is preferable when the overall
performance of the implemented language runtime is of more interest
than the speed of DSL implementation, since then a lot of intermediate
abstractions, which are artifacts of the host language, can be
avoided.

In the world of mechanized program verification, both deep and shallow
embeddings have their own strengths and weaknesses.  Although
implementations of deeply embedded languages and calculi naturally
tend to be more verbose, design choices in them are usually simpler to
explain and motivate. Moreover, the deep embedding approach makes the
problem of name binding to be explicit, so it would be appreciated as
an important aspect in the design and reasoning about programming
languages~\citep{Aydemir2008,Weirich-al:ICFP11,Chargueraud2011}. We
believe that these are the reasons why this approach is typically
chosen as a preferable one when teaching program specification and
verification in Coq~\citep{Pierce-al:SF}.

Importantly, deep embedding gives the programming language implementor
full control over its syntax and semantics.
In particular, the expressivity limits of a defined logic or a type
system are not limited by expressivity of the host language's
type system. Deep embedding makes it much more
straightforward to reason about pairs of programs by means of defining
the relations as propositions on pairs of syntactic trees, which are
implemented as elements of corresponding datatypes.  This point
becomes crucial when one needs to reason about the correctness of
program transformations and optimizing compilers~\citep{Appel:BOOK14}.

In contrast, the choice of shallow embedding, while sparing one the
labor of implementing the parser, name binder and type checker, may
limit the expressivity of the logical calculus or a type system to be
defied. In the case of Hoare Type Theory~\citep{Nanevski-al:POPL10},
for instance, it amounts to the impossibility to specify programs that
store effectful functions and their specifications into a
heap.\footnote{This limitation can be, however, overcome by
  postulating necessary axioms.}

In the past decade Coq has been used
in a large number of projects targeting formalization of logics and
type systems of various programming languages and proving their
soundness, with most of them preferring the deep embedding approach to
the shallow one. We believe that the explanation of this phenomenon is
the fact that it is much more straightforward to define semantics of a
deeply-embedded ``featherweight'' calculus~\citep{Igarashi-al:TOPLAS01}
and prove soundness of its type system or program logic, given that it
is the ultimate goal of the research project. However, in order to use
the implemented framework to specify and verify realistic programs, a
significant implementation effort is required to extend the deep
implementation beyond the ``core language,'' which makes shallow
embedding more preferable.

\subsubsection{Encoding Verification Conditions}

In a Floyd-Hoare style logic, specification of a program $c$ is given
in a form of a \textrm{tiple} $\set{P}~c~\set{Q}$, where the
assertions $P$ and $Q$ are referred to as the \emph{precondition} and
the \emph{postcondition}, respectively. The standard semantics of
the triple imposes that for any state, satisfying $P$, the final
state, after $c$ terminates, satisfies $Q$.
This definition corresponds to termination-insensitive \emph{partial
  correctness} (\ie, $c$ is allowing to not terminate at all, so any
postcondition would hold). Some program logics impose a stronger
semantics of \emph{total correctness}, requiring $c$ to terminate, in
addition to the above~\citep{Dockins-Hobor:DS10}.

This treatment of a Floyd-Hoare triple allows for verifying the
programs by means of following the inference rules of a program
logic, allowing to decompose the proof of $\set{P}~c~\set{Q}$ into the
proofs about $c$'s sub-programs~\citep{Hoare:CACM69}.
While this style of reasoning seems natural and relatively easy to
implement in a proof assistant, and is advocated by the most widely
used tutorials~\citep{Pierce-al:SF}, it is not the most convenient to
conduct the proofs in, due to the need of constantly discharge the
\emph{weakening} obligations, required for ``massaging'' the
verification goal, and represented by the following inference rule:
\begin{mathpar}
\inferrule*[Right={(Weaken)}]
 {P \Rightarrow P' \\
  \spec{P'}~ c~ \spec{Q'}\\
  Q' \Rightarrow Q}
 {\spec{P}~ c~\spec{Q}}
\end{mathpar}

A more proof-assistant-friendly way to encode the Floyd-Hoare-style
verification conditions, ``compressing'' the necessary applications of
the weakening rule, is to use the idea of predicate transformer
by~\citep{Dijkstra:CACM75} that can be used to compute a
pre-condition for a computation, for any context in which that
computation may be used.
This approach, dubbed the \emph{weakest precondition (WP) calculus} allows
one to encode the meaning of a Floyd-Hoare triple (roughly) as
follows:
\[
  \set{P}~c~\set{Q} \eqdef P \Rightarrow \mathsf{wp}~c~\set{Q},
\]
where $\mathsf{wp}~c~\set{Q}$ is the program $c$'s \emph{weakest}
precondition \wrt the imposed postcondition $Q$, expressed as a
logical formula.
Therefore, what is left to the designer of the mechanised program
logic to do is to provide the implementation of the primitive
$\mathsf{wp}$, which would ``compile'' a program and its postcondition
to the logical assertion, which can be later discharged using the host
proof assistant's machinery.

The weakest precondition approach to encoding verification
conditions for imperative programs is amazingly versatile, and, to the
best our knowledge, has been adopted in most of the major
implementations of program logics embedded into proof
assistants~\citep{Nanevski2008,mccreight2009practical,Nanevski-al:POPL10,Chargueraud2010,Chlipala:PLDI11,Swamy-al:PLDI13,Appel:BOOK14,Krebbers-al:POPL17}.

\subsubsection{Verifying sequential heap-manipulating programs}
\label{sec:reas-about-heap}

The main success in a program logic-based verification of
heap-manipulating programs has been achieved with the discovery of
Separation Logic~\citep{OHearn-al:CSL01,Reynolds:LICS02}.
It did not take long for Separation Logic to be mechanised in an ITP~\citep{Nanevski2008,Mehta-Nipkow:CADE03,mccreight2009practical},
using both deep and shallow embedding.

One of the most successful formalizations in Coq
by means of shallow embedding is the series of work on Hoare Type Theory
(HTT) by~\cite{Nanevski2008}, known as YNot. YNot has been used, among other things, in verifying
a relational database system~\citep{Malecha2010} and a secure browser kernel~\citep{Jang2012}.

In addition to adopting the WP-calculus for expressing verification
conditions in an embedding of Separation Logic into Coq, HTT first
made active use of \emph{binary postconditions}, enabling a
straightforward treatment of \emph{logical} variables, whose scope
spans both pre- and postconditions of a Floyd-Hoare
triple. Specifically, this has been achieved by making a postcondition
$Q$ to be not of type $\mathsf{state} \to \mathsf{Prop}$ (\ie, unary,
as suggested by the textbook expositions of program logics), but
rather of type $\mathsf{state} \to \mathsf{state} \to \mathsf{Prop}$,
\ie, constraining both the pre- and the post-state.
This style of specification has later been adopted by multiple other
verification frameworks~\citep{Swierstra:TPHOLS09,Swamy-al:PLDI13}.

The \textsc{Sepref}~\citep{lammich2015refinement} tool
for verifying imperative programs in Isabelle/HOL includes a separation logic
framework built on top of Imperative HOL, which is built on top of Isabelle/HOL.
\textsc{Sepref} uses refinement (Section~\ref{sec:refinement}) to derive an imperative heap-based program and
correctness proof from a functional program and correctness proof.

YNot~\citep{Nanevski2008} has implemented the \emph{heap disjointness},
inherent to Separation Logic, by means of a deep embedding of logic
reasoning principles.
Such an embedding of a domain-specific logic required a later
development of a number of tactics for making large mechanised proofs
tractable~\citep{Chlipala-al:ICFP09}.
In a later work,~\cite{Nanevski-al:POPL10} have shown shown how to
achieve almost the same expressivity with very little domain-specific
automation, by making reasoning about finite heaps decidable and
leveraging the machinery of small-scale
reflection~\citep{Gonthier2010}.

Various successful deep embeddings of Floyd-Hoare style reasoning into
Coq have been demonstrated viable for the sake of reasoning about
low-level programs using different versions of Separation
Logic~\citep{Chlipala-al:ICFP09,Chlipala:PLDI11,Chlipala2013,Chen2015,Cao2018}.
All those efforts came supplied with tailored libraries of
domain-specific tactics, with those tactics automatically
applying Separation Logic's \textsc{Frame} rule and thus progressively
reducing the size of the verification goal.


\subsection{Future of Design}

While we think that domain-specific design principles will always be important, 
we expect that there is a lot of potential for general-purpose design principles that frame proof engineering
in the context of software engineering and make novel use of what we already know.
Planning for Change investigates where
proof engineering diverges from software engineering and where it calls for specialized techniques;
continuing along these lines should drive more useful proof design techniques.

Compared to design principles for mathematics, current general-purpose design principles place
little emphasis on proof understanding.  While this is
an understandable difference in emphasis,
his can inhibit collaboration for proof engineers as well. We expect more work on proof
understanding to become common as collaboration between proof
engineers increases with the growth in large-scale verification
projects.

Automation-heavy styles can help prevent breaking changes, but have
drawbacks. Some of these drawbacks may be avoidable.  For example, one
limitation is that proof checking of the large and complex terms
these procedures produce can be slow.  Developments in proof checking such
as term simplification could make this style more tractable. Debugging
is also difficult; alleviating this concern could be as simple as
better debugging tooling for tactics.


\section{High-Level Verification Frameworks}
\label{sec:frameworks}

In the context of software engineering, a \textit{framework} is distinguished from a library or domain-specific language
in that the client relinquishes control of execution to the framework. In practice,
the concept of a framework often refers to some combination of design principles,
libraries, and tooling that together give structure to code, often within a certain domain,
regardless of control of execution. We use the latter term, as it is what is used most often
in proof engineering papers, and as the concept of control of execution does not
always make sense in the context of proof development.

Several of the libraries and languages we have already discussed (for example, Bedrock~\citep{Chlipala2013})
fit this definition of a framework. This section extends that discussion to cover frameworks
for two common domains: concurrent applications (Section~\ref{sec:reas-about-conc}) and language design and metatheory (Section~\ref{sec:metaframe}).
It then discusses frameworks for a few other domains (Section~\ref{sec:otherframe}), and concludes with
a discussion of the future of frameworks (Section~\ref{sec:futureframe}) for proof engineering.

\subsection{Frameworks for Verifying Concurrent Applications}
\label{sec:reas-about-conc}

Reasoning about concurrent programs brings new challenges into
mechanising reasoning: due to the excessively large state-space of
possible interactions between simultaneously executing processes or
threads, simply enumerating them is no longer tractable.
However, since in most of the practical applications the interaction
between processes on some sort of shared state happens only at
dedicate program points, via specific programming primitives, a
plausible way to reduce this complexity is to reduce \emph{concurrent}
reasoning to a \emph{sequential} one.
This idea has been pioneered in the work on \emph{Concurrent
  Separation Logic} by~\cite{ohearn07resources}, which
provided a series of inference rules for compositional sequential and
concurrent reasoning for shared-memory concurrency.

Similarly to plain Separation Logic, variants of CSL have been
implemented as both shallow and deep embedding with the corresponding
benefits and drawbacks.

The first \emph{shallow embedding} of Subjective Concurrent Separation
Logic, a CSL-like logic for concurrency, was due
to~\cite{LeyWild-Nanevski:POPL13}, who implemented it using Coq's
indexed types. Unlike the prior work on Hoare Type
Theory~\citep{Nanevski-al:POPL10}, in which Coq's dependent types were
only capturing the effect of an imperative program on a state, in
SCSL, the types were also carrying information about \emph{resource
  invariants}, capturing the contract of a concurrent interaction
between threads.
That work has been later extended to a more expressive
\emph{Fine-Grained Concurrent Separation Logic} (FCSL)
\citep{Nanevski-al:ESOP14,Sergey-al:PLDI15}, which provided a more
general treatment of concurrent resources, incorporating ideas from
both CSL and Rely-Guarantee-based verification
methodologies~\citep{Jones:IFIP83,Feng:POPL09}, and implementing them
in a form of a shallowly-embedded type theory for state.

The main shortcoming of both SCSL and FCSL, both being
shallowly-embedded type theories for state, are the limitations due to
the limitations of Coq's model \wrt impredicativity.
At the time of this writing, FCSL did not support higher-order heaps
(\ie, the possibility to reason about arbitrary storable effectful
procedures). It was conjectured by FCSL's authors that this obstacle
could be overcome by relying on the universe polymorphism feature
introduced in Coq version~8.5~\citep{Sozeau-Tabareau:ITP14}.
An approach based on \emph{Rely-Guarantee references}, similar to FCSL
in spirit, employed Coq as a host framework for implementing DSL (but
not proving its soundness \wrt some semantics) for streamlining
reasoning about certain concurrency
patterns~\citep{Gordon-al:TOPLAS17}, allowed by considering
Rely-Guarantee contracts, but without CSL-enabled proof modularity.

Implementation of concurrent imperative programs in Coq by means of
\emph{deep embedding} has been first considered in the context of
verifying low-level code with dynamic thread creations in CAP and CCAP
program logics~\citep{Yu-Shao:ICFP04,Feng-Shao:ICFP05}.
Targeting real architectures, those formal verification efforts
required astonishingly high proof efforts and have been eventually
superseded by a mechanized proof methodology based on \emph{certified
  abstraction layers}, not grounded in any specific Hoare-style
program logic~\citep{Gu-al:POPL15,Kim-al:APLAS17,Gu-al:PLDI18}.

Iris is another CSL-inspired mechanised verification framework that
has been in development in parallel with FCSL, with an aim to provide
more uniform foundations for reasoning about concurrency~\citep{Jung-al:POPL15}.
Due to the chosen semantic foundations, allowing for impredicativity
in the presence of mutable state (and hence, storable higher-order
procedures)~\citep{Svendsen-Birkedal:ESOP14}, Iris could not have been
implemented as a shallow embedding and, hence, has been encoded as a
deeply-embedded logic.

While that initially has been considered an significant obstacle for
verifying large concurrent programs in Iris, due to a large proof
overhead, the later introduction of Iris Proof Mode
(IPM)~\citep{Krebbers-al:POPL17} fixed this shortcoming, significantly
lowering the entrance threshold for conducting mechanised Iris
proofs~\citep{iris-tutorial}.
This has been achieved in IPM by effectively leveraging Coq's
extensible parsing and proof-by-reflection, and introducing a library
of domain-specific tactics, mimicking, for the sake of an end user of
the framework, standard CSL-style inference rules.
Due to its success, IPM itself has been later generalised to
MoSeL---an extensible proof mode allowing for reasoning not just with
Iris but with any separation-style program logics~\citep{Krebbers2018}.

As frameworks implementing Hoare-style reasoning about concurrency
with pre/postconditions, both FCSL and IPM follow the encoding style
with Dijkstra-style weakest preconditions.
%

\subsection{Frameworks for Language Design and Metatheory}
\label{sec:metaframe}

Several frameworks deal with the challenge of component reuse, one of the challenges from \poplmark (Section~\ref{sec:dsmetatheory}).
Meta-Theory \`{a} la Carte (MTC)~\citep{Delaware2013POPL} is a framework and Coq library that builds
on Data Types \`{a} la Carte~\citep{Swierstra2008} to address challenges in reuse:
\textit{extensibility} of definitions and proofs through algebraic properties that provide
control over the evaluation order, and \textit{modular reasoning} about partial definitions and proofs through algebraic combinators.
Using MTC, the proof engineer can assemble a language from existing components.

MTC does not address extensibility of languages with effects: adding
new effects breaks existing proofs. Modular Monadic Meta-Theory (3MT)~\citep{Delaware2013ICFP}
extends MTC with a methodology and monad library that includes monads
for effects as well as algebraic laws. The methodology and library make
proofs resilient to the addition of new effects to a language.
MTC and 3MT both use algebraic properties to address difficulties with component reuse---algebra in many ways offers
natural abstractions, and those techniques can apply more broadly outside of formal metatheory.

Other notable frameworks for language design include the Fiat~\citep{Chlipala2017, Delaware2015} framework for Coq,
as well as the Hybrid~\citep{felty2012hybrid} framework for
Isabelle/HOL and Coq, which addresses the difficulties of using HOAS with inductive and
coinductive proofs.

\subsection{Frameworks for Other Domains}
\label{sec:otherframe}

Concurrent applications and language metatheory are just two domains for which
verification frameworks are useful. Frameworks assist proof engineers in many other domains.
For example, a few frameworks exist for verifying distributed systems.
Verdi~\citep{Wilcox2015} is a framework for building verified distributed systems in Coq; 
it has been used to build and verify an implementation of the Raft consensus protocol~\citep{Woos2016}.
Disel~\citep{Sergey2017} is a framework for compositional verification of distributed protocols.




\subsection{Future Frameworks}
\label{sec:futureframe}

The expressiveness of the underlying logics of common proof assistants combined with their interactive natures
makes it possible to develop useful frameworks for a variety of domains. We expect that proof engineers
will continue to develop and improve on frameworks that tackle challenges associated with common domains,
as well as build new frameworks to handle challenges associated with new domains for verification as they arise.
In addition, we expect that frameworks will address challenges that current frameworks do not fully address,
such as language extension and component reuse in metatheory. 

While it is natural to apply frameworks to challenges within common domains, we also expect the development of
more general-purpose frameworks building on common proof assistants to address challenges that proof engineers
face independently of domain, or when following specific design principles.

\section{Proof Reuse}
\label{sec:proof-reuse}

Large proof developments may involve redundant efforts that can be time-consuming.
Proof reuse addresses this by repurposing existing proofs as much as possible,
minimizing the amount of redundant work that proof engineers must do.'
Early examples of proof reuse include \textit{proof by analogy}~\citep{curien1995},
the technique of adapting a proof of a theorem to a proof of a related theorem,
and \textit{proof generalization}~\citep{hasker1992generalization}, the technique of
adapting a proof of a theorem to prove a more general theorem.

Proof reuse is the proof engineering analogue to software reuse.
Like software reuse, proof reuse leverages design principles (Section~\ref{sec:designreuse}) and language constructs (Section~\ref{sec:languagereuse}).
In addition, the interactive nature of proof assistants naturally leads to a class of proof reuse technologies less explored
in the software engineering world: automated tooling (Section~\ref{sec:toolingreuse}).
This section samples these approaches.

\subsection{Design Principles for Modularity and Reuse}
\label{sec:designreuse}

Good design principles can help maximize the reusability of existing proofs.
Some of these design principles are natural generalizations of design principles
for software reuse more generally, such as \textit{aspect-oriented software development} (AOSD),
a programming approach that optimizes for separation of concern~\citep{Filman2004}.
Others, like the affinity lemmas from Planning for Change (Section~\ref{sec:des-scale}) are unique to proof engineering.

\paragraph{Design Principles from Software Engineering}
In software engineering, encapsulating behavior can help not only protect
against future changes, but can also help with reusing multiple implementations
of interfaces with the same behavior. Likewise, the interfaces and information hiding recommendations form
Planning for Change are useful not only to protect proofs against future changes, but also
to switch between different datastructure implementations with the same high-level behavior.

The work by~\citet{Delaware2011} attacks the problem of language
metatheory extension, along with the corresponding formalization in a
proof assistant and changing the corresponding
type safety proofs (i.e., progress and preservation theorems), from
the perspective of Software Product Lines (SPL). SPL is an approach to AOSD
that opportunistically reuses software by deriving many different pieces of software
from a common producer. 
\citet{Delaware2011} starts from formalizing a core language, taking a ``core'' Featherweight Java
(cFJ), and considering all further extensions to the language (casts,
interfaces, generics) as features.

What is inherent for the SPL approach is reasoning about composition
and possible interaction between features, expressed by means of an
algebra of feature operators: $\cdot$, $\#$, and $\times$. Introducing
multiple features can lead to an exponential explosion of pairwise
interaction, which, however, is rarely observed in practice, as most
of the features are mutually independent.

In order to enable feature-based decomposition of a language, all its
components (syntax, dynamic semantics, safety proofs etc.) are written
is specific languages, amenable for feature compostions. For instance,
the language syntax and its semantic/typing rules can be extended by
introducing the mechanism of \emph{variation points} (VPs) into the
corresponding grammar productions, premises and conclusions, of the
rules, reusing the intuition of SPL design.

On the implementation side, the modularity of extensions is achieved
by means of reusing Coq's capabilities for higher-order
parametrization: language component definitions are parameterized by
the corresponding variation point contexts. The crux of the technique
is identifying the effect of the VPs to the safety proofs, which are
conducted in a way, parametric with respect to the inductive cases to
be considered. For each specific combination of the features, the
top-level proof dispatches to the proofs from the corresponding
feature module.

The shortcoming of the approach is the requirement, for a core
language, to have a significant foresight when identifying the
appropriate VPs, which provide the opportunity for feature extensions.
While the paper demonstrates how to do it in the context of a
language, whose safety is formalised via the syntactic approach, it
provides little guidance with respect to other ways of stating type
soundness (e.g., via logical relations), neither does it consider
other domains beyond PL design. Overall, the approach seems to be a
bit ad-hoc, which is why further advances in this direction lead to
the creation of the monadic MTC~\citep{Delaware2013POPL} and 3MT~\citep{Delaware2013ICFP} frameworks we have already discussed.


\paragraph{Beyond Software Engineering}
Proof assistants in the LCF family are complex systems with multiple languages at different levels.
Accordingly, reuse in these systems happens not only at the term level, but also at the tactic level.
Designing powerful tactics can maximize reuse of proof scripts to prove different goals (Section~\ref{sec:des-scale}).

Among the recommendations that Planning for Change makes
is the use of \textit{affinity lemmas} that describe relationships between components.
These lemmas show that properties that hold over one component also hold over another,
which facilitates reuse of proofs across components.


\subsection{Language Constructs for Organization and Reuse}
\label{sec:languagereuse}

As in software engineering, proof assistants often provide support for reuse at the language level.
These range from entire languages optimized for reuse to useful constructs built
on existing languages that make reuse easier.
This section describes a sample of languages and language constructs for proof reuse.

\paragraph{Languages for Reuse}
Some languages are designed with the goal of optimizing for proof reuse.
For example, \cite{Felty1994} describes an ITP that is optimized for reuse at the tactic level.

In this system, reuse works by replaying tactics in a new proof setting.
To make this possible, the system automatically generalizes proofs using metavariables.
In contrast, the logical framework \textit{PR}~\citep{Caplan1995} optimizes for reuse at the level of the type theory.
The framework builds on an embedded Hoare logic, adding constructs to the logic
that aid in abstraction and reuse of proof terms.

HoTT (introduced in Section~\ref{sec:equality}) has practical proof engineering applications.
HoTT's univalence axiom gives rise to automatic \textit{transport} of functions and proofs
across type equivalences: to write the same function or proof about two equivalent types, the proof engineer needs only to
write the function or proof over one of these two types,
and then show the equivalence between them.
Cubical type theory~\citep{cohen2016cubical} provides
a computational interpretation of HoTT's univalence, so that it is no longer an axiom.

Proof assistants or extensions to proof assistants built on HoTT or cubical type theory
include Cubical Agda~\citep{cubical-agda}, CoqHoTT~\citep{coq-hott}, and RedPRL~\citep{redprl}.
While these ITPs are relatively new,
we expect that reuse will be easier in these proof assistants.
However, univalence is incompatible with the popular axiom UIP (Uniqueness of Identity Proofs, which states
that all proofs of equality at a given type are equal),
and univalent ITPs present their own difficulties,
so these are not a catch-all solution.
We discuss tooling for transport that does not rely on univalence
at the level of the type theory in Section~\ref{sec:toolingreuse}.

\paragraph{Language Constructs for Reuse}
Even in languages that are not designed with the goal of proof reuse in mind, certain language features can help make proof reuse more tractable.
The modules, type classes, and canonical structures discussed in Section~\ref{sec:desabs} are examples of these features, as are
other mechanisms for inheritance.
In addition, many proof assistants implement subtyping or type
coercions~\citep{barthe1995implicit, aspinall1996subtyping, Saibi1997, luo1999coercive, asperti2007user, callaghan2001implementation, deMoura2015} in various forms,
and these can also help make proof reuse more tractable.

One recent development is the notion of an \textit{ornament}~\citep{mcbride2010}, a programming mechanism for describing relationships
between inductive types that preserve inductive structure. That is, there is an ornament between natural numbers
and lists, and between lists and length-indexed vectors; there is no ornament
between lists and trees, since these types have different inductive structures.
Ornaments allow for the derivation of new types from existing types, and for the automatic lifting of functions and proofs
from each existing type to the corresponding new type. Lifting functions and proofs necessitates some additional automation beyond
the addition of ornaments to a language.
So far, ornaments exist in various forms as deep embeddings in Agda~\citep{Dagand17jfp, Williams2014, ko2016},
and as tooling for proof reuse (Section~\ref{sec:toolingreuse}) in Coq.

The language Cedille makes it possible to define combinators that allow for reuse of functions and proofs
across certain related datatypes without any performance penalty~\citep{Diehl2018}.
Like ornaments, these combinators facilitate reuse between unindexed and indexed versions of types
like lists and vectors. They do not support incompletely determined relations that ornaments support,
such as the ornament between natural numbers and lists (lists have a new element in the inductive case).
Applications of these combinators definitionally reduce in such a way as to facilitate efficient reuse thanks to properties
of the underlying type theory of Cedille.



\subsection{Automated Tooling for Proof Reuse}
\label{sec:toolingreuse}

Since proof assistants typically involve a heavily interactive workflow like the REPL, they lend themselves naturally to automation.
As such, in addition to the design principles and language features for proof reuse found in typical
software engineering projects, there is a body of work that uses automated tooling
to repurpose existing proofs. Section~\ref{sec:metaframe} decribes some frameworks
for component reuse, a kind of proof reuse, in mechanized metatheory. This section describes
other tooling for proof reuse.

\paragraph{Adapting Inductive Proofs} \cite{Boite2004} describes a tactic to adapt proof obligations
to changes in inductive types. This technique constructs and analyzes a dependency graph to determine
when reuse of existing proofs is possible, then reuses existing proofs when possible
and generates new proof obligations for new branches of the proof. \cite{Mulhern06proofweaving} provides a high-level
description of a possible method to synthesize missing proofs for those new obligations using a type reconstruction algorithm,
though it is not currently implemented.

\paragraph{Proof Planning} \textit{Proof planning}~\citep{Bundy1998} is a proof search technique
that uses plans to guide search for proofs with similar structures.
Proof planning can involve the use of \textit{critics}~\citep{ireland1996},
which reuse information from failing proofs to guide search for correct proofs.
While it was originally designed for use with automated theorem provers, it has also reached
interactive theorem provers.
For example, IsaPlanner~\citep{Dixon2003} is a proof planner for Isabelle with support
for rippling~\citep{shah2005}, a technique for automatic induction.
Rippling has also been implemented in an induction automation
tool for Coq~\citep{wilson2010}.

\paragraph{Proof Generalization} Proof generalization tools generalize proofs of a theorem
to obtain proofs of more general theorems. Proof generalization first arose in
the 1990s~\citep{hasker1992generalization, kolbe1998proof, pons1999conception}.
A simple example of proof generalization
is the Coq \lstinline{generalize} tactic, which does basic syntactic generalization.
The Coq documentation demonstrates this tactic on the following proof state~\citep{coq-tactics}:
\begin{lstlisting}
x, y : nat
============================
0 <= x + y + y
\end{lstlisting}
Running \lstinline{generalize (x + y + y)} on this goal produces the following proof state:
\begin{lstlisting}
x, y : nat
============================
forall n : nat, 0 <= n
\end{lstlisting}
The final generated proof term proves the original goal by specialization of this generalized goal.

The generalization technique implemented in Coq's \lstinline{generalize} tactic can handle only
simple syntactic substitution. A few tools can handle more complex transformations.
For example, \cite{Johnsen2004} presents a proof generalization tool for Isabelle
with proof terms which can handle generalizing over dependencies on other theorems,
as well as generalization over functions and types. 
The Coq proof repair tool \textsc{PUMPKIN PATCH} (Section~\ref{sec:repair}) includes an abstraction component
which does not just syntactic generalization, but also type-driven generalization.

\paragraph{Transport}
\textit{Transport} (also known as \textit{transfer}) methods
automatically adapt proofs along relations.
These tools aim to mimic the experience of mathematical proofs on paper, in which simply stating a relation
between two structures (such as an equivalence) can be
enough use theorems about one structure as theorems about the other.

This idea developed both as an extension to the language and as an approach
to automation: \cite{Barthe2001} introduced an extension to dependent
type theory with a computational intepretation of isomorphisms using
rewrites. Around the same time, \cite{Magaud2002} introduced an
automatic method for adapting proofs along binary and unary
representations of the natural numbers.

Since then, there have been many more transport tools handling more than just those two types,
including the \textit{Transfer} and \textit{Lifting} packages~\citep{Huffman2013} for Isabelle/HOL,
and a prototype Coq plugin for transporting proofs across ismorphisms and implications~\citep{ZimmermannH15}.

\textit{Univalent transport} is the particular kind of transport across type equivalences that arises
from HoTT's univalence axiom (see Sections~\ref{sec:equality} and~\ref{sec:languagereuse}).
Equivalences for Free!~\citep{tabareau2018} uses insights from HoTT
to develop and formalize a powerful tool for transporting proofs across equivalences in Coq.
In many cases, it is possible to use this tool to port functions and proofs without any axiomatic dependencies;
in some cases, the tool relies on the functional extensionality axiom.
Thus far, the primary barriers to usability of the tool are the proof burden on the user to configure the automation,
and the inefficiency of the generated functions. Nonetheless, this is a significant step toward
a robust tool for automatic transport in a proof assistant that does not depend on univalence.

The \textsc{DEVOID}~\citep{Ringer2019} Coq plugin automates transport across certain equivalences
that correspond to algebraic ornaments, a particular class of ornaments (Section~\ref{sec:languagereuse}).
\textsc{DEVOID} automatically discovers and proves the equivalences that correspond to these ornaments,
and then transports functions and proofs across those equivalences using a program transformation. \textsc{DEVOID} handles a narrow class of equivalences
relative to Equivalences for Free!, but the functions and proofs that it produces for the cases it can handle are small and efficient in comparison.

\subsection{Packaging and Distributing Programs and Proofs}

Proof assistants projects are software artifacts, and can thus be packaged and distributed in a similar way.
For Isabelle/HOL, the venue for distribution is the Archive of Formal Proofs~\citep{isabelleafp,Blanchette2015}.
Coq uses the OCaml infrastructure around the OPAM package manager to provide a similar collection of packages~\citep{CoqOPAM}.
In principle, executable verified software can also be distributed on these platforms, but can also
use conventional channels, which may raise issues of trust.

\subsection{Future of Reuse}

Component reuse, a form of proof reuse, is an
underaddressed tenant \poplmark. The solutions which do
exist are able to take advantage of common proof structure within the
domain of metatheory. We expect that similar common structure exists
for proofs in other domains, and this area is ripe for the development
of design principles, frameworks, and automated tooling to maximize
reuse.

More generally, we expect to see mainstream proof assistants continue
to integrate language reuse constructs. For example, ornaments are a
promising feature designed specifically for reuse in a dependently
typed language, but most existing implementations require the user to write
programs and proofs in a domain-specific deeply embedded logic.
\textsc{DEVOID} takes some steps toward integrating ornaments into an existing ITP
without an embedding, but it handles only a small class of ornaments and makes some additional
restrictions beyond those that the original ornaments work assumes.
We expect ornaments to integrate more smoothly with existing ITPs in the future. 

We expect that recent developments in HoTT will
fundamentally change how people view proof reuse, 
and that concepts from HoTT will continue to influence the design of proof reuse
tools for other languages. Approaches like Equivalences for Free! have the
benefit of principled design of automation with guaranteed properties,
but do not introduce univalence, and so are not incompatible with
other assumptions that programmers may want in their type theories.
These two views of univalent transport
can continue to evolve alongside one another.



\chapter{Practical Proof Development and Evolution}
\label{ch:development}

As the scale of proof development grows, priorities change. Project management becomes more important, as does dealing
gracefully with changes over time. These demands mirror the concerns the software community has addressed as 
program development has scaled. The proof engineering community has responded similarly, with interfaces,
environments, processes, and tools for effective development that scales.

This section discusses developments in user interfaces and tooling (Section~\ref{sec:user}),
proof evolution (Section~\ref{sec:evolve}), user productivity and cost estimation (Section~\ref{sec:productivity}),
and mining proof repositories (Section~\ref{sec:mining}) that address these concerns.


\section{User Interfaces and Tooling for User Support}
\label{sec:user}

Most early proof assistants shipped with a very simple user
interface: the Read-Eval-Print Loop (REPL). This interface reads in user-written
expressions in the proof assistant language, evaluates those expressions,
then prints a result or error for the user.

User interfaces for proof assistants have come a long way from the 
REPL. The support that these REPLs provide users is minimal, and so soon after
their development, many techniques to ease interaction with REPLs arose.
While the interfaces from earlier eras still see common use, we are now entering an era
of interaction that emphasizes full integrated development environments (IDEs),
with support for project management and for asynchronous development.

In parallel to this evolution of user interfaces (Section~\ref{sec:uievolve}), we are seeing an increase
in specalized interfaces (Section~\ref{sec:specialized}),
usability analysis of user interfaces (Section~\ref{sec:usability}),
and advanced tooling for user support (Section~\ref{sec:usersupport}).
We expect these traditions will merge and drive the future of interaction (Section~\ref{sec:futureinterface}) 
with proof assistants.

\subsection{The Evolution of User Interfaces}
\label{sec:uievolve}
We can think of proof assistant user interfaces as evolving in three generations:

\begin{enumerate}
\item Generation I: The REPL
\item Generation II: Separation of Concerns
\item Generation III: Full IDEs
\end{enumerate}

\paragraph{Generation I: The REPL}


The REPL was the earliest form of interaction with the proof assistant.
For example, the description of Stanford LCF~\citep{Milner1972b} calls the proof process a
``conversation between the user and the computer.'' The LCF user writes commands,
which the computer evaluates and replies to with feedback such as new goals.
In part of the example from the LCF description, the user cuts an inline lemma:
\begin{lstlisting}
*****GOAL f $\subset$ g;
\end{lstlisting}
The computer then responds acknowledging the new goal:
\begin{lstlisting}
NEWGOAL #1 f $\subset$ g 
\end{lstlisting}
The user tells the computer to prove this goal inductively:
\begin{lstlisting}
*****TRY 1 INDUCT 1; 
\end{lstlisting}
The computer responds with two intermediate goals, a base case and an inductive case:
\begin{lstlisting}
NEWGOAL #1#1 UU $\subset$ g 
NEWGOAL #1#2 fun(f1) $\subset$ g ASSUME f1$\subset$g 
\end{lstlisting}
The user then proves those subgoals, then uses the inline lemma to prove the original result.

Many ITPs followed in this tradition and introduced command line REPLs.
Examples of command line REPLs include the \lstinline{coqtop}~\citep{coq-commands} command for Coq and
the \lstinline{hol}~\citep{hol-interact} command for HOL.		
Some of these tools are still accessible even when graphical interfaces exist.
For example, Coq still exposes its \lstinline{coqtop} command,
in spite of the existence of the graphical interfaces CoqIDE~\citep{coqide} and Proof General~\citep{Aspinall2000}.

\begin{figure}
\center
\includegraphics[width=0.7\linewidth]{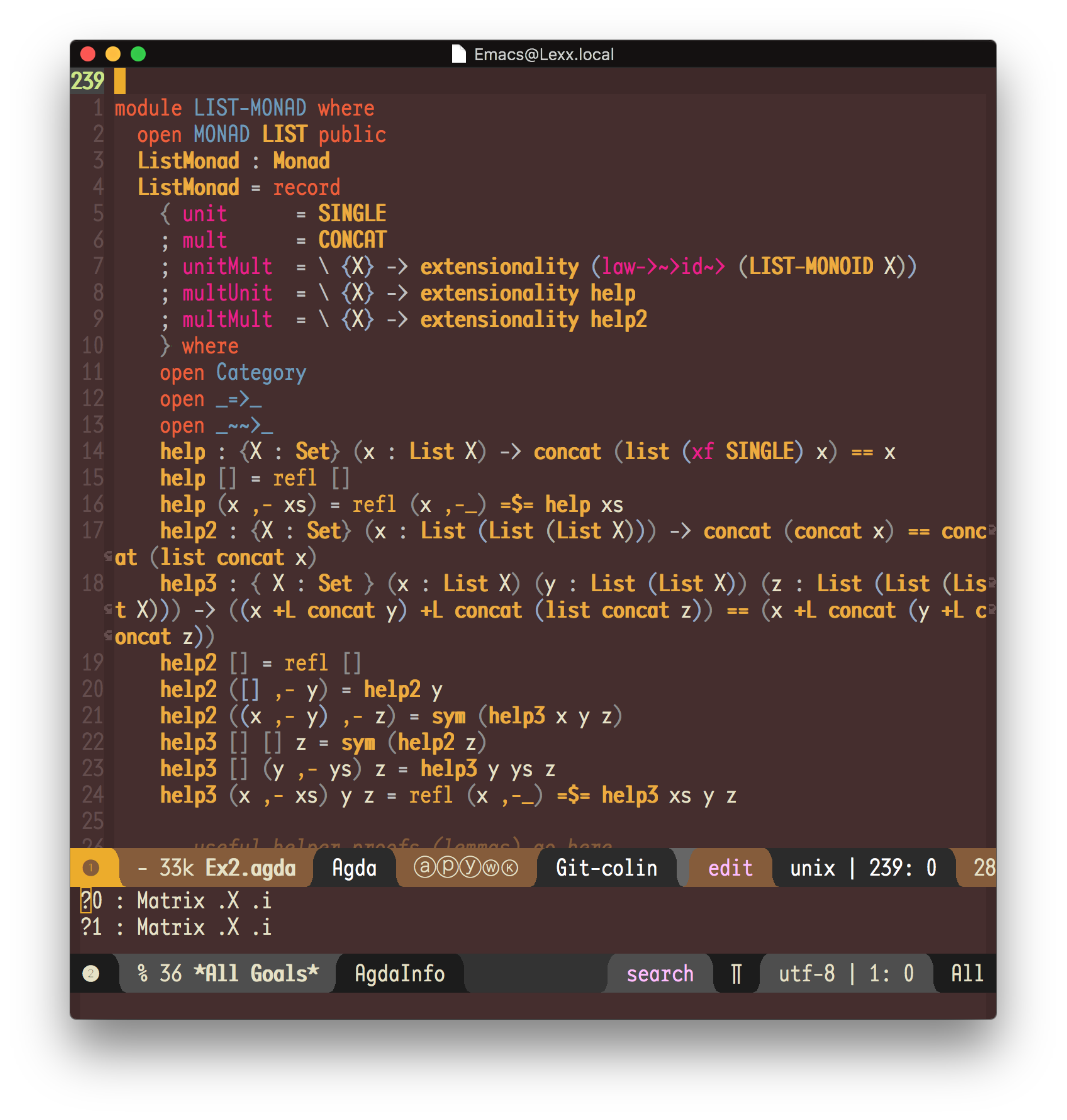}
\caption{Agda Emacs mode (from \cite{barret-agda})}
\label{fig:agda-modes}
\end{figure}

However, not all ITPs followed in this tradition.
Nuprl, for example, was distributed with a graphical interface from the start~\citep{Constable1986}.
Even those ITPs that followed in this tradition sometimes later diverged.
For example, Agda's \lstinline{--interactive} option to interact with the REPL directly 
is no longer supported; Agda interaction happens through an Emacs~\citep{emacs} or Atom~\citep{atom} mode which calls out directly to the 
backend theorem prover~\citep{agda-editing}. Figure~\ref{fig:agda-modes} shows an examples of the Agda Emacs mode.
Isabelle/HOL has recently done away with its REPL; the default interface is now Isabelle/jEdit~\citep{Wenzel2012},
which instead builds on Isabelle/PIDE~\citep{Wenzel2014}.


\paragraph{Generation II: Separation of Concerns}

The 1990s saw a surge in the release of interfaces for ITPs decoupled from the proof checker,
typically communicating with the system through a protocol.
In some ways, this was a natural path of evolution from the way that users
typically interacted with the REPLs for existing ITPs.
While the REPL was typically exposed through a command line tool,
it was common to instead use multiple Emacs buffers, 
one for development and for the proof assistant top-level,
and to copy definitions between the two.
This approach is still used in some modern proof assistants
such as HOL~\citep{hol-tutorial}.

This mode of interaction naturally led to the development of Emacs modes which interact with the REPL or theorem prover backend.
For example, Isamode~\citep{isamode} for Isabelle99 was an Emacs mode for Isabelle which smoothed interaction with the REPL.
The HOL4 Emacs mode~\citep{hol4-interact} is still used to this day.
The Agda Emacs mode, which interacts with the Agda backend, is similar in spirit but contains more advanced functionality; 
it allows the user to, for example, define holes in terms and fill those holes in later in development.
Idris includes an Emacs mode~\citep{mehnert2014tool} for interacting with the REPL, which inspired by both the Agda Emacs mode and Proof General.

Other interfaces beyond Emacs modes communicate with the backend theorem prover
or REPL in this style. For example, ALF, a predecessor to Agda, included a window-based interface for communicating 
with the backend~\citep{altenkirch1994user}. The lightweight interface 
TkHOL~\citep{tkhol} for HOL also follows in this style.
\cite{BERTOT1998} describes a generic approach for building an interface that communicates with
the ITP using a protocol, inspired by the early Coq user interface \textsc{CtCoq}.

\begin{figure}
\begin{minipage}{0.45\textwidth}
\includegraphics[width=1.0\linewidth]{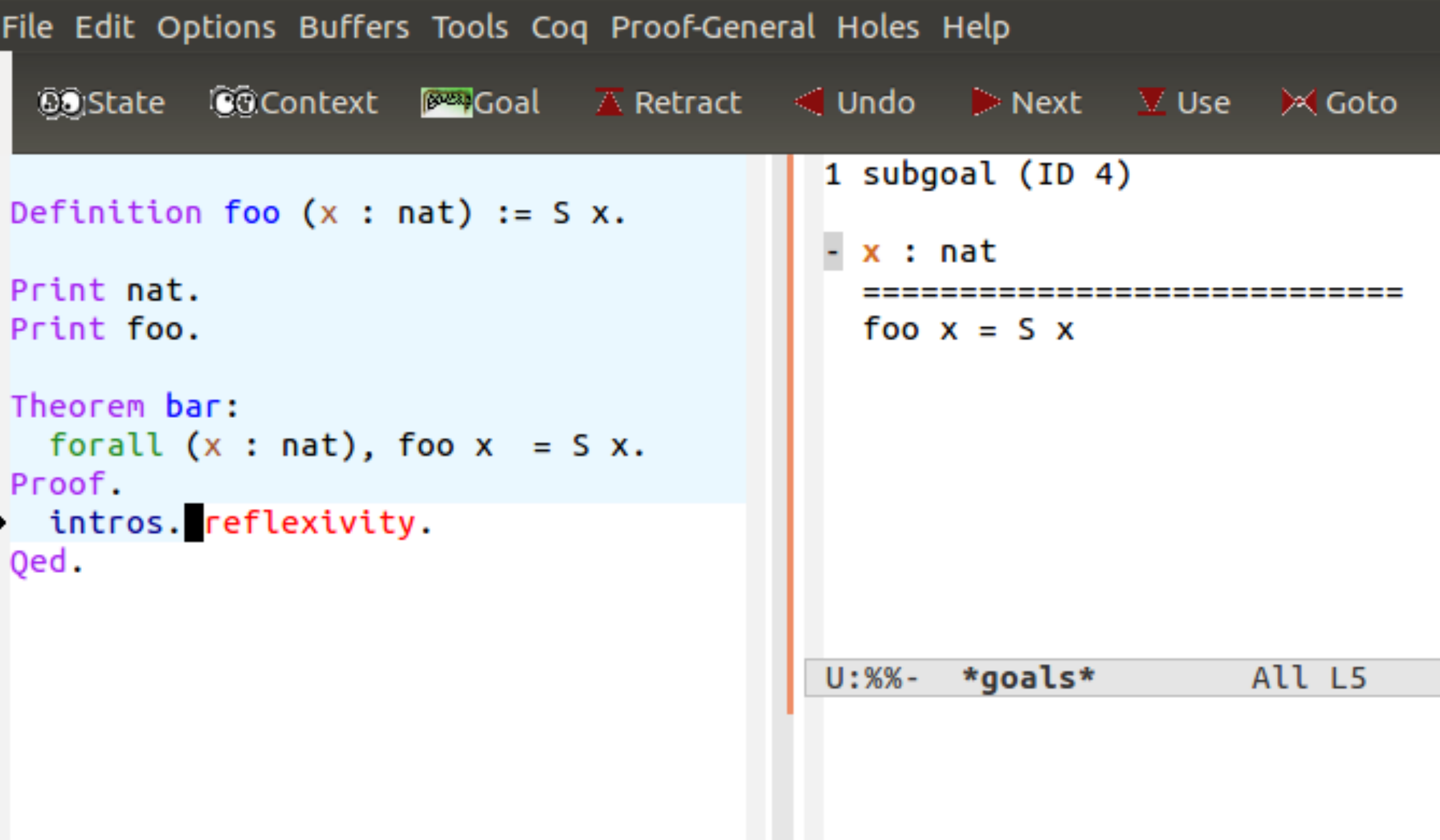}
\end{minipage}
\hfill
\begin{minipage}{0.5\textwidth}
\includegraphics[width=1.0\linewidth]{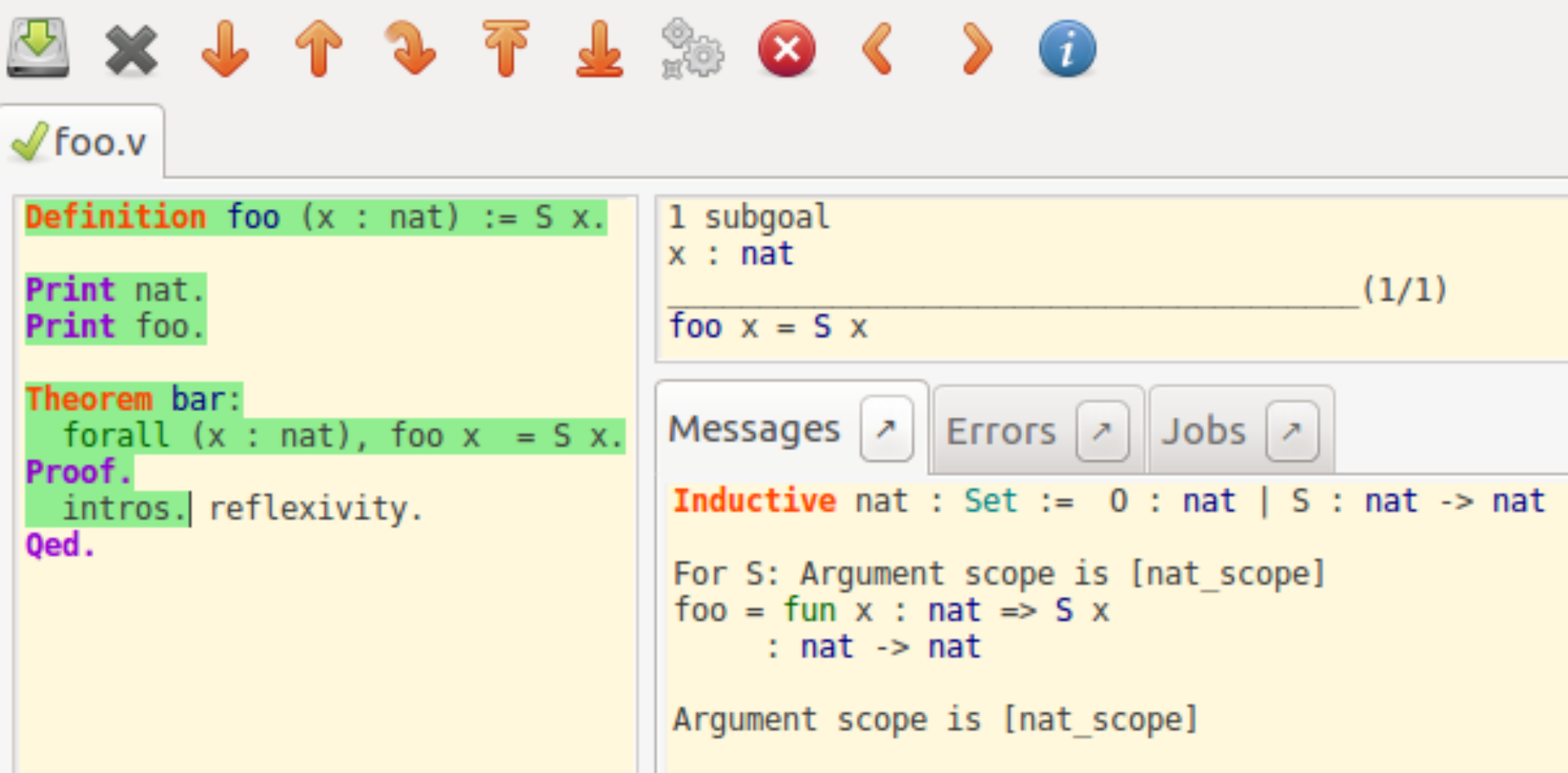}
\end{minipage}
\caption{Proof General (left) and CoqIDE (right) for Coq}
\label{fig:proofgeneral}
\end{figure}

In some cases, these interfaces were entirely
independent of the underlying proof assistant.
One notable example of such an interface from this generation is the Emacs extension 
Proof General~\citep{Aspinall2000}, an
interface for proof development that supports multiple 
proof assistants. Proof General has seen widespread use,
especially within the Coq community. 
While Proof General best supports Coq, it also has support for
LEGO, PhoX, and an old version of Isabelle, as well as experimental support
for other proof assistants~\citep{proofgeneralwebsite}.
It is simple yet easily extensible, both to support new 
proof assistants and to add new functionality for existing proof assistants. 
Company-Coq~\citep{CompanyCoq2016} for example, 
extends Proof General with many new features for Coq, including improved autocompletion, and integration of documentation. 


Following the success of Proof General, Coq released
the lightweight interface CoqIDE~\citep{coqide} as part of Coq 8.0~\citep{coqide-commit}.
Its main selling point was speed: It claimed to be faster than Proof General.
In addition, CoqIDE's native support for Coq means that it is
always maintained and distributed with new versions of Coq, imposing minimal overhead on
users. Figure~\ref{fig:proofgeneral} shows CoqIDE and Proof General side-by-side for Coq.

Both third-party interfaces and native interfaces from this generation continue to be
popular to this day. This separation of concerns has also inspired a new
generation of specialized interfaces (Section~\ref{sec:specialized}) for proof assistants.



\paragraph{Generation III: Full IDEs}

The third generation of user interfaces coincides with the rise of proof development
of large projects and the corresponding increase in concern for good proof engineering
support. Interfaces from this generation focus on scaling to large developments.
For example, early user interfaces did not support \textit{asynchronous development}: 
they did not allow the user to run the proof checker on some proofs while modifying others.
Early user interfaces also did not have support for project management, and so were
not truly full-scale IDEs.

\begin{figure}
\center
\includegraphics[width=0.7\linewidth]{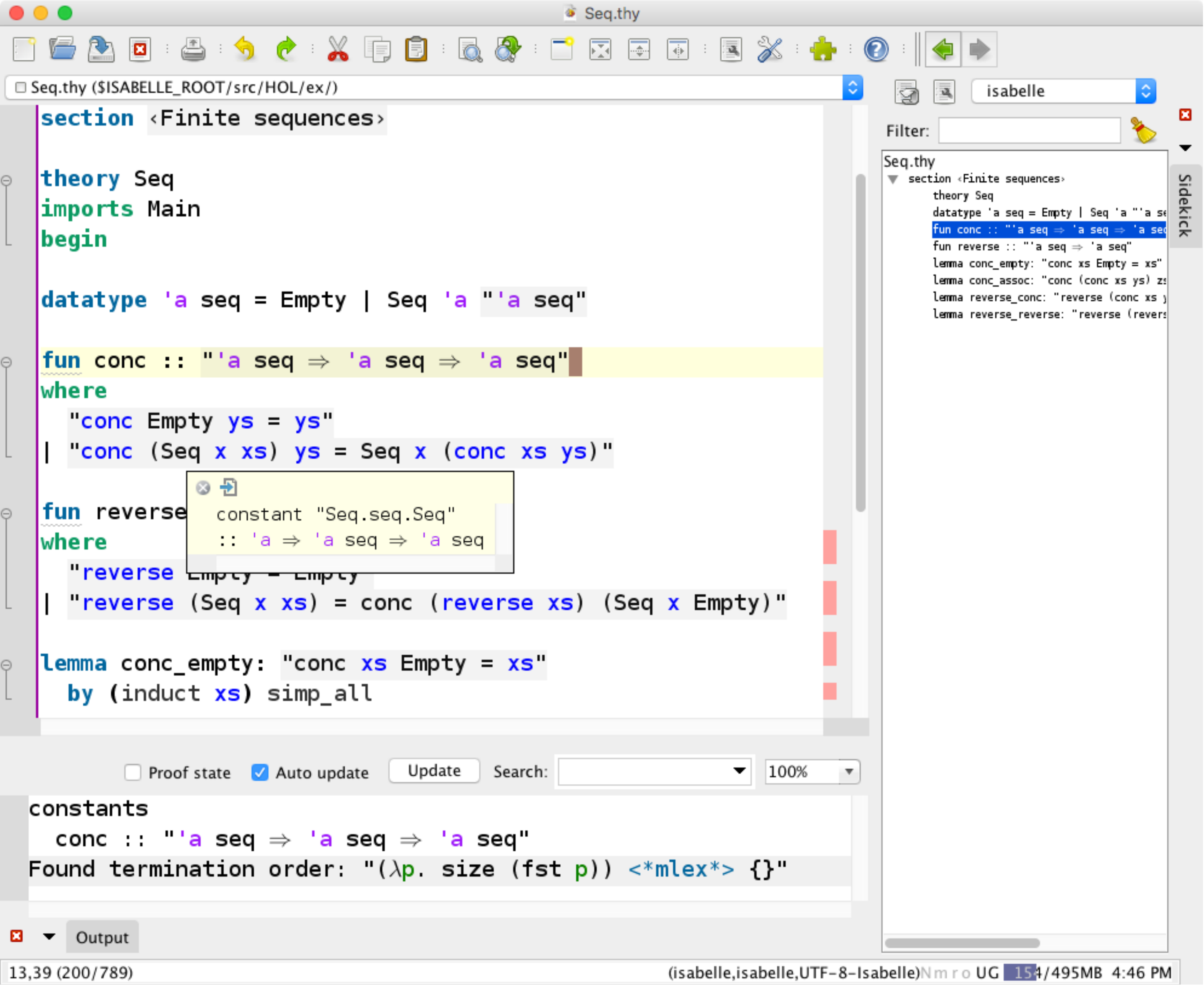}
\caption{Isabelle/jEdit (from \cite{isabelle-jedit})}
\label{fig:pide}
\end{figure}

Many IDEs in the latest wave of development address these concerns.
Coqoon~\citep{Faithfull2016}, for example, is an IDE for Coq built on Eclipse with support for both project management and asynchrony. 
The PIDE framework, originally developed
for use with the Isabelle IDE Isabelle/jEdit~\citep{Wenzel2014}, also supports asynchronous development;
Isabelle/jEdit is shown in Figure~\ref{fig:pide}.

PIDE is ultimately indifferent to the backend theorem prover;
\cite{Wenzel2013} and \cite{Barras2015} describe interfaces built on PIDE for Coq.
PIDE also has additional interfaces in Isabelle aside from the default
Isabelle/jEdit, including Isabelle/VSCode~\citep{isabelle-vscode}, an Isabelle plugin for Visual Studio~\citep{vscode}.
PIDE has seen enough success that Isabelle has done away with its REPL entirely.

Like Isabelle, Lean also has an IDE implemented as a Visual Studio plugin~\citep{lean-vscode}.
This IDE communicates with the Lean server, and supports 
incremental compilation and proof checking, debugging, documentation, and batch execution.

Many existing proof assistant interfaces have integrated features from this generation.
For example, CoqIDE now supports asynchrony~\citep{coqide}. It remains
to be seen to what extent full-scale IDEs for proof assistants will continue to evolve
and to grow in popularity.


\subsection{Specialized Interfaces}
\label{sec:specialized}

The separation of concern from Generation II interfaces for proof assistants inspired the development
of specialized interfaces. For example, web-based interfaces require minimal setup and installation,
and so are thought to be less intimidating to new users, especially students. 
Many web-based interfaces are built with students as the key audience to address concerns students have about installing and using
heavyweight IDEs. Examples of web-based interfaces for proof development include 
ProofWeb~\citep{kaliszyk2007web}, jsCoq~\citep{Gallego2016}, and PeaCoq~\citep{peacoq}.
The Lean 2 tutorial~\citep{lean-tutorial} uses the Lean.JS~\citep{lean-js} web interface for Lean to provide an interactive
learning experience directly in the browser.

Proof assistant users sometimes note that the experience of writing proofs has game-like elements.
The interactive nature of a proof assistant, for example, is similar to interacting with an adversary in a game.
There is some work on gamification of proofs that reifies this intuition into the interface itself.
In these games, players can generate program annotations~\citep{dietl2012verification},
write natural deduction proofs~\citep{lerner2015polymorphic},
and identify inductive invariants~\citep{bounov2018inferring}, all the while having low-level details of these proofs
abstracted away from them. While these games are not interfaces for well-known ITPs like Coq and Isabelle,
they may help with tasks that can assist users in writing proofs, such as finding inductive invariants. Applying
this same intuition to build new interfaces for commonly-used proof assistants may help make them more accessible
to non-experts in the future.

\subsection{Interface Usability Analysis}
\label{sec:usability}

Traditional software engineering tools and interfaces are often subject to usability analyses according to the conventions in human-computer interaction (HCI). There are some similar analyses related to proof assistants. \cite{Aitken1998} propose a three-layer model to account for user interaction with a proof assistant, and perform an empirical study which concludes that there is support for the view of ``proof as programming'' for proof assistant interaction, rather than ``proof by pointing''~\citep{Bertot1994} and ``proof as structure editing''. \cite{Kadoda1999} analyze the usability of theorem provers in a cognitive framework by using questionnaires. \cite{Aitken2000} analyze errors in proof attempts. \cite{HOFM2014raey} use focus groups to evaluate usability of proof assistants, finding that users prefer proof assistants that produce intuitive proofs, can present comprehensible proof steps, and provides a convenient interface.

\subsection{Tooling for User Support}
\label{sec:usersupport}


In parallel with the evolution of user interfaces, recent years have seen
an emphasis on tooling to help users
with proof development. These features are more useful than ever because of the
advent of Generation II user interfaces that are not tightly tied to the REPL.

Many of the user support features that are now arising for proof engineering
echo similar features that already exist in languages with more mature IDEs. 
%
%
%
%
%
For example, languages with more mature IDEs often integrate refactoring tools into those IDEs;
now that proof assistant interfaces are maturing, interfaces
with refactoring support such 
the Coq interface CoqPIE~\citep{Roe2016} are beginning to emerge.
We discuss more refactoring tools in Section~\ref{sec:refactor}.

In addition, new techniques are extending the reach of user support features to
support the challenges particular to proof development. 
For example, one common challenge in proof engineering is efficiently finding relevant datatypes and proofs
Many proof assistants distribute tools for this by default.
For example, Coq includes the \lstinline{Search} command,
which \ssreflect~\citep{Gonthier2010} extends;
Isabelle includes the \lstinline{find_theorems} and \lstinline{find_consts} commands.
This challenge has also inspired several external tools, including the web-based tool Whelp~\citep{asperti2004content} for Coq,
upon which Matita builds~\citep{asperti2007user}.

Machine-learning techniques can also help with challenges in proof developments, for example by suggesting hints to users.
Recent tooling of this flavor includes the ML4PG~\citep{Komendantskaya2012} extension to Proof General, which
uses machine learning to suggest hints during proof development,
and ACL2(ml)~\citep{Heras2013}, which uses machine learning to suggest auxiliary lemmas
for ACL2 development. Nagashima and He propose a proof method recommendation system for
Isabelle/HOL based on machine learning, which is trained on large proof corpora~\citep{Nagashima2018}.

Unlike traditional software development, proof development with ITPs
often involves significant interaction with automation. Accordingly, one question that many tools
explore is the ideal user experience for interacting with automation.
The web-based IDE PeaCoq~\citep{peacoq}, for example, has extra support for tactic previews and context management.
Matita~\citep{asperti2007user} includes special support for contextual term manipulation,
and for understanding the execution of tactical-like chains of tactics.

Another common problem in proof development is that the proof engineer may accidentally
state a false theorem, or may be unsure if a stated theorem is true.
When a stated theorem is false, it can be difficult to determine that the theorem is actually false;
the proof engineer may instead think his inability to prove the theorem is due to his own shortcomings.
Hammers and other general-purpose automation (Section~\ref{sec:autotactics}) can help a proof engineer discharge simple proof obligations
and quickly determine that a theorem is true; the proof engineer can then reprove the theroem in a different way if desired.
Property-based testing tools like Quickcheck for Isabelle~\citep{Bulwahn2012b} and QuickChick~\citep{lampropoulos2017generating, Paraskevopoulou2015} 
for Coq can help users identify counterexamples to false properties.



\subsection{Future of User Interfaces}
\label{sec:futureinterface}

Many of the Generation I interfaces
still exist today. We expect some of these will continue to exist, since they are
useful when resources are limited. However, there is a growing trend of moving away from the REPL in some ITPs
such as Isabelle/HOL; perhaps more ITPs will move in that direction.

The interactive nature of the REPL makes it simple to collect fine-grained data
on how proof engineers develop code. For proof assistants that are backed by a REPL,
collecting this data could help with the development of better tooling to support
proof engineers during development.
Similarly, while there is some empirical information on how proof engineers interact with different user interfaces,
this is still a lot more ground to cover. Even collecting simple information like the number of users of each interface for each proof assistant over time may 
help gauge the impacts of different design decisions.
More user studies on interacting with proof assistants could also help pave the way for more useful interfaces.

We expect that the separation of concerns emphasized with Generation II interfaces
has had a strong influence and will likely continue to have a strong influence. 
Separation of concerns and extensibility may be part of why
Proof General has continued to be successful after so many years.
In many ways, this mirrors the success
that is seen in successful IDEs for software engineers, such
as Eclipse~\citep{eclipse} or IntelliJ IDEA~\citep{intellij}. 
We expect
that future developments will continue to work on separation of concerns
and extensibility, with better plugin systems for IDEs to support more features
with minimal effort for the interface developer and for the proof engineer.

Emacs has played a crucial role in the history of the development of IDEs for ITPs,
with many early interfaces implemented as Emacs modes. Recent years have seen the development
of IDEs as Visual Studio plugins for Isabelle/HOL and for Lean.
In the future, perhaps more proof assistants will implement IDEs as plugins for existing IDEs such as
Visual Studio and Eclipse, and perhaps these existing IDEs will play 
a similar role to Emacs in the continued development of proof assistant IDEs.

We also expect that project management will continue to grow in importance
as large proof developments become more common. Few interfaces for proof assistants
currently have strong support for project management; this is an area ripe for
improvement. Better integration of build tools and continuous integration tools
can greatly improve development experience. 

Finally, we expect more productivity tools to emerge, like the refactoring
tools that already exist, and for these tools to be integrated into IDEs. For example, most mature IDEs for existing languages 
have strong support for debugging. Debugging tools for proof assistants, on the other hand,
are few and far between.
Better plugin systems for interfaces could help minimize the friction in supporting
these features at the IDE level. 

\section{Proof Evolution}
\label{sec:evolve}


Programs change over time, and so proofs about programs must change with those programs.
This concern is raised in the Social Processes~\citep{DeMillo1977} critique of program verification
as a barrier for the verification of real programs. 
This barrier has been realized in real developments; a review~\citep{Elphinstone2013} of the evolution of the seL4 verified 
OS microkernel~\citep{Klein2009}, for example, notes that while
customizing the kernel to different environments may be desirable, 
``the formal verification of seL4 creates a powerful disincentive to changing the kernel.''
\cite{leroy2012} motivates and describes updates to the initial CompCert memory model
that include changes in specifications, automation, and proofs~\citep{leroy-mem-2010}.

Changes in programs and proofs are not always in the proof engineer's control---
updating a standard library, for example, can lead to proofs in client code failing
during regression proof checking (Section~\ref{sec:regression}).
\textit{Reactive} approaches to proof evolution address changes that occur outside of the proof engineer's control.
These approaches contrast with and are complementary to \textit{proactive} approaches
that address brittleness ahead of time, 
such as the design principles discussed in Section~\ref{sec:proofdes}.

Consider, for example, a Coq proof that uses the \lstinline{intros} tactic. If the user does not
pass identifiers to \lstinline{intros}, then Coq automatically chooses hypothesis names. Small changes to the theorem statement or 
to the proof can change the names of the hypotheses that Coq chooses, which can make proofs that refer to those hypotheses brittle.
We briefly discussed two proactive approaches to this problem in Section~\ref{sec:des-scale}:
explicitly choosing identifiers to pass to \lstinline{intros}, and writing tactics that do not refer to these hypotheses at all.
In contrast with these proactive approaches,
the IDE CoqPIE~\citep{Roe2016} automatically renames references to hypotheses in proofs to work
around this problem reactively.

The renaming functionality of CoqPIE is an example of proof refactoring (Section~\ref{sec:refactor}),
a reactive approach to proof evolution. Proof repair (Section~\ref{sec:repair})
is a similar reactive approach to proof evolution. The main distinction between these two approaches 
is that proof refactoring is semantics-preserving, while proof repair need not be. 
Nonetheless, these technologies often overlap.

\subsection{Regression Proving}
\label{sec:regression}

Regression proving is the process of rechecking proofs after a change to a verification project, mirroring \emph{regression testing} for software projects. For large-scale projects, regression proving may require considerable machine time---from tens of minutes and hours up to several days. This can negatively affect the productivity of proof engineers. Absent domain- and context-specific knowledge, as in proof refactoring, the two main techniques to speed up regression proving are proof-checking parallelization~\citep{Wenzel2013MultiProcessing,Barras2013} and proof selection~\citep{Celik2017}.

Support for parallelization varies in degree and kind among proof assistants. Isabelle leverages the support for threads in its host compiler, Poly/ML, to spawn proof checking tasks processed by parallel workers. Using a notion of \emph{proof promises}, proofs that require previous unfinished result can proceed normally and become finalized when extant tasks terminate~\citep{Wenzel2013}. Isabelle also includes a build system with integrated support for checking of proofs and management of parallel workers. The lack of native threads in OCaml prevents similar low-cost fine-grained parallelism for Coq. However, spawning parallel operating system processes is still possible, and such processes can be leveraged for both file-level parallelism and to check fine-grained proof tasks~\citep{Barras2015}. Lean supports fine-grained parallel proof checking~\citep{deMoura2015}. Compared to parallelization of test execution for software projects, checking a proof is deterministic and has no side-effects detrimental to checking other proofs. 

A \emph{regression proof selection} (RPS) technique limits the scope of regression proving to those proofs that are affected by a change to a project. While selection at the file level (modulo file dependencies) is broadly supported via build systems such as \texttt{make}, only some proof assistants such as Isabelle and Coq supports selection of individual proofs; this is made possible by support for \emph{asynchronous} proof checking~\citep{Wenzel2014,Barras2015}. Celik~et~al. proposed an RPS technique for Coq that combines dependency analysis at the file and proof levels~\citep{Celik2017}; their tool implementation, dubbed iCoq, compares checksums of files, terms, and proof scripts to locate and run affected proofs sequentially. In an evaluation on the revision histories of several large-scale Coq projects, iCoq was up to 3 times faster than using conventional \texttt{make}-style checking with a persistent store, and up to 10 times faster than conventional checking when each revision is checked from a clean slate.

\cite{Palmskog2018} defined a taxonomy of regression proving techniques for proof assistants that include both parallelism and selection. Along one axis, they consider parallelization at the file and proof granularity. Along the other axis, they consider selection of files and proofs. Their most sophisticated technique combines proof selection and fine-grained parallelization, and consisistently outperforms other techniques on the revision histories of several Coq large-scale projects.

\cite{WenzelScalingIsabelle, WenzelFurtherScalingIsabelle} outlined how to scale Isabelle for large projects using both parallelism and other techniques.

\subsection{Proof Refactoring}
\label{sec:refactor}

\textit{Refactoring} is the restructuring of code in a way that preserves semantics~\citep{opdyke1992};
\textit{proof refactoring} is the refactoring of proofs~\citep{WhitesidePhD}.
Proof refactoring tools help automate this process, propogating a single change throughout the proof development.
Like program refactoring tools, proof refactoring tools can help keep developments
maintainable as they change over time~\citep{Bourke12}. In that way, it is possible to consider refactoring tools 
as both proactive and reactive approaches to proof evolution, though we consider them here through a reactive lens.

Some proof assistants expose tactics (Section~\ref{sec:tactics}) or proof languages (Section~\ref{sec:structuredprooflangs})
in which the proof engineer can write high-level proof scripts to guide proof search.
Some proof refactoring tools refactor these proof scripts directly.
One such tool is \textsc{POLAR}~\citep{Dietrich2013}, a generic framework for proof script refactoring. 
\textsc{POLAR} is instantiated with two languages,
both of which are based on Isabelle/Isar~\citep{Wenzel2007isar}:
Hiscript~\citep{WhitesidePhD}, a language with support for refactoring,
and $\Omega$\textsc{script}~\citep{dietrich2011}, a language with support for proof planning (Section~\ref{sec:toolingreuse}).
Refactoring in \textsc{POLAR} works through a combination of rewrite rules that operate over
a graph representation of the underlying language. 
\textsc{POLAR} implements ten kinds of refactorings by default,
and also supports custom refactorings.
It guarantees that all lemmas that go through before
the refactoring continue to go though after the refactoring.

Some proof refactoring tools focus on specific refactoring tasks that are common in proof development.
For example, Levity~\citep{Bourke12} is a proof refactoring tool for an old version of Isabelle/HOL that automatically
moves lemmas to maximize reuse. The design of Levity is informed by experiences with two large proof developments.
Levity addresses problems that are especially pronounced in the domain of proof refactoring, such as the
context-sensitivity of proof scripts. Tactician~\citep{adams2015} is a refactoring tool for proof scripts in HOL Light
that focuses on refactoring proofs between sequences of tactics and tacticals.
 
There is little work on refactoring proof terms (Section~\ref{sec:proof-objects}) directly. This is the main focus of Chick~\citep{robert2018}, 
which refactors terms in a dependently-typed functional language similar to Gallina. To use Chick, the proof engineer applies some refactorings.
Chick then uses a program differencing algorithm to determine the changes to make elsewhere in the program,
then makes those changes. Chick supports insertion, deletion, modification, and permutation of subterms.
Similarly, RefactorAgda~\citep{wibergh2019} is a refactoring tool for a subset of Agda that operates directly
over Agda terms. RefactorAgda supports many changes, including changing indentation, renaming terms, moving terms, converting between implicit and
explicit arguments, reordering subterms, and adding or removing constructors to or from types; 
it also documents ideas for supporting other refactorings, such as adding and removing arguments and indicies to and from types. 

For both Chick and RefactorAgda, only some of these changes are semantics-preserving. Adding a new index to a type, for example, does not preserve the 
semantics of the original program. Accordingly, these tools can be viewed as both refactoring and repair tools,
though the algorithms that they use are syntactic.

A natural integration point for a proof refactoring tool is at the level of a platform or an IDE.
The Coq IDE CoqPIE~\citep{Roe2016} for Coq takes this approach for refactoring proof scripts.
CoqPIE includes a \textit{Replay} button which steps through the proof while renaming
any changed hypothesis names. CoqPIE can also automatically split out intermediate goals from a proof into separate lemmas. 
There are plans to support more refactoring functionality in CoqPIE in the future.

\subsection{Proof Repair}
\label{sec:repair}

\textit{Program repair}~\citep{Monperrus2018} is the automatic patching of programs to fix bugs;
\textit{proof repair} is program repair for proofs. Proof repair tools automatically fix broken proofs.
Recent lessons from a review of a certain class of program repair tools~\citep{Qi2015}
highlight why proof repair is a particularly good domain of program repair.
The review demonstrates that many existing tools produce incorrect patches.
Among the recommendations the authors make to remedy this is the suggestion that program repair tools
make use of extra information such as specifications, code from other applications, or example patches when generating patches.

In proof repair, a specification is always available: the theorem the repaired proof ought to prove.
Some proof repair tools take this a step further and make use of additional extra information, such as examples patches.
One such tool is \textsc{PUMPKIN PATCH}~\citep{Ringer2018}, a proof repair tool for Coq that generalizes example patches. 
\textsc{PUMPKIN PATCH} takes as inputs an old proof and a new proof that addresses some change in specification.
From those, it identifies a reusable patch that describes the change in specification;
for the kinds of changes \textsc{PUMPKIN PATCH} can currently handle, this patch is a Gallina function.
The proof engineer can then use this patch to patch other proofs broken by the change in specification.
\textsc{PUMPKIN PATCH} has only preliminary tooling~\citep{pumpkin-git} for applying patches automatically,
and currently handles only simple changes.

Chick~\citep{robert2018} was developed in parallel to \textsc{PUMPKIN PATCH}, and has a similar workflow:
Chick takes a set of example changes supplied by the programmer,
and uses a program differencing algorithm to determine the changes to make elsewhere.
Unlike \textsc{PUMPKIN PATCH}, Chick also applies the changes it finds.
However, Chick does this using a syntactic algorithm that handles only simple transformations;
for this reason, it presents itself primarily as a refactoring tool, even though the changes it makes may
not preserve semantics. The refactoring tool RefactorAgda~\citep{wibergh2019} similarly decribes some semantics-changing
repairs for a subset of Agda.

While proof repair is analogous to program repair, it was born out of traditional proof reuse (Section~\ref{sec:proof-reuse}).
For example, \textsc{PUMPKIN PATCH} discovers patches which help adapt a proof of a theorem to a proof of a related theorem,
and can so be thought of as a tool to assist in proof by analogy~\citep{curien1995}.
Similarly, the proposed proof weaving~\citep{Mulhern06proofweaving} method to automatically satisfy new obligations 
generated in response to changes in inductive types can be viewed as a proof repair technique.
Proof planning critics~\citep{ireland1996} can also be viewed as a technique for proof repair.

New technologies continue to make proof repair more feasible.
GALILEO~\citep{chan2011galileo} is a tool build on Isabelle for identifying and repairing faulty ontologies in response to contradictory
evidence; it has been applied to repair faulty physics ontologies, and may have applications more generally for mathematical proofs.
GALILEO uses repair plans to determine when to trigger a repair, as well as how to repair the ontology.

Knowledge sharing methods~\citep{gauthier2014} match concepts across
different proof assistants with similar logics and identify isomorphic types,
and may have implications for proof repair.
Later work uses these methods in combination with HOL(y)Hammer to
reprove parts of the standard library of HOL4 and HOL Light using combined knowledge 
from the two proof assistants~\citep{Gauthier2015}. 
More recently, this approach has been used to identify similar concepts
across libraries in proof assistants with different logics~\citep{gauthier2017}.
These methods combined with automation like hammers may help the proof engineer 
adapt proofs between isomorphic types, and may have applications
when repairing proofs even within the same logic, using information from different 
libraries, different commits, or different representations of similar types.



\subsection{Future of Proof Evolution}

There is a lot of room for work in proof evolution---only a few techniques exist so far, 
many of which emerged in parallel. We expect these reactive approaches to continue to evolve alongside proactive
approaches like design principles, as the two approaches are complementary.
Proof evolution can help with changes that occur outside of the programmer's control,
such as changes in dependencies (examples of this can be found in \cite{Ringer2018}) and changes that are difficult to protect against
even with informed design (examples of this can be found in \cite{Klein2014}).

Ideally, proof evolution tools ought to integrate naturally with the workflows of proof engineers,
for example through integration with existing tactic or proof languages, or through IDE or continuous integration support.
While some proof evolution tools focus on this already and can offer useful insights, this can be challenging.
For example, refactoring Ltac proof scripts can be difficult, since the semantics of Ltac are not well-defined;
Ltac2~\citep{ltac2} may simplify this in the future.
\cite{robert2018} discusses the challenges involved in refactoring proof scripts in more detail.
\cite{Ringer2018} also discusses the challenges of workflow integration, along with other open problems in proof repair.
We expect to see more emphasis on addressing these challenges in the future.

There is only preliminary work exploring how much of the work from existing refactoring and repair tools
for programming carries over to the domain of proof assistants. It is worth exploring in more detail
which challenges are unique to this domain. For example, \cite{Qi2015} provides several recommendations 
for how program repair tools can make use of extra information such as examples to make searching for patches
more feasible; \textsc{PUMPKIN PATCH} and machine learning tools use examples for this purpose already.
Future proof refactoring and repair tools can similarly learn from those recommendations.

One tempting use case for proof refactoring and repair tools is when a library changes a specification
that breaks proofs in client code that uses those libraries. Current refactoring and repair tools,
however, rely each individual client to determine the appropriate refactors and repairs to make
to fix those proofs. To better address this problem, future refactoring and repair tools
can provide support at the level of library design.
A library designer may, for example, specify how something has changed to a tool;
the tool may then apply this information in client code automatically. 
Some program repair tools already support library-provided patches~\citep{Monperrus2018};
we expect to see this extend to proof refactoring and repair tools in the future.

One barrier to useful refactoring and repair tools for proof engineers is the lack of information
on the kinds of changes that proof engineers make in practice. Collecting data on the changes
that proof engineers make and classifying it could help guide refactoring and repair tools
to handle classes of changes that matter in practice, and could also help machine learning tools
gather both positive and negative examples. Similarly, collecting the benchmarks and examples from
both proactive and reactive approaches to proof evolution such as Planning for Change, seL4, iCoq, and \textsc{PUMPKIN PATCH} 
can help drive the development of future proof evolution tools and measure their success meaningfully.

\section{User Productivity and Cost Estimation}
\label{sec:productivity}

\cite{Bourke12} outline challenges in large-scale verification projects using proof assistants: (1) new proof engineers joining the project, (2) expert proof engineering during main development, (3) proof maintenance, and (4) social and management aspects. They highlight three lessons: (1) proof automation is crucial, (2) using introspective tools for quickly finding facts in large databases gain importantance for productivity, and (3) tools that shorten the edit-check cycle increase productivity, even when sacrificing soundness.

\cite{Zhang2012} present a simulation model of the process of verifying the operating system kernel seL4. Their model is expressed as a software process using the tool Vensim. \cite{Andronick2012} describe the development process and management issues in verifying seL4. They conclude that formal verification, and re-verification, for systems requiring in the order of 10,000 LOC is feasible using a proof assistant. \cite{Staples2013} studied the relationships between sizes of artifacts in seL4. They find that the formal specifications have a significant relationship with the the size of the verified executable code. \cite{Staples2014} study the \emph{proof productivity} problem in the context of seL4; they find that effort is correlated linearly with proof size. \cite{Matichuk2015} analyze the Isabelle/HOL specifications and proof scripts from the seL4 project, and find a quadratic relationship between the size of a formal property and the proof script required to prove it.

\cite{Jeffery2015} identify 30 research questions about productivity in application of formal methods, such as verification using proof assistants. \cite{Klein2015} outline the benefits of trustworthy systems.

\section{Mining and Learning from Proof Repositories}
\label{sec:mining}

Mining software repositories is an emerging field that analyzes software repositories to yield actionable information about software systems and their development and evolution. We describe similar forms of analysis that have been carried out for repositories with proof assistant code.

\cite{Wiedijk2009} compared statistics for standard libraries of several proof assistants for versions available around 2009, including Isabelle/HOL, Coq, and HOL Light. For each library, he reports the number of lines of comments, proofs, definitions, etc. Despite foundational differences, the numbers are similar, with HOL Light having the smallest number of lines for definitions. For example, the LOC shares of theorem statements, definitions, and proof in the Coq version 8.1 standard library were 11\%, 8\%, and 53\%, respectively. Wiedijk argues informally that fewer definitions per proof means higher trustworthiness, since having proofs of relevant properties yield higher confidence in the adequacy of definitions.

\cite{Blanchette2015} investigated Isabelle's Archive of Formal Proofs (AFP), analyzing among other properties the number and sizes of proofs, interdependencies between projects, and number of authors. For the AFP in aggregate, the LOC shares of theorem statements, definitions, and proofs were 19\%, 8\%, and 58\%, respectively. They found that the Isabelle Sledgehammer tool for proof automation~\citep{Blanchette2013} could prove about 60\% of all theorems in the AFP.

Software metrics provide quantitative ways to describe software artifacts and processes and discover new properties. \cite{Aspinall2016} first considered analogous metrics for formal proofs. More specifically, they define an abstract model of formal proofs and a set of proof metrics for this model, which they implement for three different proof assistants (Isabelle, Mizar, and HOL Light) and apply to several large proof corpora.

\cite{Komendantskaya2012} used machine learning with clustering algorithms to identify patterns in large collections of Coq tactic sequences and proof trees, e.g., to find structural similarities between lemmas, and \cite{Heras2013b,Heras2014} highlighted how statistical patterns in proofs can be leveraged during interactive proof development. \cite{Aspinall2016b} used machine learning, in the form of a $k$-nearest-neighbor classifier, to learn and suggest theorem names in HOL Light projects that accurately reflect their property definitions.

\cite{Muller2017} proposed a format and database for capturing and leveraging \emph{alignments} between concepts in different proof assistants, e.g., between natural numbers in Coq on one hand and Isabelle/HOL on the other. One of the basic assumptions in alignment is that concepts have syntactic and semantic similarities across environments, consistent with repetitiveness assumptions in naturalness. \cite{gauthier2017} proposed an algorithm based on heuristics for generating alignments given two proof assistant libraries, and evaluated it on libraries from six proof assistants. For example, by evaluating a library against itself for alignment, duplicated concepts can be found.

\cite{Kaliszyk2017b} leveraged statistical machine learning techniques in a tool that automatically translates (``formalizes'') mathematical texts to proof assistant code. Their approach and evaluation is based on learning and cross-validation using a corpus with established alignments between English texts and HOL Light documents, based on the Flyspeck project~\citep{Hales2017}. They find that the number of correct translations among the top 20 is 64\%.

There are many recent lines of work that learn from large proof assistant corpora to directly perform various automated reasoning tasks~\citep{Kuhlwein2012,Kuhlwein2013,Kaliszyk2014,Irving2016,Loos2017,Peng2017}; the HOLStep dataset~\citep{Kaliszyk2017} is designed as a benchmark for training and evaluating such techniques in a proof assistant context. \cite{Gauthier2017b,Gauthier2018} proposed a technique for learning from HOL4 tactic sequences and proof states and automatically suggest tactic-based proofs of theorems. They achieved around 66\% success rate on the HOL4 standard library, and by also incorporating the automated E prover into the toolchain, they raised the success rate to 69\%. \cite{Huang2018} similarly learn from tactics and proof states, but in the context of Coq and for a limited set of algebraic proof goals. \cite{Yang2019} proposed a more general tactic-based approach for learning and automatic proof suggestion for Coq, which achieved around 12\% success rate on proofs from a large dataset of 123 Coq projects. \cite{Nagashima2018} used custom encodings of proof state in Isabelle/HOL for learning in order to predict suitable proof methods (essentially powerful domain-specific proof tactics) to apply. \cite{Bansal2019} presented a learning environment for HOL Light.

\chapter{Conclusion}
\label{ch:conclusion}


Proof engineering has come far since its infancy in the 1970s~\citep{Milner1972b,Milner1972}, drawing on hundreds of years of foundational ideas in mathematics and logic. Researchers and proof engineers have used proof assistants to build software artifacts spanning hundreds of thousands of lines of code~\citep{Leroy2009,Klein2009}. 
A growing fraction of these artifacts are executable on real hardware, and of these, some are verified down to machine code for verified hardware.
Verified artifacts have shown themselves to be more reliable~\citep{Yang2011}, and are beginning to see 
industrial applications~\citep{CompCert-ERTS-2018, Erbsen2019}.

Compared to most research software, widely used proof assistants such as Coq and Isabelle/HOL are mature and well-maintained tools, with large communities, large software ecosystems, wide selections of support tools, and sophisticated interfaces.
Interest in verification across academia and industry builds additional momentum for proof assistant development.
Interest in formal proofs among mathematicians, for example, results in rich libraries~\citep{UniMath, Bauer2017}, foundational advances~\citep{univalent2013homotopy}, 
and tooling~\citep{Braibant2011} useful for verifying software.
Advances in automated theorem proving reach proof engineers through tools like hammers~\citep{Blanchette2016},
which allow them to benefit from cutting-edge research without increasing the trusted computing base.
Adaptations of research on programming practices and developer support tools and systems from software engineering
help proof engineers continually increase productivity.

We are living in the age of ``big verification''~\citep{popl-time}.
A new generation of computer science students are learning to use proof assistants in undergraduate and graduate courses~\citep{coq-classroom},
entering the workforce equipped with the skills to verify software. As the scale of that verified software continues to grow,
the challenges of proof engineering will continue to grow more 
significant and salient. This survey concludes with a discussion of five of the opportunities to address high-level challenges that remain.

\paragraph{Opportunity 1: Adapting Tools and Ideas from Software Engineering} 
Development processes, build workflows, and support infrastructure for proof assistants are far behind those for traditional software development, and proof engineer productivity is consequently far from its potential. Communities and ecosystems are small compared to those for traditional software.
Many features of and tools for proof assistants are not adequately documented, and proof engineers struggle to develop and maintain verified software in the face
of evolving proof assistants, libraries, and requirements~\citep{Bourke12}.
Domains may lack usable frameworks and libraries to build on, forcing time-consuming development from scratch.
Interfaces may lack in usability and key features, and may become unmaintained and obsolete.
Results in one proof assistant cannot be readily used in another, except in special cases.
Proof engineering is particularly far behind software engineering with respect to maintenance, disincentivizing experts~\citep{Elphinstone2013}
from changing the system once it has been verified.

Proof engineering has already benefited from traditional software engineering, for example through many of
the design principles discussed in Chapter~\ref{ch:organization}. It can continue to draw on tools and ideas from software engineering 
when applicable. The surveyed work suggests that it is sometimes necessary to adapt these tools and ideas,
both to address challenges unique to or especially pronounced in proof engineering (such as brittleness), and to fully take advantage of
the opportunities that proof engineering presents (such as the availability of full specifications).
Continuing to transfer ideas and tools from software engineering to proof engineering with these differences in mind may help close the gap between the two disciplines. Ideas and tools ripe for transfer include improved continuous integration systems, 
package systems, source code hosting, graphical interfaces, error messaging, debugging tools, and development processes.

\paragraph{Opportunity 2: Making Proof Assistants More Accessible}
Compared to traditional programming languages, systems, and environments, proof assistants can be hard for non-experts to understand.
Their foundational bases in logic, mathematics, and type theory, which have served to maintain trust, may also deter potential users due to perceived complexity.
Omissions in and misunderstandings of specifications may lead to lowered expectations and negative perceptions of formally verified software.
An overwhelming majority of large successful software verification projects using proof assistants are carried out and maintained by small teams of highly specialized and trained researchers, frequently with close ties to the institutions where the corresponding proof assistant is developed~\citep{Gonthier2013}. 

Some of this inaccessibility may dissipate over time, as students become more familiar with concepts like
inductive or dependent types through the use of interactive theorem provers in computer science courses,
or through the availability of online books and tutorials.
Some may inspire new abstractions around concepts that are not accessible to the average programmer.
For example, new abstractions may help non-experts more easily interface with unification algorithms and existential variables.
The interactive theorem proving community can continue to draw on automated theorem proving to help make interactive theorem proving more accessible and increase its reach and impact, not just through hammers, but also as part of counterexample generators and similar tools~\citep{Blanchette2010}.
Techniques from the broader formal methods community can even be certified in proof assistants, and then applied inside them. 

\paragraph{Opportunity 3: Understanding and Evaluating Development Processes}
It is difficult to improve the state of the art in proof engineering without understanding the status quo.
The surveyed work suggests that little work has been done to understand and assess the current development processes of proof engineers,
and that some of the work that has been done in this direction has been inconclusive.

Thousands of proof assistant software verification projects are publicly available on platforms such as GitHub for study, reuse, and as research subjects; curated collections such as Isabelle's Archive of Formal Proofs~\citep{isabelleafp} and Coq's OPAM package index~\citep{CoqOPAM} provide high-quality projects with extensive metadata, guaranteed to work with certain proof assistant versions. These codebases can be investigated empirically, learned from, 
and used to construct benchmarks when evaluating new proof engineering techniques. Collecting and analyzing data from other points in the development
process, such as interaction with the IDE, may further facilitate in these processes, as may user studies and community-wide retrospective discussions
of influential work.

\paragraph{Opportunity 4: End-to-End Verification}
Bugs affecting the soundness of a proof assistant with respect to its foundations, although rare, threaten to undermine trustworthiness.
\cite{coq-critical}, for example, contains preliminary documentation of 23 (now fixed) critical bugs in the history of stable releases of Coq. Of the bugs listed,
only 1 (fixed in 2015) was assessed as likely to be exploited by chance, but the risk of others was not determined.
Extensions to proof assistant kernels may further cloud understanding of the implemented metatheory. 

The wide availability of verified systems software provides the basis of a fully verified software ecosystem, with all-encompassing end-to-end guarantees for both functional and non-functional properties for a wide range of practical software. This could rule out many recently discovered critical bugs and even many prevalent security flaws~\citep{Chlipala2018}. Proof assistant self-verification projects and certified compilers for proof assisants, although still emerging and used by few, promise to eventually make the trusted computing base of verified ecosystems minimal. One important step towards a verified ecosystem is to formalize additional practical programming languages and their semantics and runtime environments.

\paragraph{Opportunity 5: Looking to New Applications}
In the Social Processes critique of program verification, \cite{DeMillo1977} wrote that:

\begin{quote}
We believe that, in the end, it is a social process that determines whether mathematicians feel confident about a theorem--and we believe that, because no comparable social process can take place among program verifiers, program verification is bound to fail. We can't see how it's going to be able to affect anyone's confidence about programs.
\end{quote}
The authors of this survey find it surprising with this in mind that \cite{Leroy:POPL06, Leroy2009} put years of effort into verifying
an optimizing C compiler, with little precedent to suggest that such a project could succeed.
But it did succeed, and in so doing increased confidence in both the compiler itself~\citep{Yang2011}
and in the compiled software~\citep{CompCert-ERTS-2018}.

As proof engineers and researchers like \cite{Leroy:POPL06} stretch the boundaries of current proof engineering
techniques on new and interesting domains, proof engineering research can continue to grow to support
the new needs that arise. This cycle can help build a world with, for example, safer transportation systems or more reliable medical devices.
It can help build a world that is a bit more trustworthy.

\begin{acknowledgements}
We would like to thank Dan Grossman, Derek Dreyer, Xavier Leroy, Benjamin Pierce, Andrew Appel, Bob Harper, Jonathan Aldrich, Karl Crary, Adam Chlipala, Chris Martens, Joachim Breitner, Christine Rizkallah, Giuliano Losa, Thomas Tuerk, Roberto Guanciale, Doug Woos, James R. Wilcox, Ryan Doenges, Jared Roesch, Sorin Lerner, Leslie Lamport, Fred Schneider, John Leo, Bob Atkey, Lars Hupel, Buday Gergely, Makarius Wenzel, Jasmin Christian Blanchette, Matthieu Sozeau, Cyril Cohen, David Thrane Christiansen, Sam Tobin-Hochstadt, Mario Alvarez, Joomy Korkut, Robert Rand,
Taylor Blau, Anna Kornfeld Simpson, Jonathan Sterling, Colin Barret,
Toby Murray, Gerwin Klein, Graydon Hoare, Emilio J\'esus Gallego Arias, and Brendan Zabarauskas.
We are grateful to the anonymous reviewers for their detailed comments and valuable suggestions.
This material is based upon work supported by the National Science Foundation Graduate Research Fellowship under Grant No. DGE-1256082, the National Science Foundation under Grant Nos. CCF-1652517, CCF-1836813, and CCF-1749570, and by a Research Award from Facebook.
Any opinions, findings, and conclusions or recommendations expressed in this material are those of the 
authors and do not necessarily reflect the views of the National Science Foundation.
\end{acknowledgements}


\backmatter  
\printbibliography

\end{document}